\documentclass[final,1p,times]{elsarticle}

\usepackage{graphicx,psfrag,color}
\usepackage{mathbbol,amsmath,amsfonts,amssymb,bm,dsfont,amsthm}
\usepackage{wasysym}
\usepackage{hyperref}
\usepackage[mathscr]{euscript}
\usepackage{array,lscape}
\usepackage[utf8]{inputenc}
\usepackage[english]{babel}
\usepackage{braket}
\hypersetup{
    colorlinks,
    citecolor=blue,
    filecolor=blue,
    linkcolor=blue,
    urlcolor=blue
}

\newcommand{\cc}{\mathrm{K}}

\renewcommand{\k}{{\mathbf{k}}}
\renewcommand{\j}{{j}}

\renewcommand{\r}{{r}}

\newcommand{\F}{{\mathbb{F}}}
\newcommand{\R}{{\mathbb{R}}}
\newcommand{\C}{{\mathbb{C}}}
\renewcommand{\H}{{\mathbb{H}}}

\renewcommand{\S}{T}

\newcommand\norm[1]{\lVert#1\rVert}

\renewcommand{\i}{{\rm i}}
\newcommand{\n}{{\rm n}}

\newcommand{\pert}{{{w}}}
\newcommand{\cond}{{\varkappa}}
\newcommand{\ind}{\kappa}
\newcommand{\Ker}{{\rm Ker\ }}

\def\Rds{\mathbb{R}}

\theoremstyle{plain}
\newtheorem{thm}{Theorem}[section]
\newtheorem{prop}[thm]{Proposition}
\newtheorem{coro}[thm]{Corollary}
\newtheorem{lem}[thm]{Lemma}

\newtheorem{defn}[thm]{Definition}

\newtheorem{remark}[thm]{Remark}
\newtheorem{exmp}[thm]{Example}

\journal{Annals of Physics}

\begin{document}

\begin{frontmatter}

\title{\mbox{Wiener-Hopf factorization approach to a bulk-boundary correspondence} 
and stability conditions for topological zero-energy modes}

\author[dartmouth,calgary]{Abhijeet Alase}
\ead{alase.abhijeet@gmail.com}
\author[sunypoly,dartmouth]{Emilio Cobanera}
\author[indiana]{Gerardo Ortiz}
\author[dartmouth]{Lorenza Viola}

\address[dartmouth]{Department of Physics and Astronomy, Dartmouth College, Hanover NH 03755, USA}
\address[calgary]{{Institute for Quantum Science and Technology, and Department of Physics and Astronomy, \\
University of Calgary, Calgary, AB T2N 1N4, Canada}}
\address[sunypoly]{\mbox{Department of Mathematics and Physics, SUNY Polytechnic Institute, Utica NY 13502, USA }}
\address[indiana]{Department of Physics, Indiana University, Bloomington IN 47405, USA}

\vspace{10pt}

\begin{abstract}
Both the physics and applications of fermionic symmetry-protected topological phases rely heavily on a principle known as {\em bulk-boundary correspondence}, which predicts the emergence of protected boundary-localized energy excitations (boundary states) if the bulk is topologically non-trivial. Current theoretical approaches formulate a bulk-boundary correspondence as an equality between a bulk and a boundary topological invariant, where the latter is a property of boundary states. However, such an equality does not offer insight about the stability or the sensitivity of the boundary states to external perturbations. To solve this problem, we adopt a technique known as the Wiener-Hopf factorization of matrix functions. Using this technique, we first provide an elementary proof of the equality of the bulk and the boundary invariants for one-dimensional systems with arbitrary boundary conditions in {\em all} Altland-Zirnbauer symmetry classes. This equality also applies to quasi-one-dimensional systems (e.g., junctions) formed by bulks belonging to the same symmetry class. We then show that only topologically non-trivial Hamiltonians can host {\it stable} zero-energy  edge modes, where stability refers to continuous deformation of zero-energy excitations with external perturbations that preserve the symmetries of the class. By leveraging the Wiener-Hopf factorization, we establish bounds on the sensitivity of such stable zero-energy modes to external perturbations. Our results show that the Wiener-Hopf factorization is a natural tool to investigate bulk-boundary correspondence in quasi-one-dimensional fermionic symmetry-protected topological phases. Our results on the stability and sensitivity of zero modes are especially valuable for applications, including Majorana-based topological quantum computing. 
\end{abstract}


\end{frontmatter}

\newpage

\tableofcontents

%\pagebreak

\section{Introduction} 
The bulk-boundary correspondence is arguably the most striking feature of symmetry-pro- tected topological (SPT) phases of mean-field, ``free" fermionic matter. It is a relationship between the topological features of certain vector bundles of Bloch states over the Brillouin zone~\cite{kz16}
and the emergence of boundary-localized quasiparticle modes in a system with terminations; see e.g. 
Refs.~\cite{Chiu16, Prodan, alldridge} for 
in-depth discussions. These boundary modes are essential for identifying SPT phases of free fermions experimentally. They are also the foundation on which all technological applications of many topological materials are predicated.
To mention some paradigmatic examples,
in the context of incompressible liquids,
the boundary modes of the two-dimensional
electron gas in the quantum Hall regime are responsible for the quantization of the transverse Hall conductance of the system
~\cite{thouless82,halperin82, hatsugai93}.
In the context of superconductors,  
Majorana boundary modes~\cite{Kitaev} are key to some topological quantum computing architectures 
that are under intense investigation~\cite{freedman,MMilestones} even as the debate over putative experimental detection of Majorana modes rages on~\cite{retraction}. 
 
How does the non-trivial topology of Bloch states, a piece of information associated to translationally invariant systems, engenders 
boundary modes of a terminated system?
How is the topological classification 
of SPT phases mirrored by the topologically-dictated boundary modes? 
And, last but not least, how robust are the boundary modes against
deviations from the ideally terminated, clean system?
These questions can be turned into sharp mathematical problems, albeit not easy ones. Nonetheless, there is by now a body of knowledge associated to them that can be considered satisfactory even by the standards of mathematical physics. For example, the rigorous connection between the topological invariant of an SPT phase and its boundary modes,  with and without
disorder, was established using various tools, including Green's function~\cite{Essin11}, 
T-duality~\cite{Mathai15, Mathai16, Hannabuss18a, Hannabuss18b}, K-theory~\cite{Prodan,alldridge,Bourne17} 
and others~\cite{Avila13,Essin15, Cedzich16, Sedlamyr17,Bal17, Shapiro17, Cedzich18, Chen18, Rhim18}. From these and many other works a shared picture has emerged: ``the" bulk-boundary correspondence
is a theorem relating topological invariants of systems without boundaries 
to the boundary modes of bulk-disordered, terminated systems. 
We will maintain this picture and call this 
result the standard bulk-boundary 
correspondence \cite{Chiu16, Prodan, alldridge}. 

Including bulk disorder in a mathematically consistent
framework is extraordinarily taxing and makes the rigorous bulk-boundary correspondence one of the great accomplishments of mathematical physics. In this paper, we prove a bulk-boundary correspondence with emphasis on boundary conditions (BCs) rather than bulk disorder. The surface of an actual piece of matter is never ideal, it can undergo
relaxation, or reconstruction, and/or suffer
from specific BCs~\cite{PRB2}. In addition,
a surface is only partially determined by bulk
properties. Within the mean-field approximation, the range of 
surface phenomena is modeled by BCs that can be fairly independent from bulk conditions \cite{ours2}. Hence, it is justified to regard the questions of the robustness of boundary modes against bulk disorder and against non-ideal BCs as conceptually independent.

Zooming in on one-dimensional (1D) systems, one spatial dimension is special because there is no surface metal. The topological boundary modes are always mid-gap {\em zero-energy} quasiparticle spectral modes (zero modes, ZMs, henceforth). As a consequence, the many-body system is {\em fully gapped} for both open (that is, ideally terminated) and periodic BCs. In this context, the standard bulk-boundary correspondence predicts the emergence of boundary ZMs, but not their number nor whether that number is stable against symmetry-preserving perturbations. This feature is a shortcoming because the number of ZMs determines, and is determined by, the ground degeneracy of the gapped many-body fermion system. This same degeneracy is exploited for storing quantum information in 
topological quantum memories~\cite{freedman,dennis}.
So the question arises, is this many-body degeneracy of topological significance? 
And, can the 
invariance of the number of topological ZMs, if any, be predicted from bulk properties of the system alone? 

Going beyond the actual number of topological ZMs, how sensitive are the ZMs themselves to perturbations? 
Beyond its importance for fundamental physics, the answer to this question could prove important for applications. For example, the sensitivity to symmetry-preserving perturbations could be relevant to assess the way disorder deteriorates the quality of sensing in a recent proposal based on topological edge states~\cite{sarkar2022free}.
As compelling as this question is, it highlights the fact that even finding a good measure of sensitivity of a set of topological ZMs to perturbations is a problem in and of itself. Provided one finds a satisfactory solution, then comes the question of whether this sensitivity can be bounded in terms of bulk properties of the system, and how tightly so. Some of these questions have been recently tackled for two specific models of interest, see Refs.~\cite{boutin18,Wang2021}. 

Within the above motivating context, in this paper we solve the following problems associated to quasi-1D SPT phases of free fermions, classified within the framework of the well-known ``tenfold way'' \cite{Kitaev09,Ryu10}. 
Here, a quasi-1D system refers to 
a) a surface-terminated d-dimensional 
system that is periodic along d-1 dimensions; or 
b) systems with hyperplanar interfaces separating clean bulks, that is, a ``junction''; or 
c) a finite multi-ladder system. 
The systems considered are allowed to have multiple fermionic degrees of 
freedom associated with each lattice point present in the effective model
of the system. In this setting, our main contributions are as follows:

\begin{itemize}
    \item\textit{Spectral flattening to ``dimer'' Hamiltonians.---} We describe a systematic method for constructing a  homotopy (``adiabatic deformation") connecting any quasi-1D Hamiltonian to a canonical ``dimerized" flat-band Bloch Hamiltonian  without leaving the appropriate symmetry class. Our method applies under both periodic and semi-open BCs. Thus, in one spatial dimension, the topological
    features of fermionic SPT 
    phases are captured by short-range correlated systems.
    \item \textit{A bulk-boundary correspondence under open boundary conditions.---} 
    We prove a rigorous bulk-boundary correspondence for clean, quasi-1D systems subject to boundary disorder. 
    It covers all ten symmetry classes.
    We also extend our results from wires to junctions under the assumption that the bulks forming the junction as well as the tunnel Hamiltonian satisfy the same set of symmetries.
    \item \textit{Stability of manifolds of 
    topological ZMs.---} We identify a good quantifier of the stability of manifolds of ZMs and prove that only topologically non-trivial Hamiltonians can host some stable (in our sense) ZMs. Nonetheless, not all topological ZMs need be stable and we include examples. The ZMs of topologically trivial wires, 
    such as certain Andre\'{e}v bound modes, are provably 
    fragile. 
\item \textit{Sensitivity of manifolds
    of ZMs to perturbations.---}
    We introduce a quantitative measure of the 
sensitivity of boundary modes to symmetry-preserving perturbations and derive an upper bound for it in terms of a quantity that depends only on the Bloch Hamiltonian. We conclude that the ZMs of flat-band Hamiltonians are the least sensitive to symmetry-preserving perturbations.
\end{itemize}

Our central mathematical tool for analysis is a 
{\em Wiener-Hopf (WH) factorization} of matrix-valued functions
of a single variable. This tool, to our knowledge, has not been exploited in the context of bulk-boundary correspondences before. 
While we will not dive into technicalities right away, 
we remark that, in its most standard form, the WH factorizations~\cite{gohberg2005convolution, Gohberg03} are of enormous value for solving certain functional equations and linear equations in infinite dimensions. In physics, the WH factorization appears often in connection to  
the inverse scattering problem~\cite{Gohberg03}. 

Although the theory of the standard WH factorization 
has been thoroughly explored, there are only a few results available for the canonical WH form of matrix functions which satisfy unitary or anti-unitary constraints, as the Bloch Hamiltonians do, see Refs.~\cite{Ran94, Rodman14, Shelah17, Lancaster95, 
Voronin11, Guo98, Kravchenko08, Ehrhardt, Iftime01, Youla78}. Thus, in this paper we solve a mathematical problem of independent interest:
\begin{itemize}
    \item \textit{Symmetry-constrained 
    WH factorizations.---} The WH factorization is a factorization of matrix functions of one complex variable. In our context, the matrix functions are the Bloch Hamiltonians analytically continued off the Brillouin zone.  We prove that there exist relatives of this standard WH factorization, we called them {\em symmetric WH (SWH) factorizations}, such that the factors of a Bloch Hamiltonian in some symmetry class are themselves consistent with the classifying symmetries. The outcomes are ten distinct but closely related SWH factorizations. They are calculated systematically by way of an iterative algorithm and constitute the foundation of all the results of this paper.
\end{itemize}

We note that some of the above conclusions were included in a preliminary form in Chapter 4 of 
%AA' PhD thesis of one of us without proof
Ref.\,\cite{Springer}.
The present work is a substantially expanded and improved version of the same chapter. It includes several new results and contains full mathematical proofs that were omitted in the thesis chapter, including the proofs for SWH factorizations and for the bulk-boundary correspondence under arbitrary boundary conditions.

Let us conclude this introduction with some additional remarks about our bulk-boundary correspondence. 
For the symmetry class AIII, and this class only, our result was proved in Ref.~\cite{Prodan} as an intermediate step towards the standard bulk-boundary correspondence. For the other symmetry classes, our bulk-boundary correspondence is implicit in later work of Ref.~\cite{alldridge}, heavily-reliant on K theory: Our results should, in principle, emerge as a corollary of \cite{alldridge} by removing the bulk disorder from the systems. Our contribution provides a proof that requires only basic, very concrete algebra. As a result, one 
obtains a very clear and intuitively appealing picture of how topologically mandated ZMs emerge and interplay with arbitrary BCs. 
The price to pay for bypassing K theory is the restriction to quasi-1D. There is no good generalization of the WH factorization for matrix functions of more than one complex variable. 

The organization of this paper is as follows. In the background section, Sec.\,\ref{sec:background}, we review the tenfold way classification of free-fermion SPT phases and discuss some basic results regarding the standard
matrix WH factorization. In section Sec.\,\ref{sec:symwh}, we provide a SWH factorization for each symmetry class of the tenfold way.
In Sec.\,\ref{sec:spectralflattening}, we explain how SWH factorization leads to spectral flattening of 1D Hamiltonians to ``dimerized'' Hamiltonians.
Section\,\ref{sec:bb} is devoted to a self-contained derivation of our bulk-boundary correspondence for 1D non-interacting SPT phases using the SWH factorizations. We also additionally show how a similar derivation applies to certain junctions. 
In Sec.\,\ref{sec:stability}
we derive conditions for the stability of ZMs. 
Finally, in Sec.\,\ref{sec:sensitivity}, we derive bounds on the sensitivity of ZMs. We summarize our results and discuss their implications in Sec.\,\ref{sec:conclusion}.

\section{Background}
\label{sec:background}

In this section we cover a substantial amount of background material, with the goal of making this paper fairly self-contained and accessible to a broad readership. We start with a general description of systems of free fermions in second quantization (i.e., quadratic in creation and annihilation fermionic operators) and the passage down to a first-quantized-like formalism for calculating quasiparticle energies and wave functions (the ``modes") from a Bogoliubov-de Gennes (BdG) Hamiltonian. It is understood that the quadratic form results from some mean-field approximation, 
and we focus on tight binding models of finite range. Our style of presentation and  conceptual emphasis are influenced by the work of Zirnbauer and collaborators over the last few years, see for example Refs.~\cite{kz16, alldridge}, with twists coming from our own work~\cite{PRL,JPA,PRB1,PRB2,Xu20}. 
Next, we zoom in on algebraic aspects of 1D free-fermion systems with and without boundary. The key observation is that
translation symmetry up to a boundary forces the BdG Hamiltonian to be an element of a block-Toeplitz algebra. We also motivate and introduce at this point our model of boundary conditions and the continuation of the Bloch Hamiltonian into a meromorphic matrix function of one complex variable.
After that, we review the main classification schemes of 
free fermions, namely, the Altland-Zirnbauer classification of random matrix ensembles of disordered free-fermion systems~\cite{Zirnbauer} and the tenfold way classification of ground states of translation invariant free-fermion systems
\cite{Kitaev09,Ryu10,kz16}, as well as the standard bulk-boundary correspondence. The emphasis is on accurate but often qualitative ideas and results that inform our work most closely. We also put on record, for future reference, the topological invariants of the five non-trivial tenfold-way classes in one dimension.  We conclude this background section with a brief introduction to the standard WH factorization. 

\subsection{The many-body models and classifying symmetries}
\label{auxBdG}

\subsubsection{The model Hamiltonians}

As mentioned, we shall focus on systems of fermions on a quasi-1D
lattice with $0<N\leq\infty$ sites along the first (the longest) direction. 
The positions of the lattice points in this direction are labeled by $j \in \{0,\dots,N-1\}$. 
There can be internal degrees of freedom for each $j$ labeled generically by $\{m = 1,\dots,d_{\rm int}\}$. 
Coordinates for all other lattice directions (if present) are absorbed in the internal degrees of freedom.
If spin \(s=1/2\) is one of the internal degrees of freedom and there is the need to show it separately, then we use the standard label  $\sigma \in \{\uparrow,\downarrow\}$.
In addition, sometimes it is necessary to isolate a sublattice degree of freedom; we label it $\nu \in \{1,-1\}$. 
The remaining internal degrees of freedom are denoted generically by 
$p \in \{1,\dots,P\}$. 
The operator $c^\dagger_{jm}$ and its adjoint are the creation and annihilation operators of fermions in the states \((j,m)\). They satisfy  the canonical anticommutation relations
\[
[c^{\;}_{jm},c^{\;}_{j'm'}]_+= 0_F,\quad [c^{\;}_{jm},c_{j'm'}^\dagger]_+=\delta_{jm,j'm'}1_F, 
\]
in terms of the anticommutator \([a,b]_+=ab+ba\) and the zero and identity operators of Fock space.
Within the mean-field and tight-binding approximations~\cite{FW2012}, a many-body Hamiltonian for a system of independent fermions is an operator of the form  
\begin{align*}
\widehat{H}&=\widehat{K}+\widehat{\Delta},\\
\widehat{K}&=\sum_{jm,j'm'} K_{jm,j'm'}c_{jm}^\dagger c_{j'm'},\\
\widehat{\Delta}&=\frac{1}{2}\sum_{jm,j'm'}\left(\Delta_{jm,j'm'}c^\dagger_{jm}c^\dagger_{j'm'}+\Delta^*_{jm,j'm'}c^{\;}_{j'm'}c^{\;}_{jm}\right),
\end{align*}
where $^*$ denotes complex conjugation.
Then, \(\widehat{H}\) is Hermitian provided that
\(
K_{jm,j'm'}=K_{j'm',jm}^*,
\)
and, due to the fermionic statistics, 
\( \Delta_{jm,j'm'}=-\Delta_{j'm',jm}. \)
With hindsight, it is useful to write the kinetic energy operator in a more symmetric form at the expense of a constant shift. 
Due to the canonical anticommutation relations, we can redefine
\begin{align}
\nonumber
\widehat{K}& \mapsto \widehat{K} - \frac{1}{2}{\mbox{tr}}(K)\, 1_F, \ \ \mbox{ where } \mbox{tr}(K)=\sum_{jm}K_{jm,jm}, \mbox{ so that } \\
\label{hatK}
\widehat{K}&=\frac{1}{2}\sum_{jm,j'm'}
\left(K_{jm,j'm'}
c_{jm}^\dagger c^{\;}_{j'm'}
-K_{j'm',jm}
c^{\;}_{jm}c^\dagger_{j'm'}\right).
\end{align}
In this way, we can write
\begin{eqnarray}
\widehat{H}=\widehat{K}+\widehat{\Delta} + \text{constant}.
\end{eqnarray}
As the constant term corresponds to an energy shift, we redefine $\widehat{H}$ to be the traceless Hamiltonian
\begin{eqnarray}
\widehat{H} \mapsto \widehat{H} =\widehat{K}+\widehat{\Delta}.
\label{many-body}
\end{eqnarray}

This paper is concerned with the quasiparticles of the system modeled by the Hamiltonian \(\widehat{H}\). They are fully determined by \(\widehat{H}\) and so it will be the main focus of our investigation. We reserve the name {\em quadratic fermionic Hamiltonian} (QFH) for \(\widehat{H}\). A QFH in the mathematical sense need not be interpreted  as a Hamiltonian only. Certain generators of continuous symmetry transformations are QFHs in the mathematical sense.

\subsubsection{The classifying symmetry operations}
\label{symopsonsite}

In quantum mechanics, by Wigner's theorem a symmetry transformation is represented by a unitary or antiunitary operator on the Hilbert space associated to the pure states of the system. If a symmetry operator commutes with the  Hamiltonian of a system, then it is a symmetry of that system. 
In this paper, we will investigate classes of QFH characterized by symmetries shared by all the QFH in a given class. The symmetry classes are the ones identified by Altland and Zirnbauer in their seminal paper Ref.\,\cite{Zirnbauer}. In the Altland-Zirnbauer (AZ) classification scheme, there are ten classes characterized by some subset of the symmetries of particle number, spin rotation, time reversal, and a sublattice symmetry. The topological classification of QFHs, the tenfold way classification, adds translation symmetry to the set of classifying symmetries, and thus space dimensionality as a meaningful label. 

\medskip

\noindent \underline{Spin rotations.}---
The rotations in spin space are generated by the Hermitian operators
\begin{align}
\label{spingen}
\widehat{S}_x&=\frac{1}{2}\sum_{jm}\left(c^\dagger_{j\uparrow m}c_{j\downarrow m}^{\;}+c^\dagger_{j\downarrow m}c_{j\uparrow m}^{\;}\right),\nonumber\\
\widehat{S}_y&=-\frac{i}{2}\sum_{jm}\left(c^\dagger_{j\uparrow m}c_{j\downarrow m}^{\;}-c^\dagger_{j\downarrow m}c_{j\uparrow m}^{\;}\right),\nonumber\\
\widehat{S}_z&=\frac{1}{2}\sum_{jm}\left(c^\dagger_{j\uparrow m}c_{j\uparrow m}^{\;}-c^\dagger_{j\downarrow m}c_{j\downarrow m}^{\;}\right),
\end{align}
satisfying the usual commutation relations of angular momentum operators. 

\bigskip

\noindent\underline{Time reversal.}--- The time-reversal operation, denoted by $\Theta$, is the anti-unitary symmetry induced by the mapping 
\begin{align*}
\Theta\, i1_F\,\Theta^{-1}&=-i1_F,\\
\Theta\, c_{j\uparrow m}\, \Theta^{-1} &= c_{j\downarrow m},\\
\Theta\, c_{j\downarrow m}\, \Theta^{-1} &= -c_{j\uparrow m},
\end{align*}
and the adjoints of these relations (one can check that \((\Theta A \Theta^{-1})^\dagger=\Theta A^\dagger \Theta^{-1}\) for any antiunitary transformation \(\Theta\)). 
It is immediate to check that 
\(\Theta \widehat{S}_{\alpha}
\Theta^{-1}=-\widehat{S}_{\alpha}\) for \(\alpha=x,y,z\).
Since $\Theta^2 = -1_F$, it affords a representation in Fock space of the group $\mathbb{Z}_4$.
\bigskip

\noindent\underline{Particle number.}---
The Hermitian operator
\begin{align}
\widehat{N} =\sum_{jm} c_{jm}^\dagger c_{jm}^{\;}
\end{align}
is the observable associated to the total number of electrons. It generates a representation of the group \(U(1)\). Regarded as a symmetry, it is called the particle-number symmetry operation. 

\bigskip

\noindent\underline{Sublattice symmetry operation.}--- 
It is an antiunitary transformation, denoted by $\Sigma$,  induced by the relations 
\begin{align}
\Sigma\, i1_F\, \Sigma^{-1}&=-i1_F,\\
\Sigma\, c_{j\nu m}\, \Sigma^{-1} &= (-1)^{\nu} c^\dagger_{j\nu m},\\
\Sigma\, c^\dagger_{j\nu m}\, \Sigma^{-1} &= (-1)^{\nu} c_{j\nu m},
\end{align}
where $\nu$ labels the sublattice, e.g. the twp sublattices in graphene or in Su-Schrieffer-Heeger model~\cite{SSH}, which we revisit later.
It follows from this definition that the sublattice symmetry operation is its own inverse. Hence, it amounts to a representation of the group $\mathbb{Z}_2$. Our specific sublattice symmetry may seem overly restrictive as compared to, say, that of Ref.\,\cite{kz16} but it is not so. We will revisit this point later.

The sublattice symmetry operation properly deserves to be called a particle- hole transformation as it exchanges the creation and annihilation operators but, unfortunately, that name is often associated to a different property of QFHs. 
Note that the above definition of 
the sublattice symmetry rests on the fact that 
the exchange \(c \leftrightarrow c^\dagger\) is 
a implementable by a unitary transformation on the fermionic Fock space. The analogous statement for bosons is false; see Ref.\,\cite{Xu20} for a physical discussion of this point and its ramifications. 

\bigskip

\noindent\underline{Lattice translations.}--- A translation to the left is generated by the unitary transformation induced by the relations
\begin{align*}
{\cal U}^{\;}_L\, c_{0m} \,{\cal U}_L^\dagger =c_{N-1,m},\quad
{\cal U}^{\;}_L\, c_{jm} \,{\cal U}_L^\dagger =c_{j-1,m}, \quad j=1,\dots,N-1. 
\end{align*}
The shifts to the right are generated by \({\cal U}^{\;}_R={\cal U}_L^\dagger\).

\subsection{The Bogoliubov-de Gennes representation}
\label{auxBdG}

The QFHs close a real Lie algebra.  What this statement means is that i) given two QFHs \(\widehat{H}_\i\), \(\i=1,2\), there exists a third QFH \(\widehat{H}_3\) such that
\begin{align}
    \label{commQFHs}
[\widehat{H}_1,\widehat{H}_2]=i\widehat{H}_3 ,
\end{align}
and ii) a finite real linear combination of QFHs is another QFH (the emphasis on reality comes from requiring Hermiticity). This fact is consistent with the characterization of the QFH associated to a free-fermion Hamiltonian/symmetry generator: a commutator is necessarily traceless. 

The following analysis is standard. It goes back to Bogoliubov and became hugely popular after the work of de Gennes. 
Our presentation, however, is somewhat unconventional and inspired by the recent work of Zirnbauer \& co.~\cite{alldridge}. Consider the auxiliary complex vector space 
\[
{\mathcal{H}}_{\rm BdG} \equiv \text{Span\,}_\mathbb{C}\{c^{\;}_{jm},c^\dagger_{jm}\ |\ j=0,\dots,N-1\ \mbox{and}\ 
m  = 1,\dots, d_{\rm int} \},
\] 
comprising linear combinations of creation and annihilation operators. Its dimension is twice the number of single-particle (SP) state labels.  We will denote elements of \({\cal H}_{\rm BdG}\) with lower case Greek letters topped with a small hat; e.g., 
\begin{eqnarray}
\label{etahat}
\hat{\eta}
=\sum_{j=0}^{N-1}\sum_{m=1}^{d_{\rm int}}(\eta_{0jm}c_{jm}^\dagger+\eta_{1jm}c^{\;}_{jm}).
\end{eqnarray}
The fundamental reason for introducing the auxiliary BdG space is the way it interplays with QFHs. The commutator of a QFH with a linear combination of creation and annihilation operators yields another linear combination of creation and annihilation operators. Hence, the commutator with a QFH induces a linear transformation \(H:{\cal H}_{\rm BdG}\rightarrow {\cal H}_{\rm BdG}\) via the formula
\[
H[\widehat{\eta}]\equiv[\widehat{H},\hat{\eta}] .
\]
Moreover, if Eq.\,\eqref{commQFHs} holds, then the associated linear transformations of \({\cal H}_{\rm BdG}\)
satisfy \([H_1,H_2]=iH_3\). One can show this as follows. First,
\[
iH_3[\hat{\eta}]=i[\widehat{H}_3,\hat{\eta}]=
[[\widehat{H}_1,\widehat{H}_2],\hat{\eta}] .
\]
By the Jacobi identity and our definitions, 
\[
[[\widehat{H}_1,\widehat{H}_2],\hat{\eta}]=H_1 H_2[\hat{\eta}]-H_2 H_1[\hat{\eta}]=[H_1,H_2][\hat{\eta}] ,
\]
whereby the claim follows. 

In short, the BdG space carries a representation of the Lie algebra of QFHs. One can show that it is a faithful representation, that is,  that the mapping \(\widehat{H}\mapsto H\) is one-to-one. The linear transformation \(H\) is called the BdG Hamiltonian. The advantage of one version of the underlying Lie algebra over the other is the following. The QFHs act on a Hilbert space of dimension that grows exponentially with the number of SP state labels. By contrast, the dimension of the auxiliary BdG space grows only linearly. 

The auxiliary BdG space inherits structure from the fact that it is generated by fermionic creation and annihilation operators. Specifically:
\begin{enumerate}
    \item \textit{A real structure.---} A real structure on a complex vector space is an antilinear map \(\rho\) that is its own inverse. For the BdG space, the real structure is given by the map \(\rho:{\cal H}_{\rm BdG}\rightarrow {\cal H}_{\rm BdG}\) defined as 
 \[
 \rho[\hat{\eta}]\equiv \hat{\eta}^\dagger.
 \]
In other words, a vector of the auxiliary BdG space is regarded as real if the same object, regarded as an operator of the Fock space, is Hermitian.

    \item \textit{ A Hermitian inner product $\langle \cdot | \cdot \rangle$.---} 
 The basic fact is that the anticommutator of two linear combinations of creation and annihilation operators is a multiple of the identity. The Hermitian inner product of two vectors of the BdG space is calculated by way of the formula 
   \[ \langle \hat{\eta}_1| \hat{\eta}_2 \rangle \mathbb{1} =
[\,\hat{\eta}_1^\dagger, \hat{\eta}^{\;}_2\,]_+.
\]
\end{enumerate}

Let us see how these pieces of structure interact with the basic object of interest, the representation of the Lie algebra of QFHs. As for the real structure, one has that
\[
\rho H \rho[\hat{\eta}]= \rho H [\hat{\eta}^\dagger] = 
 -[\widehat{H},\hat{\eta}]=-H[\hat{\eta}].
\]
It follows that \(iH\) is a purely real matrix in some basis. The condition
\[
\rho H\rho=-H ,
\]
satisfied by any BdG Hamiltonian/symmetry generator, is most often called \textit{the particle-hole symmetry} of a superconductor. As this analysis shows, it has nothing to do with any many-body symmetry in the usual, Wigner sense (in connection to the existence of Majorana ZMs, this observation was first made in \cite{Ortiz-Dukelsky, Ortiz-Cobanera}).  
Following Zirnbauer, we call it ``the Fermi constraint''.  

Next let us consider the inner product. As one would hope, a BdG Hamiltonian is Hermitian with respect to it. To check that this is so, notice that
\[
[\hat{\eta}_1,[\widehat{H},\hat{\eta}_2]]_+=[\widehat{H},[\hat{\eta}_1,\hat{\eta}_2]_+]-[[\widehat{H},\hat{\eta}_1],\hat{\eta}_2]_+.
\]
Hence, by our definitions,
\[
\langle \hat{\eta}_1|H[\hat{\eta}_2]\rangle =\langle H[\hat{\eta}_1]|\hat{\eta}_2\rangle ,
\]
and so the BdG Hamiltonian is Hermitian.

Let us conclude with a comment on symmetry operations. The continuous symmetries of rotations in spin space and particle number are generated by operators that are QFHs in the mathematical sense. Hence, they induce symmetry generators in the BdG space. The many-body Hamiltonian commutes with spin or particle number rotations if and only if the BdG Hamiltonian commutes with the corresponding BdG generators of said symmetries.

Discrete symmetry operations are yet to be discussed, but the idea is always the same. Translations, time reversal, and the sublattice symmetry operation map linear combinations of creation and annihilation operators to other linear combinations. Hence, they induce linear or antilinear  maps of the BdG space according to the formulas
\[
\bm{T}[\hat{\eta}]\equiv 
{\cal U}^{\;}_L\,\hat{\eta}\, {\cal U}_L^\dagger,\quad 
{\cal T}[\hat{\eta}]\equiv\Theta\, \hat{\eta}\,\Theta^{-1},\quad {\cal V}[\hat{\eta}]\equiv\Sigma\,\hat{\eta}\,\Sigma^{-1}.
\]

\subsection{From BdG operators to matrices}
\label{sec:bdgtomatrix}

The problem of associating matrices to the BdG Hamiltonian and other linear or antilinear transformations of the BdG space is, in principle, straightforward but needs to be addressed with some care. Challenges arise from the desire to take advantage of the algebra of arrays of objects. 

Let us begin by defining a row array \(\hat \Psi^\dagger\) of creation and annihilation operators and of order \(1\times 2d_{\rm int}N\). The entries of the array follow the dictionary order, that is, \(\hat \Psi^\dagger_{\tau j m}\)  appears  to the left of \(\hat \Psi^\dagger_{\tau'j'm'}\) if \(\tau < \tau'\), or if \(j<j'\) whenever \(\tau=\tau'\), or if \(m<m'\) whenever \(\tau=\tau'\) and \(j=j'\). The entries of the array are 
\[
\hat \Psi^\dagger_{0jm}=c^\dagger_{jm},\quad \hat \Psi^\dagger_{1jm}=c^{\;}_{jm}.
\]
Now consider Eq.\,\eqref{etahat}. We can define a column array \(\eta\) of order \(2d_{\rm int}N\times 1\) with numerical entries \(\eta_{\tau jm} \in \mathbb{C}\). Having done so, one can take advantage of the algebra of arrays and write 
\(
\hat{\eta}=\hat\Psi^\dagger \eta
\)
without further ado. Because the entries of these arrays are labeled by an ordered collection of indices, one can display them as a layered object. More specifically,
\begin{align}
\label{taulayer}
\eta=\begin{bmatrix}\eta_{\tau = 0}\\\eta_{\tau = 1}\end{bmatrix}, 
\quad
\eta_\tau=\begin{bmatrix}\eta_{\tau\,j=0}\\ \vdots\\ \eta_{\tau\,j= N-1}\end{bmatrix}, \quad
 \eta_{\tau j}=
 \begin{bmatrix} \eta_{\tau j\, m=1}\\
\vdots\\
\eta_{\tau j\, m=d_{\rm int}}
\end{bmatrix},\quad \tau=0,1,\; j=0,\dots,N-1.
\end{align}
The order of indices that gives rise to the layers is interchangeable, and it is in fact helpful to pull the indices of 
interest to the beginning or to the end as we shall see later.

Moving on to the QFHs, let us introduce a column array \(\hat \Psi\) of creation and annihilation operators with entries 
\[
\hat \Psi^{\;}_{0 j m}=c^{\;}_{jm},\quad\hat \Psi_{1jm}=c_{jm}^\dagger.
\]
Then one can check that
\[\widehat{H}=\frac{1}{2}\hat \Psi^\dagger H \hat \Psi ,
\]
in terms of the square matrix \(H\) of order \(2d_{\rm int}N\) with entries \(H_{\tau jm, \tau' j' m'}\). Referring back to the defining  Eq.\,\eqref{hatK} and \eqref{many-body}, one concludes that 
\[
H_{0jm,0j'm'}=K_{jm,j'm'},\quad H_{1jm,0j'm'}=\Delta_{jm,j'm'}, 
\]
and so on. The matrix \(H\) has layers just like the vector \(\eta\).  Suppressing the indices \(jm\), one can  write  
\begin{equation}
\label{BdG}
H = \begin{bmatrix}  K & \Delta \\ -\Delta^* & -K^{\rm T}\end{bmatrix},\quad K^\dagger=K,\quad \Delta^{\rm T}=-\Delta, 
\end{equation}
with ${\rm T}$ representing transposition with respect to the composite index \(jm\). 

At this point we have associated \(\widehat{H}\) to both a linear transformation of the BdG space and a matrix, and we have denoted both with the same letter. The reason is that 
\[
H[\hat{\eta}]=[\widehat{H},\hat{\eta}]=\widehat{H\eta}=\hat \Psi^\dagger H\eta.
\]
In words, the matrix of Eq.\,\eqref{BdG} is precisely the matrix of the linear transformation associated to the commutator. Now, we can drop the creation and annihilation operators altogether and work with numerical arrays only.

Consider the matrix representation of some basic results and symmetries: 

\medskip

\noindent\textit{Hermiticity condition.}--- The BdG matrix,  Eq.\,\eqref{BdG}, is Hermitian in the usual sense because the BdG operator is Hermitian with respect to the inner product of the BdG space.

\medskip

\noindent\textit{The Fermi constraint (``particle-hole symmetry").}--- Let us work out the matrix representation  of the reality condition \(\rho H \rho=-H\). Let \(\tau_x,\tau_y,\tau_z\)  be defined so that (recall Eq.\,\eqref{taulayer})  
\[
\begin{bmatrix}[\tau_x\eta]_{\tau = 0}\\
[\tau_x\eta]_{\tau = 1}\end{bmatrix}
=\begin{bmatrix}\eta_{\tau = 1}\\
\eta_{\tau = 0}\end{bmatrix},\quad
\begin{bmatrix}[\tau_z\eta]_{\tau = 0}\\
[\tau_z\eta]_{\tau = 1}\end{bmatrix}
=\begin{bmatrix}\eta_{\tau = 0}\\
-\eta_{\tau = 1}\end{bmatrix}.
\]
Then, one finds that 
\[
\rho[\hat{\eta}]=\hat{\eta}^\dagger=\widehat{\tau_x\eta^*}=\hat \Psi^\dagger\tau_x\eta^*.
\]
It is convenient to write \(\eta^*=\cc\eta\), that is, \(\cc\) is the operation of complex conjugation extended to the numerical column vector \(\eta\). Hence, at the level of matrices, the reality condition reads
\[
\tau_x\cc H\tau_x\cc =\tau_xH^*\tau_x=-H ,
\]
which is precisely the ``particle-hole symmetry" of a superconductor as described in the literature at large. 

\medskip

\noindent\textit{Particle number symmetry.}---
 It is immediate to check that 
\(\widehat{N}=
\widehat{\tau_z}\). Hence, particle number is conserved by some QFH if and only if its BdG Hamiltonian commutes with \(\tau_z\). This condition is equivalent to \(\Delta=0\). 

\medskip

\noindent\textit{Spin rotation symmetry.}--- Setting $\hbar=1$, the action of spin rotation symmetries in Eq.\,\eqref{spingen}
on the BdG space translate to matrices
\begin{align*}
\widehat{S}_x=\frac{1}{2}\widehat{\tau_z\sigma_x},\quad
\widehat{S}_y=\frac{1}{2}\widehat{\sigma_y},\quad 
\widehat{S}_z=\frac{1}{2}\widehat{\tau_z\sigma_z}. \end{align*}
In particular, a free-fermion system is invariant under spin rotations if and only if its BdG Hamiltonian commutes with \(\tau_z\sigma_x\), \(\sigma_y\), \(\tau_z\sigma_z\). Notice that the \(\tau\) and \(\sigma\) matrices commute as they act non-trivially on different indices. 

\medskip

\noindent\textit{Time-reversal symmetry.}---
For a time-reversal transformation one finds that
\[
{\cal T}[\hat{\eta}]=\Theta \hat{\eta}\Theta^{-1}=\hat\Psi^\dagger i\sigma_y \cc \eta.
\]
Hence, a system of time-reversal-invariant free fermions is associated to a BdG Hamiltonian such that 
\[
i\sigma_y\cc H(-i\sigma_y)\cc =\sigma_yH^*\sigma_y=H.
\]
There is a remarkable mathematical interpretation of this condition. As it turns out, by Frobenius theorem \cite{frobenius}, there are exactly three associative real division algebras: the real numbers \(\mathbb{R}\) themselves, the complex numbers \(\mathbb{C}\), and the quaternion algebra \(\mathbb{H}\). The quaternion algebra can be realized inside the algebra \(M_2(\mathbb{C})\) as the real linear combinations of the matrices
\begin{equation}
\label{quat}
\bm{1}_{\H} = \begin{bmatrix}1 & 0 \\ 0 & 1 \end{bmatrix},\ 
\bm{i} = \begin{bmatrix}0 & i \\ i & 0 \end{bmatrix},\ 
\bm{j} = \begin{bmatrix}0 & 1 \\ -1 & 0 \end{bmatrix},\
\bm{k} = \begin{bmatrix}i & 0 \\ 0 & -i \end{bmatrix}.
\end{equation}
These matrices are also a basis, over the complex numbers, of \(M_2(\mathbb{C})\) itself. Hence, given a matrix \(q=q_0\bm{1}_{\H}+q_1\bm{i}+q_2\bm{j}+q_3\bm{k}\), the question arises whether it is a quaternion, that is, whether the coefficients \(q_0,\dots,q_3\) are real numbers. The answer is the affirmative provided that 
\(
q=-\bm{j}q^*\bm{j}.
\)
Hence, the time-reversal symmetry condition \(H=\sigma_yH^*\sigma_y=-\bm{j}H^*\bm{j}\) means, precisely, that the BdG Hamiltonian can be regarded as a matrix over the quaternions in terms of the quaternionic basis \(\bm{i}=i\sigma_x\), \(\bm{j}=i\sigma_y\), and \(\bm{k}=i\sigma_z\). 

\medskip

\noindent\textit{Sublattice symmetry.}--- Let
\(\nu_z\) be the matrix induced by the relation
\[
\begin{bmatrix}[\nu_z\eta]_{\nu = 0}\\
[\nu_z\eta]_{\nu = 1}\end{bmatrix}=\begin{bmatrix}1& 0\\0&-1\end{bmatrix}\begin{bmatrix}\eta_{\nu = 0}\\
\eta_{\nu = 1}\end{bmatrix}.
\]
Then, for sublattice symmetry, one finds that
\[
{\cal V}[\hat{\eta}]=\Sigma\hat{\eta}\Sigma= \hat \Psi^\dagger \tau_x\nu_z\cc \eta.
\]
Hence, a BdG Hamiltonian is sublattice symmetric if 
\[
H=\tau_x\nu_z\cc H\tau_x\nu_z\cc =\nu_z\tau_xH^*\tau_x\nu_z=-\nu_zH\nu_z
\]

\noindent\textit{Translation symmetry.}--- Left translation $\bm{T}$ 
induces the matrix \(V\) such that
\[
\begin{bmatrix} 
[V\eta]_{j=0}\\
[V\eta]_{j=1}\\
\vdots\\
[V\eta]_{j= N-2}\\
[V\eta]_{j=N-1}
\end{bmatrix}=
\begin{bmatrix} 
\eta_{j=1}\\
\eta_{j=2}\\
\vdots\\
\eta_{j= N-1}\\
\eta_{j=0}
\end{bmatrix} .
\]
It can be described as the generator of left shifts. The right shift is generated by \(V^\dagger\). A BdG Hamiltonian is translation symmetric if it commutes with \(V\).

\subsection{Fundamentals of the Altland-Zirnbauer and tenfold way classifications}

The AZ~\cite{Zirnbauer} and the tenfold way~\cite{Kitaev09, Ryu10} classifications are two distinct but related classifications of systems of independent fermions. In this section we recall the fundamentals of each as they relate to the work in this paper. 

\subsubsection{The Altland-Zirnbauer classification}
\label{AZblocks}

The four on-site symmetries of Sec.\,\ref{symopsonsite}, spin rotation, time reversal, particle number and particle-hole or sublattice symmetry, stand out on physical grounds as first emphasized by Altland and Zirnbauer. On the basis of certain combinations of these symmetry operations, and they alone, Altland and Zirnbauer classified random matrix ensembles of free-fermion systems into ten classes \cite{Zirnbauer}.   Table\,\ref{table:az} shows these classes, 
emphasizing the fact that {\em the classifying symmetries are rooted in standard (unitary or antiunitary) many-body symmetries}; see Ref.\,\cite{alldridge} and the Introduction of Ref.\,\cite{kz16} for a comparison between the AZ and the tenfold way classifications). 
Figure\,\ref{fig:AZclasses} summarizes the relationships between the various AZ classes. 

The symmetry operations  of the AZ classification impart additional structure to the BdG Hamiltonian~\cite{agarwala}. 
We now describe the unitary transformations that bring forth these additional structures. For each symmetry class,
we identify a $d \times d$ block of the BdG Hamiltonian (with $d \le 2d_{\rm int}$ in general), that we label by $A$, which contains all information about the 
Hamiltonian $H$. This block $A$ will be the central object in proving the results of this paper.

\begin{table}[t]
\begin{center}
\begin{tabular}{|c | c| c| c| l|}
\hline
{\bf Cartan label} & {\bf T} & {\bf C} & {\bf S} & {\bf Many-body symmetry group} \\
\hline
A & 0& 0& 0& U(1)\\
\textbf{AIII} & 0& 0& 1&U(1) $\times \mathbb{Z}_2$ \\ 
\hline
\textbf{D} & 0&+&0& $\{1\}$ \\
\textbf{DIII} & -&+&1& $\mathbb{Z}_4$ \\ 
AII & -&0&0& $\ \mathbb{Z}_4 \ltimes$ U(1) \\
\textbf{CII} & -&-&1&$(\mathbb{Z}_4 \ltimes$ U(1)) $\times \mathbb{Z}_2$ \\
C & 0&-&0&SU(2) \\
CI & +&-&1&SU(2) $\times \mathbb{Z}_4$ \\
AI & +&0&0& SU(2) $\times (\mathbb{Z}_4 \ltimes$ U(1)) \\
\textbf{BDI} & +&+&1&SU(2) $\times (\mathbb{Z}_4 \ltimes$ U(1)) $\times \mathbb{Z}_2$ \\
\hline
\end{tabular}
\caption{Cartan labels are associated to many-body symmetries (last column) and related  properties of the BdG Hamiltonian (characterized by the entries in the columns T, C, and S, where presence(lack) of the symmetry is labeled with a $1 (0)$ and, when present, a sign +(-) means that the squared of the symmetry is +1(-1)). The symbol $\{1\}$ denotes the trivial group of one element (physically, there are no symmetries).
The group \(U(1)\) is generated by the particle number operator, \(\mathbb{Z}_2\) by the particle-hole operation, \(\mathbb{Z}_4\) by the time reversal operation, and \(SU(2)\) by the spin operators. The semi-direct product of groups \(\ltimes\) arises because time reversal and particle number do not commute. The classes in bold font are the topologically non-trivial classes in one dimension. }
\label{table:az}
\end{center}
\end{table}

\begin{figure}[t]
\begin{center}
\includegraphics[width = 8cm]{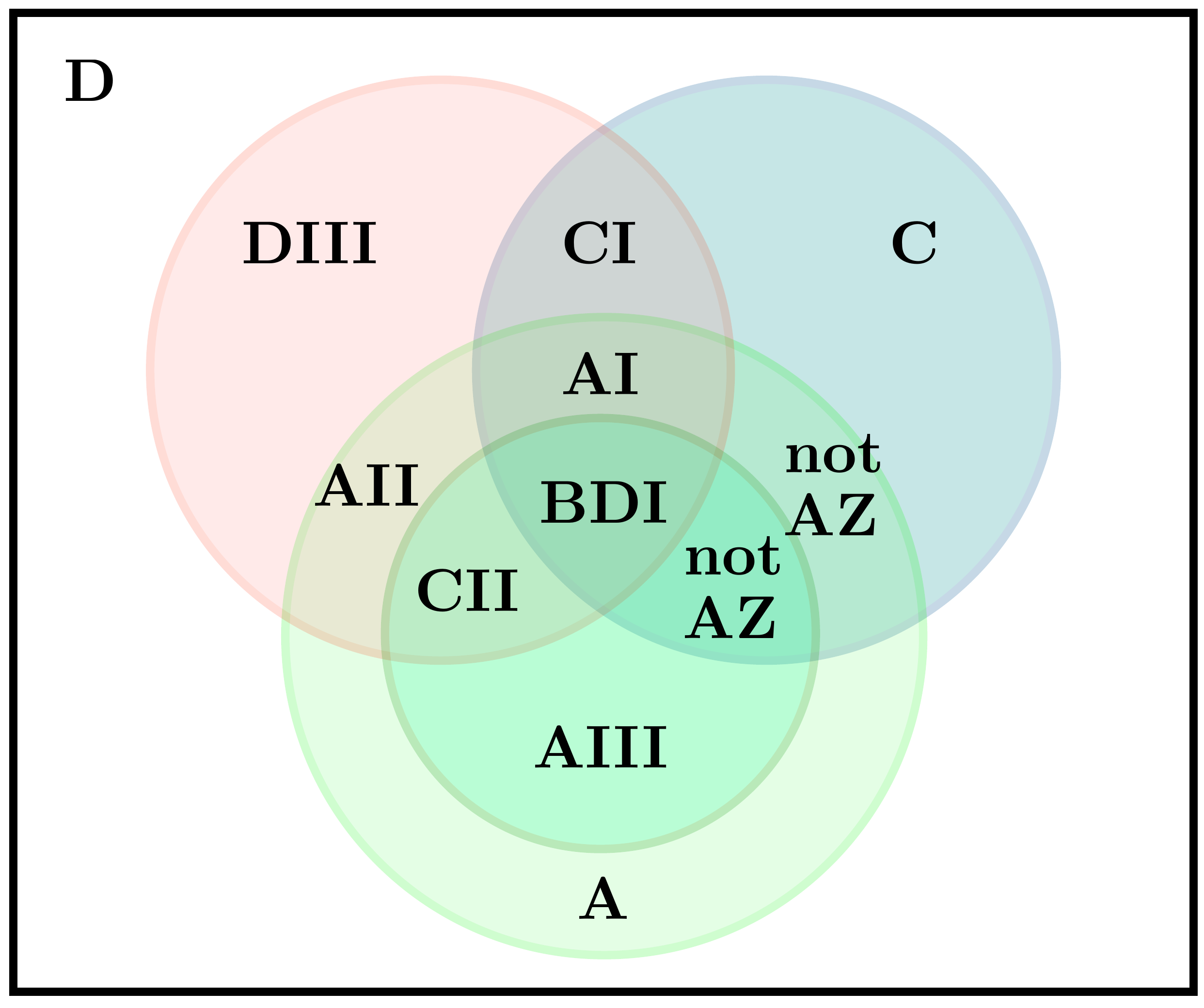}
\caption{Nine out of the ten AZ symmetry classes are properly contained in class D. The three larger circles correspond to time reversal, spin rotation, and particle number symmetry respectively in the clockwise sense. The smaller circle corresponds to the sublattice symmetry. 
The relationships between the classes can be read of this figure. An intersection of two or more AZ classes is a symmetry class with all the symmetries of the classes that meet. It is again an AZ class but for two exceptions: the symmetry class characterized by particle number and spin rotation symmetry and the class characterized by the latter two and, in addition, sublattice symmetry.}
\label{fig:AZclasses}
\end{center}
\end{figure}

\medskip

\underline{Class D},  no symmetries required.--- 
This class is simply the class of BdG Hamiltonians,  that is, matrices that satisfy the Fermi constraint; see Eq.\,\eqref{BdG}.  There is a natural, canonical representation of this class of matrices. The unitary transformation 
\(U^\dagger_{\rm D}\equiv \frac{1}{\sqrt{2}}\begin{bmatrix} 1 & i \\ 1 & -i\end{bmatrix}\), acting on the particle-hole space, that is, as a block transformation, maps a BdG Hamiltonian to its canonical form, that we denote by \(H_{\rm D}\). Explicitly,
\begin{align}
\label{repHD}
H_{\rm D}=U_{\rm D} H U_{\rm D}^\dagger= iA ,
\end{align}
in terms of 
\begin{align*}
A\equiv \begin{bmatrix} K_{\rm Im} +\Delta_{\rm Im} & 
K_{\rm Re} -\Delta_{\rm Re} \\
 -K_{\rm Re} -\Delta_{\rm Re} & K_{\rm Im} -\Delta_{\rm Im}\end{bmatrix}=A^*=-A^T ,
\end{align*}
where the subscripts ${\rm Re}$ and ${\rm Im}$ denote real and imaginary parts. Notice that the matrix \(A\) is a real and antisymmetric matrix of even order. 

The unitary transformation \(U_{\rm D}\) does not depend on the Hamiltonian. Rather, it is a property of the AZ symmetry class, class \(D\) in this case. Its physical meaning is revealed by its action on the creation and annihilation operators. Since
\[
\widehat{H}=\frac{1}{2}\hat{\Psi}^\dagger H\hat{\Psi}=
\frac{i}{2}\hat{\Psi}^\dagger U^\dagger_D A U_{\rm D}\hat{\Psi},
\]
one finds that
\[\begin{bmatrix}
c^\dagger_{jm}&c_{jm}
\end{bmatrix} U_{\rm D}^\dagger
=\begin{bmatrix}
a_{jm} &b_{jm} 
\end{bmatrix} ,\] 
where the \(a_{jm}, b_{jm}\) are Hermitian and satisfy the relations expected of the generators of a complex Clifford algebra of even dimension. They are the Majorana basis popularized by Kitaev \cite{Kitaev}.

\medskip

\underline{Class DIII}, time reversal symmetry.---
This symmetry is equivalent to the pair of constraints
\[
\sigma_yK^*\sigma_y=K,\quad \sigma_y\Delta^*\sigma_y=\Delta
\]
on the blocks of the BdG Hamiltonian. In general, \(K\)
is characterized by three independent blocks in spin space: the complex Hermitian matrices  \(K_{\uparrow\uparrow}\) and \(K_{\downarrow\downarrow}\) 
and the complex matrix 
 \(K_{\downarrow\uparrow}\). 
 Due to time reversal symmetry, \(K\) takes the more constrained form
 \[
 K=
 \begin{bmatrix}
 K_{\uparrow\uparrow} &  K_{\uparrow\downarrow}\\
  -K_{\uparrow\downarrow}^*& K_{\uparrow\uparrow}^*
 \end{bmatrix},\quad K_{\uparrow\uparrow}^\dagger=K_{\uparrow\uparrow},\quad K_{\uparrow\downarrow}^{\rm T}=-K_{\uparrow\downarrow}.
  \]
A similar analysis of the pairing block yields
\[
\Delta =\begin{bmatrix} 
\Delta_{\uparrow\uparrow} & \Delta_{\uparrow\downarrow}\\
-\Delta_{\uparrow\downarrow}^* & \Delta_{\uparrow\uparrow}^*\end{bmatrix},
\quad \Delta_{\uparrow\uparrow}^{\rm T}=-\Delta_{\uparrow\uparrow},\quad \Delta_{\uparrow\downarrow}^\dagger=\Delta_{\uparrow\downarrow}.
\]
The canonical form of a BdG Hamiltonian in this class, call it \(H_{\rm DIII}\), is obtained by way of the unitary transformation 
\[
U_{\rm DIII}^\dagger = \frac{1}{\sqrt{2}}\begin{bmatrix}
0 & -i & -i & 0 \\
1 & 0 & 0 & -1 \\
i & 0 & 0 & i \\
0 & -1 & 1 & 0
\end{bmatrix} ,
\]
acting on the combined Nambu and spin space \({\cal H}_\tau\otimes {\cal H}_\sigma\).
One finds that 
\begin{align}
\label{repDIII}
H_{\rm DIII} = U_{\rm DIII} H U_{\rm DIII}^\dagger=
\begin{bmatrix}
0& A^\dagger\\
A &0\\
\end{bmatrix}, 
\end{align}
in terms of 
\begin{align*}
A \equiv
\begin{bmatrix}
iK_{\uparrow\downarrow} - \Delta_{\uparrow\uparrow}  & K_{\uparrow\uparrow} - i\Delta_{\uparrow\downarrow} \\
 - K_{\uparrow\uparrow}^* + i\Delta_{\uparrow\downarrow}^* &  -iK_{\uparrow\downarrow}^* + \Delta_{\uparrow\uparrow}^* \\
\end{bmatrix}= - A^{\rm T}.
\end{align*}
The matrix \(A\) is a complex and antisymmetric matrix of even order.
\medskip

\underline{Class C}, spin rotation symmetry.--- 
This symmetry condition amounts to requiring that the BdG Hamiltonian should commute with the matrices \(\tau_z\sigma_x\), \(\sigma_y\) and \(\tau_z\sigma_z\). It follows that 
\[K=\begin{bmatrix}
K_{\uparrow\uparrow}& 0\\
0& K_{\uparrow\uparrow}\\
\end{bmatrix}, \quad K_{\uparrow\uparrow}^\dagger=K_{\uparrow\uparrow},
\]
and 
\[
\Delta=\begin{bmatrix}
0& \Delta_{\uparrow\downarrow}\\
-\Delta_{\uparrow\downarrow}& 0\\
\end{bmatrix},
\quad \Delta_{\uparrow\downarrow}^{\rm T}=\Delta_{\uparrow\downarrow}.
\]
The normal form of a BdG Hamiltonians in this class is obtained by way of the unitary transformation
\[
U_{\rm C}=\begin{bmatrix}
1&0&0&0\\
0&0&0&1\\
0&1&0&0\\
0&0&-1&0\\
\end{bmatrix} ,
\]
acting on the combined Nambu and spin space and mixing these two labels. One finds that
\begin{align}
\label{repHC}
H_{\rm C}=U_{\rm C}HU_{\rm C}^\dagger=
\begin{bmatrix}
iA& 0\\
0& iA\\
\end{bmatrix} , 
\end{align}
in terms of 
\begin{align*}
 iA\equiv 
\begin{bmatrix}
K_{\uparrow\uparrow} & \Delta_{\uparrow\downarrow} \\
-\Delta_{\uparrow\downarrow}^* & -K_{\uparrow\uparrow}^*\\
\end{bmatrix}=(iA)^\dagger.
\end{align*}
Now one can check that 
\[
\begin{bmatrix}0& -i\\i&0\end{bmatrix}A^*
\begin{bmatrix}0&-i\\i&0\end{bmatrix}=-A.
\]
Hence, \(A\) is an anti-Hermitian and quaternionic matrix.
\medskip

\underline{Class CI}, spin rotation and time reversal symmetry.--- The additional requirement, as compared to class C, of time reversal
symmetry forces the matrices \(K_{11}\) and \(\Delta_{21}\) above to be real Hermitian and real symmetric respectively.
As a consequence, it becomes possible to put \(H_{11}\) in block off-diagonal form. 
The normal form \(H_{\rm CI}\) of the BdG Hamiltonians in this class
is obtained by way of the unitary transformation
\[
U_{\rm CI}=\frac{1}{\sqrt{2}}
\begin{bmatrix}
1 &-i&0&0\\
-i& 1&0&0\\
0&0&1 &-i\\
0&0&-i &1\\
\end{bmatrix}
U_{\rm C}.
\]
Letting \(H_{\rm CI}=U_{\rm CI}HU_{\rm CI}^\dagger\), one finds that 
\begin{align}
\label{repHCI}
H_{\rm CI}=
\begin{bmatrix}
0& A^\dagger&0&0\\
A&0&0&0\\
0& 0& 0&A^\dagger\\
0&0& A&0\\
\end{bmatrix} ,
\end{align}
in terms of 
\begin{align*}
A\equiv \Delta_{\uparrow\downarrow}-iK_{\uparrow\uparrow}=A^{\rm T}.
\end{align*}
 The matrix \(A\) is a generic complex and symmetric matrix.
\medskip

The remaining six classes all share particle number symmetry. Any unitary symmetry, e.g. spin rotation symmetry, is associated to a block diagonalization of the BdG Hamiltonian by a suitable change of basis. However, particle number is special. The reason is that for particle number conserving free-fermion systems, 
one can diagonalize the many-body Hamiltonian either by diagonalizing the BdG Hamiltonian or its upper left block \(K\) and then constructing the many-body ground state. For particle number symmetric systems, the common practice is to shift the focus from the BdG Hamiltonian to its upper left block. Hence, in the following we will relabel \(K\) as \(H\) and call \(H\) the SP, as opposed to BdG, Hamiltonian. The context should make clear which interpretation of \(H\) is appropriate. 

\underline{Class A}, particle number symmetry.--- 
As explained above, this class is characterized by the SP Hamiltonian  which, as a matrix, is just a generic complex Hermitian matrix, that is,
\begin{align}
\label{repHA}
H_{\rm A} =  H_{\rm A}^\dagger.
\end{align}
For this class, we set $A = H_A$. 
As it should be clear by now, the fact that the BdG/SP 
Hamiltonians are Hermitian matrices is an integral part of the AZ symmetry analysis, every bit as important as the many-body symmetries. The reason we point it explicitly for class A is because it is the only constraint on this class.

\medskip

\underline{Class AI}, particle number, spin rotation, and time reversal symmetry.---
By combining the analysis for classes C and DIII and setting the pairing to vanish (due to particle number symmetry) one concludes that the SP Hamiltonians of this class are block-diagonal. That is, 
\begin{align}
    \label{repHAI}
H_{\rm AI}=\begin{bmatrix}
K_{\uparrow\uparrow}&0\\
0&K_{\uparrow\uparrow}\\
\end{bmatrix}, \quad K_{\uparrow\uparrow}^*=K_{\uparrow\uparrow}=K_{\uparrow\uparrow}^\dagger,
\end{align}
and so the blocks $A = K_{\uparrow\uparrow}$ are generic real Hermitian matrices.
\medskip

\underline{Class AII}, particle number and time reversal symmetry.--- The SP Hamiltonian satisfies the additional, as compared to class A, condition  
\begin{align}
\label{repHAII}    
H_{\rm AII}=\sigma^y{H_{\rm AII}}^*\sigma^y,
\end{align}
because of time reversal symmetry.  Hence, this class consists of quaternionic Hermitian matrices, and we set $A = H_{\rm AII}$.
\medskip

\underline{Class AIII}, particle number and sublattice symmetry.--- The SP Hamiltonian is block off-diagonal, that is,
\begin{align}
\label{repHAIII}
H_{\rm AIII} = \begin{bmatrix} 0 & A^\dagger \\ A & 0\end{bmatrix},
\end{align}
and the blocks are generic complex square matrices.
\medskip

\underline{Class BDI}, 
particle number, spin rotation, time reversal, and sublattice symmetry.---
We are adding sublattice symmetry to AI. Hence,
\(K_{\uparrow\uparrow}\) is a real Hermitian matrix. In addition, it is block off-diagonal due to the sublattice symmetry. In all,
\begin{align}
\label{repHBDI}
H_{\rm BDI}
=\begin{bmatrix}
0& A^\dagger&0&0\\
A& 0&0&0\\
0&0&0&A^\dagger\\
0& 0& A&0\\
\end{bmatrix}, \quad A^*=A,
\end{align}
and so the off-diagonal blocks of the diagonal blocks are generic real matrices. 

\medskip

\underline{Class CII},
particle number, time reversal, and sublattice symmetry.---
The SP Hamiltonian is again block off-diagonal and the blocks are generic quaternionic square matrices due to time reversal. As an equation, 
\begin{align}
\label{repHCII}
H_{\rm CII}=\begin{bmatrix} 0 & A^\dagger \\ A & 0\end{bmatrix},\quad \sigma_yA^*\sigma_y=A.
\end{align}
\medskip

In summary, the BdG or SP Hamiltonians of any of the ten AZ symmetry classes are parameterized by various kinds of special matrices. Table \ref{table:classes} summarizes our work so far. Whether or not a given Hamiltonian enjoys additional symmetries is besides the point. In this respect, the approach we follow here is different from the tenfold way as described in, say, Ref.\,\cite{Ryu10}. Within their framework, the Hamiltonians are classified based on on-site commuting antilinear and anti-commuting symmetries at the SP level.  
If the Hamiltonian under consideration commute with some unitary matrix, then it is block-diagonalized and it is the blocks that are classified. 
In this sense, a time reversal symmetry of a block 
need not be associated to the physical
time reversal operation of spinful fermions. Rather, it can be any antiunitary operation that 
commutes with the  Hamiltonian (block), irrespective of its origin. The same is true about the other symmetries.
In any case, our results rely only on the structure
of the BdG Hamiltonian and therefore  
also hold true within the framework of Ref.\,\cite{Ryu10}.

\begin{table}[t]
\begin{center}
\begin{tabular}{|l|c |  l | }
\hline
{\bf Symmetries} & {\bf Class} & {\bf BdG Hamiltonian}
\\
\hline
none & \textbf{D} & 
$i\times$  real and antisymmetric  of even order\\
\({\cal T}\) & \textbf{DIII} & block off-diagonal, complex antisymmetric \\
& & blocks of even order\\
\hline
\hline
 & & {\bf Diagonal blocks of the BdG Hamiltonian}\\
 \hline
\(SU(2)\) & C &   
$i\times$  quaternionic and anti-Hermitian  \\
\(SU(2)\), \({\cal T}\) & CI & block
off-diagonal, complex symmetric blocks\\
\hline
 \hline
 & & {\bf SP Hamiltonian}\\
 \hline
 \(U(1)\) & A & complex Hermitian\\
 \(U(1)\), \({\cal T}\) & AII & quaternionic Hermitian\\
 \(U(1)\), \({\cal V}\) & \textbf{AIII} & block off-diagonal, complex blocks\\
 \(U(1)\), \({\cal T}\), \({\cal V}\) & \textbf{CII} & block off-diagonal, quaternionic blocks\\
 \hline 
 \hline
 & & {\bf Diagonal blocks of the SP Hamiltonian} \\
 \hline
 \(U(1)\), \(SU(2)\), \({\cal T}\) & AI & real symmetric \\
 \(U(1)\), \(SU(2)\), \({\cal T}\), \({\cal V}\)&
 \textbf{BDI}& block off-diagonal, real blocks\\
 \hline
\end{tabular}
\caption{The symmetries of the AZ classification impart some very specific structure to the BdG Hamiltonian. The unitary symmetries lead to up to two block-diagonalizations. For the particle number symmetry, the upper left block of the BdG Hamiltonian coincides with the SP Hamiltonian. The blocks associated to the spin rotation symmetry do not have a special name. The classes in bold font are the topologically non-trivial classes in one dimension.}
\label{table:classes}
\end{center}
\end{table}

\subsubsection{The tenfold way classification}
\label{fyibulkinvs}

The framework of the tenfold way 
enlarges the set of symmetries of the AZ classification by including translation symmetry. As a consequence, a symmetry class of the tenfold way is labeled by the AZ label and the dimension of space, that is, the number of independent generators of translations. A given, fully-gapped BdG Hamiltonian in a symmetry class of the tenfold way (characterized by an AZ label plus dimensionality) can be placed
in a topological phase by computing the appropriate bulk topological invariant associated to the class (if non-trivial).
A Hamiltonian in a non-trivial symmetry class is said to have  {\em SPT} order if it cannot be connected adiabatically to a trivial, gapped Hamiltonian without explicitly
breaking the symmetries of the class or closing the gap. 
In the following we quote from the literature explicit formulae for the topological (bulk) invariants for the five non-trivial classes in 1D \cite{Chiu16}. The Fermi energy level is assumed to be at zero for all SP Hamiltonians.

The translation-invariant  Hamiltonians of interest are of the form
\begin{eqnarray}
\label{Hamtransinv}
\widehat{H}
= \! \sum_{r,j\in\mathbb{Z}} \!\Big[\hat{\Psi}^\dagger_{j}K_{r}\hat{\Psi}^{\;}_{j+r}
+\frac{1}{2}(\hat{\Psi}^\dagger_{j}\Delta_{r}\hat{\Psi}_{j+r}^\dagger +\text{H.c.})\Big],\quad
\end{eqnarray}
where the integers $j (r)$ represent points (displacements) of the 1D Bravais lattice (isomorphic to) \(\mathbb{Z}\) and $K_r, \Delta_r$ are $d_{\rm int} \times d_{\rm int}$ matrices. We assume throughout
this paper that the summation over \(r\) is finite, that is, we consider only finite range models with $|r| \le R$. In the context of this paper, the infinite Bravais lattice is preferable to periodic BCs because then the crystal momentum is a continuous variable \cite{PRB1}.

For defining the Bloch Hamiltonian and the bulk topological invariants, we need to 
describe a useful tensor factorization of $\mathcal{H}_{\rm BdG}$. In Sec.\,\ref{sec:bdgtomatrix},
we described how operators in $\mathcal{H}_{\rm BdG}$ can be represented as $2Nd_{\rm int}$-dimensional 
complex vectors. In the bra-ket notation, we may adopt the representation
\begin{equation*}
    c_{jm}^\dagger \mapsto \ket{j,\tau = 0, m}, \quad c_{jm} \mapsto \ket{j,\tau = 1,m},
\end{equation*}
where $\tau$ is the Nambu index. In this notation, we can write
\begin{equation*}
    \mathcal{H}_{\rm BdG} = \text{Span}\,\{\ket{j,\tau,m}\}.
\end{equation*}
We may also define the SP Hilbert space to be 
\begin{equation*}
    \mathcal{H}_{\rm SP} = \text{Span}\,\{\ket{j,\tau=0,m}\}.
\end{equation*}
Now it is straightforward to tensor factorize the BdG space as
$\mathcal{H}_{\rm BdG} \cong {\mathcal{H}}_N  \otimes {\mathcal{H}}_{\rm int}$~\cite{PRB1,PRB2} using the mapping
\begin{equation*}
    \ket{j,\tau, m} \mapsto \ket{j}\ket{\tau,m}.
\end{equation*}

Now let 
\[h_r\equiv \begin{bmatrix}
 K_r & \Delta_r\\
 -\Delta_{r}^* & -K_{r}^*
\end{bmatrix}=h_{-r}^\dagger, \quad -R \le r \le R.
\]
Then, we can express the BdG Hamiltonian as~\cite{PRL,PRB1} 
\begin{eqnarray}
\label{boldH}
H=\mathbb{1}\otimes h_0+\sum_{r=1}^R(V^r \otimes h_r+V^{\dagger\, r}\otimes  h_r^\dagger). 
\end{eqnarray}
Here $V$ and $V^\dagger$ denote the left and right shift operators on the infinite lattice,
and hence both are unitary operators.

The simultaneous eigenstates of the translations \(V,V^\dagger\) are 
\[|k\rangle =\sum_{j}e^{ikj}|j\rangle,\quad k\in[-\pi,\pi)\]
For \(|\psi\rangle=|k\rangle|\psi_{\rm int}\rangle\), one finds that 
\(H|\psi\rangle=|k\rangle H_k|\psi_{\rm int}\rangle.\)
The matrix-valued mapping of the unit circle (the Brillouin zone) 
\[H_k=\begin{bmatrix}
K_k & \Delta_k\\
-\Delta_{-k}^* & -K_{-k}^*
\end{bmatrix}
=h_0+\sum_{r=1}^R (e^{ikr}h_r+e^{-ikr}h_r^\dagger)\]
 is the Bloch Hamiltonian of size \(2d_{\rm int} \times 2d_{\rm int}\). 
We can similarly define the symbol $A_k$ of the operator $A = \mathbb{1}\otimes a_0+\sum_{r=1}^R(V^r \otimes a_r+V^{\dagger\, r}\otimes  a_r^\dagger)$,
where $A$ is the block of $H$ as defined in Sec.\,\ref{AZblocks},
and each $a_r \in \mathbb{F}_d$ where 
$\mathbb{F} \in \{\mathbb{R},\mathbb{C},\mathbb{H}\}$ depending on 
the AZ symmetry class.
Note again that the relation of $A$ to $H$ depends on the 
AZ symmetry class of $H$.

For classes AIII and CII, the topological invariant is the winding number of
$\det A_k$, given by
\begin{equation}
\label{invariantchiral}
Q^B(H) \equiv \frac{1}{2\pi i}\int_{k=-\pi}^{\pi} dk \frac{d}{dk}\log \det A_k,
\end{equation}
whereas for class BDI, the invariant is defined to be 
\begin{equation}
\label{invariantchiral2}
Q^B(H) \equiv \frac{1}{\pi i}\int_{k=-\pi}^{\pi} dk \frac{d}{dk}\log \det A_k.
\end{equation}
The factor of $2$ in Eq.\,\eqref{invariantchiral2} accounts for identical spin $\uparrow$ and $\downarrow$ blocks
of the Hamiltonian. Consequently, the invariant for class BDI is always even. We will later see that the bulk invariant 
for class CII is also always even, although this is not reflected
immediately in the formula for the invariant. The even-valuedness of the invariants of classes BDI and CII can be attributed
to the Kramers' degeneracy due to time-reversal symmetry. 
For class D, the bulk invariant is \cite{Kitaev}
\begin{equation}
\label{invariantD}
Q^B(H) \equiv \text{sgn} \left[ \frac{\text{Pf}(A_{k=0})}{\text{Pf}(A_{k=\pi})}\right],
\end{equation}
where $\text{Pf}$ stands for the ``Pfaffian'' \cite{Pfaffian} of an antisymmetric matrix. Likewise, for class DIII, it is \cite{Zhang10}
\begin{equation}
\label{invariantDIII}
Q^B(H) \equiv \left[ \frac{\text{Pf}(A_{k=0})}{\text{Pf}(A_{k=\pi})}
\right]
\exp\left(-\frac{1}{2}\int_{k=0}^{\pi}dk \frac{d}{dk}\log \det A_{k}\right),
\end{equation}
These bulk invariants label the distinct topological phases in each of the five non-trivial AZ symmetry classes in one dimension.

\subsection{The bulk-boundary correspondence}

One of the hallmark predictions associated to the tenfold way classification is the bulk-boundary correspondence. Loosely speaking, it states that a system of independent fermions in an SPT phase should display protected boundary modes. To discuss the bulk-boundary correspondence quantitatively, one needs to 
consider systems terminated on one edge.

\subsubsection{Boundary conditions}
\label{bcs}

Let us consider a half-infinite strip terminated on the left edge.
The sites are labeled by the non-negative integers \(j=0,1,2,\dots,\infty\). 
In this semi-open presentation of the system,
the translation symmetry is broken by the termination only.
We say that the system is subject to \textit{semi-open BCs}. The information about the physical termination of the system can be encoded in a set of shift operators different from the unitary shifts 
\(V, V^\dagger\). 
Let \(T \equiv \sum_{j=0}^\infty |j\rangle\langle j+1|\), so that 
\[
TT^\dagger=\mathbb{1},\quad T^\dagger T=\mathbb{1}-|0\rangle\langle 0|.
\] 
The spectrum of these shift operators consists of the closed unit disk. There are no eigenvectors of \(T^\dagger\). By contrast, the vectors \[
|z\rangle\equiv \sum_{j=0}^\infty z^j|j\rangle,\quad |z|<1,
\]
are eigenvectors of \(T\) since \(T|z\rangle=z|z\rangle.\) 

Next let us consider the Hamiltonian of Eq.\,\eqref{boldH} and replace the invertible shift operators \(V,V^\dagger\)
associated to periodic BCs by the shift operators \(T\) and \(T^\dagger\). Then, the BdG Hamiltonian
\begin{eqnarray}
\widetilde{H}=\mathbb{1}
\otimes h_0+\sum_{r=1}^R(T^r\otimes h_r+T^{\dagger\, r}\otimes  h_r^\dagger)
\end{eqnarray}
models the same system as Eq.\,\eqref{boldH} but with translation symmetry broken by the termination. In this context, \(j=0\) labels the first site of the strip which appears to its left. The BCs are open (also, ``hard-wall") or, more precisely, semi-open since the strip continues to infinity to the right. 

In order to model surface relaxation or reconstruction and/or surface disorder within the mean-field approximation, we allow for BCs other than semi open. Mathematically, it amounts to adding a modification to the BdG Hamiltonian \(\widetilde{H}\) so that 
it becomes
\begin{equation}\label{HBC}
H_{\rm tot}
=\widetilde{H}+W.
\end{equation} 
Bulk disorder can be also be modeled in terms of an additive modification. What makes $W$ a BC is the requirement that it should be a {\em compact  operator} that descends from a QFH \(\widehat{W}\).  
A compact operator can always be described as the limit of a sequence, convergent in the operator norm, of operators acting non-trivially on a finite number of sites only. 
Hence, if \(W\) acts non-trivially away from the edge, it does so in a quickly decreasing fashion as a function of distance from the boundary.
Practically, this assumptions is valid if the length $L$ of the 1D system under consideration is much longer than the length-scale $l$ of $W$ (Fig.\,\ref{fig:boundarydisorder}). 
Note that $l$ is not required to be smaller than or comparable to the penetration depth $\lambda$ of the ZMs (which we will discuss in the next section), or the range $R$ of hopping and pairing that determines the boundary size.  
Under these conditions, the corresponding SP Hamiltonians \({H}_{\rm tot}=
\widetilde{H}+W\) are self-adjoint elements of {\it the block-Toeplitz algebra}, defined as the algebra generated by the shifts \(T, T^\dagger\) with matrix coefficients and closed
in the operator norm. 

\begin{figure}[t]
\begin{center}
\includegraphics[width = 8cm]{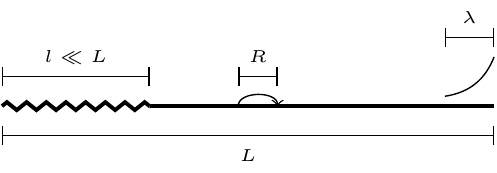}
\caption{The comparison of the length scales of the system, boundary disorder, the ZMs and the range of hopping and pairing. The thick line denotes 
the 1D lattice system under consideration and the zigzag line denotes the 
disorder. The curved arrow shows the farthest hopping on the 
lattice and the curve on the right end denotes the wave function of
the ZM with the largest penetration depth.}\label{fig:boundarydisorder}
\end{center}
\end{figure}

\subsubsection{Boundary invariants in one dimension}

In quantitative terms, the bulk-boundary correspondence states that the bulk topological invariants of the tenfold way classification (see
for example Sec.\,\ref{fyibulkinvs}) match certain ``boundary invariants" associated to terminated systems. 
Taking inspiration from Ref.\,\cite{Prodan}, 
we call an integer-valued quantity a \textit{boundary invariant}{} if it depends only on the eigenvectors of the Hamiltonian associated to  localized zero-energy states
and is provably invariant under compact perturbations obeying the on-site symmetries of the appropriate tenfold way class.
This definition is consistent with the definition of Ref.\,\cite{Prodan} in the case of class AIII in 1D, and also with the definitions in Ref.\,\cite{alldridge}.
Based on the results in the literature, we identify  
\[
Q^\partial(\widetilde{H}+W) \equiv \mathcal{N}_{\nu=0} - \mathcal{N}_{\nu=1} 
\]
to be the boundary invariant for classes AIII, BDI and CII,
where $\mathcal{N}_\pm$ refers to the number of 
zero energy edge states with chirality (eigenvalue of chiral operator) $\pm 1$. 
For classes D and DIII, the boundary invariants are instead identified to be
\[
Q^\partial(\widetilde{H}+W) \equiv (-1)^{\mathcal{N}} \quad \text{and}
\quad Q^\partial(\widetilde{H}+W) \equiv (-1)^{\mathcal{N}/2} ,
\] 
respectively, where $\mathcal{N} = \dim \Ker (\widetilde{H}+W) $ is the number of zero-energy edge states. 
We later use these explicit formulae for boundary invariants to
prove their equality with the bulk invariants
in each of the five non-trivial AZ symmetry classes in 1D.

\subsection{A functional calculus for banded block-Toeplitz operators}
\label{whbasics}

A {\em banded block-Toeplitz} (BBT) operator is an operator of the form
\[ \widetilde{A}(T,T^\dagger)\equiv \sum_{r=p}^q T^\r \otimes a_\r , \] 
with the understanding that 
\(T^{r}=(T^\dagger)^{|r|}\) if \(r<0\) and \(a_r\in\mathbb{C}_d\),
the algebra of complex \(d\times d\) matrices. The lighter notation \(\widetilde{A}\) will be favored when possible. As we explained in Sec.\,\ref{bcs}, the BdG Hamiltonian of a clean system of finite range and subject to semi-open BCs is an example of a BBT operator. The fundamental challenge posed by these systems, as compared to translation-symmetric systems, stems from the fact that it is not possible to diagonalize the shift operators simultaneously or even individually. Nonetheless, there is a powerful functional calculus for BBT operators based on the fact that the shift operators are the one-sided inverse of each other.

\subsubsection{The symbol}
\label{func}
A \textit{ matrix Laurent polynomial}
is a function on the Riemann sphere \(z\in \mathbb{C}\cup\{\infty\}\) of the form  $A(z,z^{-1}) = \sum_{r=p}^q z^\r a_\r$,
with \(p\leq q\) finite integers and coefficients \(a_r\in\mathbb{C}_d\). On the Riemann sphere, see Fig.\,\ref{graphicWH}, the variables \(z,z^{-1}\) are equivalent, meaning that the exchange \(z\leftrightarrow z^{-1}\) induces an automorphism of the algebra of matrix Laurent polynomials. Notice that 
\(z=\infty, 0\) precisely when \(z^{-1}=0,\infty\).
The notation \(A(z,z^{-1})\) is slightly redundant but useful. 

The mapping 
\[
\widetilde{A}(T,T^\dagger)\equiv \sum_{r=p}^q T^\r \otimes a_\r\mapsto s(\widetilde{A})(z,z^{-1})\equiv \sum_{r=p}^q a_rz^r
\]
of BBT operators to matrix Laurent polynomials
identifies the two sets. The restriction of the of \(s(\widetilde{A})\) to the unit circle, effectively a matrix Fourier sum, is the \textit{symbol} of the BBT operator. We will abuse the language slightly and call \(A(z,z^{-1})=s(\widetilde{A})(z,z^{-1})\) the symbol of \(\widetilde{A}\) even if we think of it as a function on the Riemann sphere. There is no danger in doing so because we focus on banded BBT operators with polynomial symbols.

 For \(A(z,z^{-1})=\sum_{r=p}^q a_rz^r\) as usual, let
\begin{align*}
A^\dagger(z,z^{-1})\equiv\sum_{r=p}^qz^{-r}a_r^\dagger,\quad
A^{\rm T}(z,z^{-1})\equiv\sum_{r=p}^qz^{-r}a_r^{\rm T}.
\end{align*}
These operations are defined so that 
\[
s(\widetilde{A}^\dagger)=s(\widetilde{A})^\dagger,\quad 
s(\widetilde{A}^{\rm T})=s(\widetilde{A})^{\rm T}.
\]
The combination
\[
A^*(z,z^{-1})=A^{\dagger\,{\rm T}}(z,z^{-1})=A^{{\rm T}
\, \dagger}(z,z^{-1})\equiv \sum_{r=p}^qz^{r}a_r^*
\]
is important as well. The relation \(s(\widetilde{A}^*)=s(\widetilde{A})^* \) reflects the fact that the matrix elements of the shift operators
in the canonical basis \(|j\rangle|m\rangle\) are real. Notice that, in general, \(A^*(z,z^{-1})\neq A(z,z^{-1})^*\) where, on the right-hand side, the star denotes the complex conjugation of numbers. 

The symbol map is linear but not fully multiplicative. One should expect so, because the Laurent polynomials close an algebra, but the BBT operators do not. The product of two BBT operators is not, in general, another BBT operator because \(\mathbb{1}- TT^\dagger=|0\rangle\langle 0|\). Nonetheless, 
the symbol is multiplicative in a restricted sense. 
Let \(\widetilde{P}(T)=\sum_{r=0}^q T^r\otimes p_r\) and
\(\widetilde{Q}(T^\dagger)=\sum_{r=p}^0 T^r\otimes q_r\)
denote upper and lower triangular 
BBT operators respectively. Then the product
\[
\widetilde{P}(T)\widetilde{A}(T,T^\dagger)\widetilde{Q}(T^\dagger)=\widetilde{B}(T,T^\dagger)
\]
is again a BBT operator and, moreover,
\[P(z)A(z,z^{-1})Q(z^{-1})=B(z,z^{-1}).\]
That is, the symbol of the product, for these kind of products, is equal to the product of the symbols.

\subsubsection{The standard Wiener-Hopf factorization of matrix functions 
and block-Toeplitz operators}

Is there a particularly useful factorization of 
a matrix Laurent polynomial consistent with the functional calculus of Toeplitz operators induced by the symbol map? Not in general. However, for a subclass of symbols, there is indeed one that stands out: the Wiener-Hopf factorization.
There is an associated factorization of the corresponding subclass of BBT operators (the Fredholm BBT operators, as it turns out),
which we also refer to as the WH factorization of BBT operators~\cite{GohbergBook2, Gohberg03}.

\begin{figure}
\begin{center}
\includegraphics[width=.7\columnwidth]{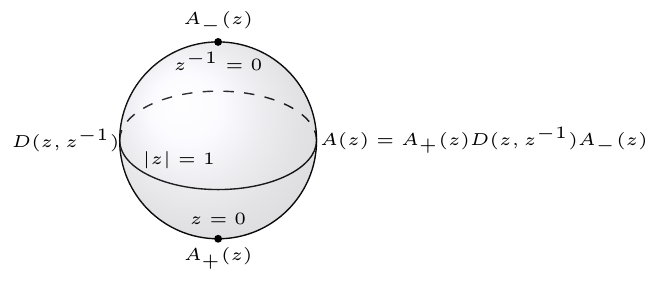}
\end{center}
\vspace*{-5mm}
\caption{The Brillouin zone can be embedded in the Riemann sphere, and the analytic continuation of 
the Bloch Hamiltonian of a fully gapped system defines a matrix Laurent polynomial \(A(z)\) invertible on the unit circle. 
The side factors \(A_+(z), A_-(z)\) in the WH factorization of \(A(z)\) must be invertible in the closed disks (or, equivalently, half spheres) bounded by the Brillouin zone and including \(z=0, z^{-1}=0\) respectively.  }
\label{graphicWH}
\end{figure}

\begin{defn}[Wiener-Hopf factorization]
\label{def:wh}
Let $A(z,z^{-1})$ denote a matrix Laurent polynomial.
A WH factorization of \(A\) is a factorization  of the form
\begin{equation}
\label{wh}
A(z,z^{-1}) = A_+(z) D(z,z^{-1}) A_-(z^{-1}) ,
\end{equation}
which obeys the following additional properties:
\begin{enumerate}
\item
$A_+(z) = \sum_{r= 0}^q  a_{+,r} {z}^\r$ is  invertible for all $|z|\le 1$. 
\item
$A_-(z^{-1}) = \sum_{r=p}^0  a_{-,r} {z}^{-|\r|}$ is invertible for all $|z^{-1}|\leq 1$. 
\item
$D$, the normal form of \(A\), is a diagonal matrix function of the form
$D(z,z^{-1})= \sum_{m=1}^{d}z^{\ind_m}|m\rangle\langle m|$. The set of integers
$\{\ind_m\, |\,  m=1,\dots,d\}$ is ordered so that $\ind_1 \ge \ind_2 \ge \dots \ge \ind_d$.
\end{enumerate}
\end{defn}
\noindent
Figure\,\ref{graphicWH} summarizes some key aspects of the WH factorization. Here, we use the letter $A$ to denote the matrix
Laurent polynomial being factorized. We had previously used the symbol $A$ to denote the blocks of the BdG/SP Hamiltonian 
in a suitable basis depending on its symmetry class. This use is deliberate, as our later analysis will rest on the WH factorization of
the block $A$ of the Hamiltonian. Note that WH factorization restricted to
matrix Laurent polynomials is sometimes referred to as
Birkhoff factorization after its inventor.
We choose to use the name WH factorization 
to make better connection with many associated results in 
the modern literature.

The following well-known proposition~\cite{GohbergBook2, Gohberg03} 
provides hints to why gapped nature of the Hamiltonian is central to our WH factorization-based 
approach to bulk-boundary correspondence. 
\begin{prop}
\label{proof:wh}
The Laurent polynomial \(A(z,z^{-1})\) admits a WH factorization if and only if it is invertible on the unit circle, that is, provided that \(\det A(z,z^{-1})\neq 0\) for all \(z\) with \(|z|=1\).
\end{prop}
\noindent Clearly, the Hamiltonian $H$ is gapped if and only if $A$ is invertible, and only then can it be subjected to WH factorization.

While the normal form $D$ of \(A\) is unique, the WH factorization itself is {\it not} unique in general because of the side factors. For future reference, we quote here without proof a proposition
that characterizes the degree of non-uniqueness of the WH factorization. For a proof, see Ref.\,\cite{gohberg2005convolution}.

\begin{prop}[Uniqueness of the partial indices \cite{Gohberg03}]
\label{lem:uniquepi}
If $A \in \C_d[z,z^{-1}]$ admits two WH factorizations $A=A_+DA_- = 
A_+'D'A_-'$, then $D = D'$, and
$A_+ = BA_+'$, where $B\in\C_d[z]$
is a unit polynomial matrix, such that
\begin{enumerate}
\item If $\ind_m<\ind_{m'}$, then $B_{mm'} = 0$.
\item If $\ind_m=\ind_{m'}$, then $B_{mm'}$ is a constant.
\item If $\ind_m>\ind_{m'}$, then $B_{mm'}$ is a polynomial in 
$z$ of degree $\ind_{m}-\ind_{m'}$.
\end{enumerate} 
\end{prop}

The ordered set $\{\ind_1\geq \ind_2\geq \dots \geq \ind_d\}$ of integers is the 
set of \textit{partial indices}  of \(A\). It induces an additional important bit of 
structure on the set of labels \({\cal M}\equiv \{1,\dots,d\}\), and, by extension, 
the ordered basis \(\{|m\rangle\}_{m\in{\cal M}}\). That is, it induces the partition  into three disjoint subsets
\[
\mathcal{M}_+ = \{m \,|\, \ind_m > 0\},\quad
\mathcal{M}_0 = \{m \,|\, \ind_m = 0\},\quad  \mathcal{M}_- = \{m \,|\, \ind_m < 0\}, 
\]
so that ${\cal M} = \mathcal{M}_+ \cup \mathcal{M}_- \cup \mathcal{M}_0$. The partial indices of \(A\) and the associated partition of labels will feature repeatedly throughout the rest of the paper.

The WH factorization induces, by way of the symbol map, a factorization 
\[
\widetilde{A}(T,T^\dagger)=\widetilde{A}_+(T)\widetilde{D}(T,T^\dagger)\widetilde{A}_{-}(T^\dagger)\]
of certain BBT operators. (We will shortly clarify which ones.) 
We refer to this factorization as
the {\it WH factorization of the BBT operator}. 

The WH factorization requires strict analytical conditions of the factors and exists only for matrix Laurent polynomials that are invertible on the unit circle. What do these stipulations imply for the factors of the WH factorization and the BBT operator that is being factorized? Let us start with the side factors.
The matrix function $A_+(z)$ ($A_-(z^{-1})$) is invertible inside (outside) the unit circle.
Therefore, its inverse has an absolutely convergent Fourier series inside (outside) the unit circle. 
It follows that $\widetilde{A}_+$ ($\widetilde{A}_-$) is an invertible BBT operator because 
the spectrum of $\S$ ($\S^\dagger$) is contained in the unit disk. The inverses  \(\widetilde{A}_+^{-1}\) and \(\widetilde{A}_-^{-1}\) are block-Toeplitz operators but not banded in general. They are, however, upper and lower triangular respectively.

The fact that the side factors are invertible  operators implies that \(\widetilde{D}\) is equivalent to \(\widetilde{A}\) up to a choice of bases.
Our immediate goal is to determine the kernel and cokernel of \(\widetilde{D}\). For powers of the shift operators they can be deduced from the formulas
\[
T^{\r}|j\rangle = \left\{ 
\begin{array}{lcl}|\j-\r\rangle & \text{if} & \j\ge \r\\
0 & \text{if} & \j<\r \end{array} \right.,\quad
\langle \j' |(T^{\dagger})^\r = \left\{ 
\begin{array}{lcl}\langle \j'-\r| & \text{if} & \j'\ge \r\\
0 & \text{if} & \j'<\r \end{array} \right.,
\]
and so the kernel and cokernel of $\widetilde{D}$ are
\begin{eqnarray*}
\text{Ker}\,\widetilde{D}&=& \text{Span}\,\{|\j\rangle|m\rangle, \
m \in \mathcal{M}_+,\ \j=0,\dots,\ind_m-1\},\\
\text{Coker}\,\widetilde{D} &=& \text{Span}\,\{\langle\j|\langle m|, \
m \in \mathcal{M}_- ,\ \j=0,\dots,-\ind_m-1  \}.
\end{eqnarray*}
Finally, 
the kernel and cokernel of $\widetilde{A}$ are 
\begin{equation}
\label{kernel}
\text{Ker}\,\widetilde{A} = \widetilde{A}_-^{-1}\text{Ker}\,\widetilde{D},\quad
\text{Coker}\,\widetilde{A} = (\text{Coker}\,\widetilde{D})\widetilde{A}_+^{-1}.
\end{equation}
The dimension of the kernel and the cokernel of $\widetilde{A}$,
also called its ``defect numbers'', are determined by the partial indices as 
\begin{eqnarray*}
\text{dim}\,\text{Ker}\,\widetilde{A} = 
\sum_{m \in \mathcal{M}_+}\ind_m,\qquad 
\text{dim}\,\text{Coker}\,\widetilde{A} =  
-\sum_{m \in \mathcal{M}_-}\ind_m.
\end{eqnarray*}
In particular, they are finite-dimensional.

\begin{exmp} 
\label{exmp:wh}
If \(A(z,z^{-1})=\begin{bmatrix} \frac{1}{2} & z\\ 1 & \frac{1}{2}\end{bmatrix}\), then the algorithm
of \ref{computeWHF} yields the WH factorization 
\[
A_+(z)=\begin{bmatrix} 1 &\frac{1}{2}\\ 0 &1\end{bmatrix},\quad
D(z)=\begin{bmatrix} z & 0\\ 0 & 1\end{bmatrix}, \quad 
A_-(z^{-1})=\begin{bmatrix} 0 & 1-\frac{1}{4}z^{-1}\\ 1 & \frac{1}{2}\end{bmatrix} .
\]
The matrix Laurent polynomial above is the symbol of the BBT operator
\[
\widetilde{A}=1\otimes \begin{bmatrix}\frac{1}{2} & 0\\ 1 &\frac{1}{2}\end{bmatrix}+T\otimes\begin{bmatrix} 0& 1\\0 &0 \end{bmatrix}.
\]
Its WH factorization is 
\[
\widetilde{A}_+=1\otimes \begin{bmatrix} 1& \frac{1}{2}\\ 0& 1\end{bmatrix},\]
\[\widetilde{D}=1\otimes\begin{bmatrix} 0& 0\\0&1\end{bmatrix}+T\otimes \begin{bmatrix} 1& 0\\0&0\end{bmatrix}, \quad
\widetilde{A}_-=1\otimes\begin{bmatrix} 0& 1\\1 & \frac{1}{2}\end{bmatrix}-\frac{1}{4}T^\dagger\otimes 
\begin{bmatrix}0& 1\\0 & 0\end{bmatrix} .
\]
The partial indices of \(\widetilde{A}\) are \(\ind_1=1, \ind_2=0\), whwreas 
its kernel  is spanned by $\widetilde{A}_-^{-1}|j=0\rangle|m=1\rangle$. To compute 
$\widetilde{A}_-^{-1}$, we first compute its symbol as
\[
A_-^{-1}(z^{-1}) = -(1-z^{-1}/4)^{-1}\begin{bmatrix}\frac{1}{2} & \frac{z^{-1}}{4}-1\\-1 & 0\end{bmatrix} = 
\sum_{j=0}^{\infty}\frac{z^{-j}}{4^j}\begin{bmatrix}-\frac{1}{2} & -\frac{z^{-1}}{4}+1\\1 & 0\end{bmatrix}.
\]
To obtain $\widetilde{A}_-^{-1}$, it suffices to replace $z^{-1} \mapsto T^\dagger$. Finally, the kernel of \(\widetilde{A}\) is spanned by
\[
\widetilde{A}_-^{-1}|j=0\rangle|m=1\rangle = \sum_{j=0}^{\infty}\frac{1}{4^j}|j\rangle\begin{bmatrix}-\frac{1}{2} \\1 \end{bmatrix}.
\]
\end{exmp}

Operators with a finite-dimensional kernel and cokernel
are called \textit{Fredholm operators}. 
They are special because they are almost, but not quite, invertible. The WH factorization shows that if the symbol \(A\) of  \(\widetilde{A}\) is
invertible on the unit circle, then \(\widetilde{A}\) is a Fredholm operator. As it turns out, the converse is true as well, see \ref{app:whproof}. 
The remarkable proposition ``\(\widetilde{A}\) is a Fredholm operator if and only if its symbol \(A\) is invertible" is a special case of a more general theorem that holds for the full Banach algebra (closed in the operator norm) of block-Toeplitz operators.

The topological structure of the set of Fredholm operators
(as a subset of the normed space of bounded operators) is well understood. It consists of a countable infinity of connected components labelled by the integers. In any one component, the integer label coincides with the {\em Fredholm index} (or ``analytic index'' \cite{Atiyah97}) of any of the operators in it. In other words, the Fredholm inded is an integer-valued continuous function, and 
is computed as the dimension of the kernel minus the dimension of the cokernel of a Fredholm operator.
For our BBT operators, 
\begin{equation}
\label{index1}
\text{index}(\widetilde{A}) \equiv \text{dim}\,\text{Ker}\,\widetilde{A}-\text{dim}\,\text{Coker}\,\widetilde{A}
= \sum_{m=1}^{d}\ind_m.
\end{equation}
A Fredholm operator need not be a block-Toeplitz operator, let alone a banded one. Nonetheless, we see that every connected component of the set of Fredholm operators contains  representative BBT operators. 

The Fredholm index is a topological, more precisely, homotopy invariant. It means that a continuous deformation of some Fredholm operator into another one can only change the Fredholm index if the operator looses its Fredholm status somewhere along the way. In turn, for BBT operators, this means that the induced deformation of the symbol stops being invertible  somewhere along the way. One can think of it as a kind of generalized second-order phase transition. Another homotopy invariant
is the ``secondary index'' of Ref.\,  \cite{Atiyah97} (see  \cite{Schulz-Baldes15} for the physical meaning of this index). It is specifically associated to the antisymmetric Fredholm operators
and it can only take two values conventionally chosen to be \(0\) or \(1\).
For an antisymmetric BBT operator, it can be calculated in terms of the partial indices by the formula
\[
\text{index}_2(\widetilde{A}) \equiv \text{dim}\,\text{Ker}\,\widetilde{A}\; \text{mod}\,2
= \sum_{m \in \mathcal{M}_+}\ind_m\, \text{mod}\,2.
\]

\subsubsection{A factorization of block-Laurent operators}
The Hamiltonian of a quasi-1D and translation-invariant system, see Eq.\,\eqref{boldH},
is an example of a block-Laurent operator. A banded block-Laurent operator is of the form 
\(
A(V,V^\dagger)=\sum_{r=p}^q V^r\otimes a_r.
\)
Note that \(V^{r}=V^{\dagger\, |r|}\)
if \(r\) is negative.
These operators are bounded and their spectral properties can be investigated in great detail with the help of the Fourier transform. The Fourier transform puts a block-Laurent operator into a block diagonal form with finite-dimensional blocks
\(
A(e^{ik},e^{-ik})=\sum_{r=p}^q e^{ikr} a_r
\)
labeled by the real number \(k\in[-\pi,\pi)\). There is then little need for any other factorization as a rule. Nonetheless, the WH factorization of the associated matrix Laurent polynomial \(A(z,z^{-1})\), provided it is invertible on the unit circle, does induce a factorization of the block-Laurent operator of the form 
\(A(V,V^\dagger)
=A_+(V)D(V,V^\dagger)A_-(V^\dagger)\) 
or, equivalently,
\(A(e^{ik},e^{-ik})=
A_+(e^{ik})D(e^{ik},e^{-ik})A_-(e^{-ik})\)
for the ``Bloch Hamiltonian."
The so-factorizable block-Laurent operators are necessarily invertible and so the partial indices do not have the same meaning for them as they do for block-Toeplitz operators. Rather, as we will show, there are formulas for calculating the bulk invariants of topologically non-trivial Hamiltonians solely in terms of the partial indices. The proposition \ref{prop:boundgap}
is another  very different and quite remarkable application of the WH factorization of a block-Laurent operator.

\section{Symmetric Wiener-Hopf factorization of Hamiltonians}
\label{sec:symwh}

Let us consider the problem of combining the framework of the AZ classification with the WH factorization of matrix functions and BBT operators. The WH factorization of a BBT Hamiltonian exists provided that its symbol is invertible on the unit circle. In physical language, the unit circle is the Brillouin zone of a 1D system,  the symbol  is the  Bloch Hamiltonian, and the invertibility condition implies that the quasiparticle energy spectrum is gapped around zero energy. Hence, the WH factorization exists for gapped, clean free-fermion systems and we have ended up squarely in the territory of the tenfold way classification and the bulk-boundary correspondence. 
In one dimension, the bulk-boundary correspondence relates topological invariants of the Bloch Hamiltonian to boundary invariants determined by the ZMs (kernel vectors) of the BBT Hamiltonian. Hence, it seems natural that the WH factorization, essentially a tool for determining the kernel and cokernel of a Fredholm operator, should yield insight into the bulk-boundary correspondence. This is indeed the case, as we show in this paper;  however, there is a twist: to succeed, one must modify the standard WH factorization to account for the symmetries of the AZ classification. 

\begin{table}
\begin{center}
\begin{tabular}{|c | l | c | c|}
\hline
{\bf Class} &
\(\mathbb{F}\) &  {\bf Additional} & {\bf Factorization}  \\
& &  {\bf constraint} & \\
\hline
BDI, AIII, CII, &\(\mathbb{R}\), \( \mathbb{C}\), or \(\mathbb{H}\) & none & Theorem \ref{thm:symwh1} \\
A, AI, AII &\(\R\), \(\C\), or \(\H\) & \(A^\dagger=A\) & Theorem \ref{thm:symwh2} \\
D, DIII &\(\mathbb{R}\) or \(\mathbb{C}\) &  \(A^{\rm T}=-A\) & Theorem \ref{thm:symwh3} \\
C & \(\mathbb{H}\) & \(A^\dagger=-A\) & Theorem \ref{thm:symwh5} \\
CI &\(\mathbb{C}\) & \(A^{\rm T}=A\) & Theorem \ref{thm:symwh4} \\
 \hline
\end{tabular}
\caption{Relationship between the AZ symmetry class and the corresponding
SWH factorization of the symbol of BBT operators in the class. }
\label{table:classes2}
\end{center}
\end{table}

In this section we will need a number of variations of the standard WH factorization. We collect them here as theorems of independent mathematical interest. To the best of our knowledge, they are not available in the literature and so we also provide complete proofs in the  \ref{app:symwhproofs}. There are published results similar in spirit, see for example,  \cite{Ran94, Rodman14, Shelah17, Lancaster95,Voronin11, Guo98,Kravchenko08, Ehrhardt, Iftime01, Youla78}. However, the only published results that are directly useful for our purposes is the factorization of Hermitian matrix polynomials of Ref.\,\cite{Ran94}, along with the standard WH factorization itself. 

In what follows, we will denote the set of square arrays of order \(d\) with entries in \(\mathbb{F}\) as  \(\mathbb{F}_d\) (a shorthand for the more standard notation \(M_d(\F)\)), and the associated algebra of matrix Laurent polynomials as \(\mathbb{F}_{d}[z,z^{-1}]\). The notation \(\mathbb{F}\) will stand for one of the three associative real
division algebra, that is,  \(\mathbb{F}=\mathbb{R}\), \(  \mathbb{C}\), or \(\mathbb{H}\). As mentioned before, 
$\mathbb{R}_d[z,z^{-1}] \subset \mathbb{C}_d[z,z^{-1}]$ and
$\mathbb{H}_d[z,z^{-1}] \subset \mathbb{C}_{2d}[z,z^{-1}]$ if we 
consider the $2\times 2$ complex representation of quaternions.
In this sense, any $A \in \mathbb{F}_d[z,z^{-1}]$ admits a standard WH
factorization if $A$ is invertible on the unit circle. 
For such an $A$, we define the set of {\it symmetric partial indices} 
$\mathcal{K} = \{\ind_1, \dots, \ind_d\}$ to be the same as the set of 
standard partial indices of $A$ 
if \(\mathbb{F}=\mathbb{R}\) or 
\(\mathbb{C}\). If \(\mathbb{F}=\mathbb{H}\), the standard partial indices of a quaternionic matrix Laurent polynomial come in pairs of duplicates.
Therefore, \(\mathbb{F}=\mathbb{H}\), set of standard partial
indices consists of two copies of the set of symmetric partial indices 
$\mathcal{K}$. We denote both the standard and 
symmetric partial indices by symbols $\{\ind_m\}$, as it will be clear from the
context which set of partial indices we are referring to. However, we reserve
the symbol $\mathcal{K}$ for the set of symmetric partial indices.
The partition \(\{1,\dots,d\}=\mathcal{M}_+\cup\mathcal{M}_0\cup\mathcal{M}_-\) associated to \({\cal K}\) is defined as before
for the standard WH factorization, according to the signs of the partial indices.

%%%%%%%%%%%%%%%%%%%%%%%%%%%%%%%%%%%%%%%%%%%%%%%%%%%%

 \subsection{Classes AIII, BDI, and CII}
 
We already have all the tools we need for factorizing a Hamiltonian in class AIII. A BBT Hamiltonian in this class can be put in off-diagonal form so that 
 \[
 \widetilde{H}_{\rm AIII}=\begin{bmatrix} 0& \widetilde{A}^\dagger\\
 \widetilde{A}& 0\\
 \end{bmatrix}.
 \]
The blocks are generic complex BBT operators (see Table\,\ref{table:classes2}). Technically, this means that they belong to the tensor product of the Toeplitz algebra and the algebra of complex matrices. Hence, the blocks admit a standard WH factorization,  
\(\widetilde{A}=\widetilde{A}_+\widetilde{D}\widetilde{A}_-\) with no loss of information. Hence, 
\begin{align}
\label{WHAIII}
\widetilde{H}_{\rm AIII}=\begin{bmatrix} \widetilde{A}_-^\dagger & 0\\
0& \widetilde{A}_+\\
\end{bmatrix} 
\begin{bmatrix} 0 & \widetilde{D}^\dagger \\
\widetilde{D}& 0\\
\end{bmatrix} 
\begin{bmatrix} \widetilde{A}_- & 0\\
0& \widetilde{A}_+^\dagger\\
\end{bmatrix} 
\equiv \widetilde{H}_+\widetilde{H}_F\widetilde{H}_+^\dagger.
 \end{align}
We call this factorization the WH factorization of a BBT Hamiltonian in class AIII. %There is an associated factorization of the Bloch Hamiltonian by way of the WH factorization of the symbol of \(\widetilde{H}_{\rm AIII}\).  
The side factors \(\widetilde{H}_+\) and \(\widetilde{H}_+^\dagger\) are invertible BBT operators,
upper and lower triangular respectively.  The middle factor \(\widetilde{H}_F\) is itself a Hamiltonian in Class AIII. Since
\(
\widetilde{D}=\sum_{r=1}^m T^{\kappa_r}\otimes |r\rangle\langle r|
\),
\(\widetilde{H}_F\) is arguably the simplest possible form that a Hamiltonian can take in this class.  The subscript \(F\) refers to that fact. 

The following theorem is the appropriate tool for extending this analysis to the classes BDI and CII. We will state it in the language of the symbol, with the understanding that it translates to WH factorizations via the functional calculus of BBT operators.

\begin{thm} [The  Wiener-Hopf factorization  over \(\R\), \(\C\), or \(\H\)]
\label{thm:symwh1}
Let $A \in\F_d[z,z^{-1}]$ denote a matrix Laurent polynomial over \(\F=\R\), \(\C\) or \(\H\). If \(A(z,z^{-1})\) is invertible on the unit circle, then there exists a factorization 
$A = A_+DA_-$ such that 
\(A_+ \in \F_d[z]\) is invertible for \(|z|\leq1\), \(
 A_- \in \F_d[z^{-1}]\) is invertible for \(|z^{-1}|\leq 1\), and
 \(D = \sum_{m=1}^{d}z^{\ind_m}|m\rangle\langle m|\bm{1}_\F
\)
for $\{\ind_1,\dots,\ind_d\} = \mathcal{K}$.  
\end{thm}

Let us consider next the class BDI. A BBT Hamiltonian in this class takes the form 
\begin{align*}
\widetilde{H}_{\rm BDI}=
\begin{bmatrix}
0& \widetilde{A}^\dagger&0&0\\
\widetilde{A}& 0&0&0\\
0&0&0&\widetilde{A}^\dagger\\
0& 0& \widetilde{A}&0\\
\end{bmatrix}, \quad \widetilde{A}^*=\widetilde{A}.
\end{align*}
The blocks are generic real BBT operators.
Hence, according to the above theorem and the functional calculus of BBT operators, it follows that 
\[
\widetilde{A}=
\widetilde{A}_+\widetilde{D} \widetilde{A}_- ,
\]
in terms of real BBT side factors; the structure of the middle factor is unchanged with respect to class AIII.  The associated factorization of the Hamiltonian is 
\begin{align}
\label{WHBDI}
&\widetilde{H}_{\rm BDI}=
\widetilde{H}_+\widetilde{H}_F\widetilde{H}_+^{\rm T} , 
\end{align}
in terms of 
\begin{align*}
\widetilde{H}_+\equiv 
\begin{bmatrix} \widetilde{A}_-^{\rm T} & 0 & 0 & 0\\
0& \widetilde{A}_+ & 0 & 0\\
0&0 &\widetilde{A}_-^{\rm T}& 0\\
0& 0& 0& \widetilde{A}_+ \\
\end{bmatrix},\quad
\widetilde{H}_F\equiv
\begin{bmatrix}
0& \widetilde{D}^{\rm T}&0&0\\
\widetilde{D}& 0&0&0\\
0&0&0&\widetilde{D}^{\rm T}\\
0& 0& \widetilde{D}&0\\
\end{bmatrix}.
\end{align*}
The diagonal factor is such that  \(\widetilde{D}^\dagger=\widetilde{D}^{\rm T}\). The choice is one of emphasis when it comes to the classes AIII and BDI. Notice that \(\widetilde{H}_F\) itself belongs to BDI.

Finally, a BBT Hamiltonian in class \(CII\) is of the form
\[
\widetilde{H}_{\rm CII}=
\begin{bmatrix}
0&\widetilde{A}^\dagger\\
\widetilde{A}& 0\\
\end{bmatrix},\quad \sigma_y\widetilde{A}\sigma_y=\widetilde{A}^*.
\]
That is, the blocks are generic quaternionic BBT matrices. The associated factorization is 
\begin{align}
\label{WHCII}
\widetilde{H}_{\rm CII}=\begin{bmatrix} \widetilde{A}_-^\dagger & 0\\
0& \widetilde{A}_+\\
\end{bmatrix} 
\begin{bmatrix} 0 & \widetilde{D}^\dagger \\
\widetilde{D}& 0\\
\end{bmatrix} 
\begin{bmatrix} \widetilde{A}_- & 0\\
0& \widetilde{A}_+^\dagger\\
\end{bmatrix}
\equiv \widetilde{H}_+\widetilde{H}_F\widetilde{H}_+^\dagger ,
 \end{align}
in terms of quaternionic factors. Explicitly,
\(
\sigma_y\widetilde{A}_{\pm}\sigma_y=\widetilde{A}_{\pm}^*\)
and \( \sigma_y\widetilde{D}\sigma_y=\widetilde{D}*
\). It follows that \(\widetilde{H}_F\) itself belongs to the class CII. Visually this factorization is identical to that of the class AIII, but there is a crucial difference: for the class CII, the partial indices all come in identical pairs. It is the manifestation in this context of the phenomenon of Kramers degeneracy.

\subsection{Classes A, AI, and AII}

These classes are characterized by SP Hamiltonians that are generic complex, real, or quaternionic Hermitian matrices. The following theorem is the correct tool for factorizing them. We first provide the following:

\begin{defn}
For a Hermitian matrix Laurent polynomial over 
\(\R\) or \(\C\), its signature \(s\) is the difference between
the number of positive and the negative eigenvalues of $A(z,z^{-1})$
for \(z=1\). For a  Hermitian matrix Laurent polynomial over \(\H\),
the signature is  half of this quantity.
\end{defn}
In fact, the signature is constant over the unit circle \cite{Ran94}. In the following, \(\bar{m}\equiv d+1-m\).

\begin{thm} [A factorization of Hermitian matrix Laurent polynomials] \label{thm:symwh2}
Let the matrix $A=A^\dagger \in\F_d[z,z^{-1}]$ be invertible on the unit circle for \(\F=\R\), \(\C\), or \(\H\).
Then there exists a factorization 
$A = A_+ D_1 
A_+^\dagger$ such that $A_+ \in \F_d[z]$
is invertible for \(|z|\leq 1\)
and
\begin{align*}
D_1
=\sum_{m \in \mathcal{M}_+}(z^{\ind_m} |\bar{m}\rangle\langle{m} | 
+z^{-\ind_m}|m\rangle\langle \bar{m}|)\bm{1}_\F+ \sum_{m \in \mathcal{M}_0}s_m|m\rangle\langle m|\bm{1}_\F
\end{align*}
in terms of $\{\ind_1,\dots,\ind_d\} = \mathcal{K}$ and a set of signs $\{s_m=\pm 1\}_{m\in{\cal M}_0}$ such that  $s_m>s_{m'}$ for $m>m'$. The signs add up to the signature of \(A\), that is, \(\sum_{m\in\mathcal{M}_0}s_m = s$. 
\end{thm}

As we apply this theorem to Hamiltonians, we will rename the side factor \(A_{+}\) to \(H_+\) and the middle factor \(D_1\) to \(H_F\). Explicitly,
\[
\widetilde{H}_F =\widetilde{D}_1=\sum_{m \in \mathcal{M}_+}(T^{\ind_m} \otimes|\bar{m}\rangle\langle{m} |\otimes \bm{1}_\F 
+h.c.)+ \sum_{m \in \mathcal{M}_0}s_m \widetilde{I} \otimes |m\rangle\langle m|\otimes \bm{1}_\F ,
\]
and \(\widetilde{H}_+\) is an invertible, upper triangular BBT. Then, 
\begin{align}
\label{WHA}
\widetilde{H}_{\rm A}&=\widetilde{H}_+\widetilde{H}_F\widetilde{H}_+^\dagger,\\
\label{WHAI}
\widetilde{H}_{\rm AI}&=\widetilde{H}_+\widetilde{H}_F\widetilde{H}_+^\dagger,\quad \ \ \ \ \ \ \ \widetilde{H}_S^*=\widetilde{H}_S,\quad S=F, +,\\
\label{WHAII}
\widetilde{H}_{\rm AII}&=\widetilde{H}_+\widetilde{H}_F\widetilde{H}_+^\dagger,\quad \sigma_y\widetilde{H}_S^*\sigma_y=\widetilde{H}_S,\quad S=F, +. 
\end{align}
The key point is that \(H_F\) belongs to its respective class in all cases.

\subsection{Classes D and DIII}

These classes are parametrized by real and complex antisymmetric matrices, respectively. The following theorem is the right tool for factorizing Hamiltonians in these classes.

\begin{thm}[A factorization of antisymmetric matrix Laurent polynomials]
\label{thm:symwh3}
Let 
 $A=-A^{\rm T}\in\mathbb{F}_d[z,z^{-1}]$ be invertible on the unit circle for \(\F=\R\) or \(\F=\C\). There exists a
factorization $A = A_+ D_3 
A_+^{\rm T}$ such that $A_+ \in \F_d[z]$
is invertible for \(|z|\leq1\) and
\begin{align*}
&D_3
=-\sum_{m\in \mathcal{M}_+} \big(z^{\ind_m} |\bar{m}\rangle\langle{m} |\;  
- z^{-\ind_m}|m\rangle\langle{\bar{m}} |\big) 
 +  \sum_{\substack{m\in \mathcal{M}_0\\[1pt] m\le d/2}} \big( |\bar{m}\rangle\langle{m} |\;  
- |m\rangle\langle{\bar{m}} |\big) ,
\end{align*}
in terms of $\{\ind_1,\dots,\ind_d\} = \mathcal{K}$.
\end{thm}

Reasoning as before and letting \(H_F\) stand for \(iD_3\), one finds that 
\begin{align}
\label{WHDIII}
\widetilde{H}_{D}=\widetilde{H}_+\widetilde{H}_F\widetilde{H}_+^T ,
\end{align}
in terms of 
\[\widetilde{H}_F^*=\widetilde{H}_F^T=-\widetilde{H}_F, 
\quad \widetilde{H}_+^*=\widetilde{H}_+.
\]
In particular, \(\widetilde{H}_F\) belongs to class D as well. For the class DIII,
\begin{align}
    \label{WHDIII}
    \widetilde{H}_{\rm DIII}=
    \begin{bmatrix}
    \widetilde{A}_-^\dagger & 0\\
    0& \widetilde{A}_+\\
    \end{bmatrix}
    \begin{bmatrix}
    0& \widetilde{D}_3^\dagger\\
    \widetilde{D}_3& 0\\
    \end{bmatrix}
    \begin{bmatrix}
    \widetilde{A}_- & 0\\
    0& \widetilde{A}_+^\dagger\\
    \end{bmatrix}\equiv \widetilde{H}_+\widetilde{H}_F\widetilde{H}_+^\dagger
\end{align}
with \(H_F\) also in the class DIII.

\subsection{Class C}

For the following factorization theorem, notice that \(\bm{k}^\dagger=-\bm{k}\) in the standard matrix realization of the quaternion algebra, see Eq.\,\eqref{quat}.

\begin{thm} [A factorization of quaternionic anti-Hermitian matrix Laurent polynomials] 
\label{thm:symwh4}
Let $A=-A^\dagger \in\H_d[z,z^{-1}]$ be invertible on the unit circle.
There exists a
factorization $A = A_+ D_4
A_+^\dagger$ such that $A_+ \in \H_d[z]$
is invertible for \(|z|\leq 1\) and
\begin{align*}
D_4 
= -\sum_{m\in \mathcal{M}_+} \big(z^{\ind_m} |\bar{m}\rangle\langle{m} |\;  
- z^{-\ind_m}|m\rangle\langle{\bar{m}} |\big)\bm{1}_\H 
+  \sum_{m\in \mathcal{M}_0} \big( |m\rangle\langle m |\big)
\bm{k} ,
\end{align*}
in terms of $\{\ind_1,\dots,\ind_d\} = \mathcal{K}$.
\end{thm}

As a consequence, the factorization of a BBT Hamiltonian in class \(C\) is \begin{align}
    \label{WHC}
    \widetilde{H}_{\rm C}=
    \begin{bmatrix}
    \widetilde{A}_+& 0\\
    0& \widetilde{A}_+\\
    \end{bmatrix}
    \begin{bmatrix}
    i\widetilde{D}_4& 0\\
    0&i\widetilde{D}_4\\
    \end{bmatrix}
    \begin{bmatrix}
    \widetilde{A}_+^\dagger& 0\\
    0& \widetilde{A}_+^\dagger\\
    \end{bmatrix}\equiv H_+H_FH_+^\dagger ,
\end{align}
with \(H_F\) also in the class C.

\subsection{Class CI}

A BBT Hamiltonian in this class takes the form
\begin{align*}
\widetilde{H}_{\rm CI}=
\begin{bmatrix}
0& \widetilde{A}^\dagger&0&0\\
\widetilde{A}&0&0&0\\
0& 0& 0&\widetilde{A}^\dagger\\
0&0& \widetilde{A}&0\\
\end{bmatrix},\quad \widetilde{A}^T=\widetilde{A}.
\end{align*}
The blocks are generic complex and symmetric BBT operators.  

\begin{thm} [A factorization of complex symmetric matrix Laurent polynomials ] 
\label{thm:symwh5}
Let $A=A^{\rm T} \in\C_d[z,z^{-1}]$ be invertible on the unit circle. There exists a
factorization $A = A_+ D_2 
A_+^{\rm T}$ such that $A_+ \in \mathbb{C}_d[z]$
is invertible for \(|z|\leq 1\)
and 
\begin{align*}
D_2 
= 
\sum_{m \in \mathcal{M}_+}\big(z^{\ind_m}|\bar{m}\rangle\langle{m} |  
+ z^{-\ind_m}|m\rangle\langle{\bar{m}} |\big) 
+ \sum_{m \in \mathcal{M}_0}|m\rangle\langle m| ,
\end{align*}
in terms of $\{\ind_1,\dots,\ind_d\} = \mathcal{K}$. 
\end{thm}

It follows that 
\begin{align}
    \label{WHCI}
    \widetilde{H}_{\rm CI}=\widetilde{H}_+\widetilde{H}_F\widetilde{H}_+^T ,
\end{align}
in terms of
\begin{align*}
    \widetilde{H}_{+}=
\begin{bmatrix}
\widetilde{A}_+^*& 0&0&0\\
0&\widetilde{A}_+&0&0\\
0& 0& \widetilde{A}_+^*&0\\
0&0&0& \widetilde{A}_+\\
\end{bmatrix}, \quad
\widetilde{H}_F=
\begin{bmatrix}
0& \widetilde{D}_2^*&0&0\\
\widetilde{D}_2&0&0&0\\
0& 0& 0&\widetilde{D}_2^*\\
0&0& \widetilde{D}_2&0\\
\end{bmatrix} .
\end{align*}
As usual, \(\widetilde{H}_F\) also belongs to the class CI.

\subsection{Remarks on the factorization results}
\label{sec:wh}

The proofs of the above factorization theorems are collected in the appendices, while Table~\ref{table:classes2} provides an overall summary. In the table one can see the pattern at stake: There are three associative real division algebras and two matrix involutions, transposition and hermitian conjugation,
for a total of fifteen classes of matrices over these three algebras. However, the classes of real hermitian (class AI)/antihermitian and real symmetric/antisymmetric (class D) matrices are the same, respectively. In addition, complex antihermitian matrices can be identified with complex hermitian matrices (class A) because,
%Notably, there are certain constraints that ``shine'' for their absence, namely, $A$ complex skew-Hermitian, $A$ quaternionic symmetric, and $A$ quaternionic anti-symmetric. The first condition is, equivalently, $\mathbb{F}=\mathbb{C}$ and $A^\dagger=-A$. In this case, 
If \(A\) is any complex antihermitian matrix, then \(iA\) is complex hermitian and the other way round. Therefore, a suitable factorization of \(A\)
can be found using that of $iA$.
Finally, the ten AZ classes do not include quaternionic symmetric/antisymmetric matrices explicitly. The reason is that these matrices are equivalent to quaternionic antihermitian (class C) and quaternionic hermitian (class AII) matrices respectively. Explicitly,
if \(A\) is a quaternionic matrix such that \(A^\dagger=\pm A\),
then \((jA)^T=\mp jA\).
Therefore, the ten conditions
in Table~\ref{table:classes2} exhaust {\em all} possibilities for matrix Laurent polynomials over $\mathbb{F} = \mathbb{R},\mathbb{C},\mathbb{H}$
that are symmetric, antisymmetric, Hermitian or anti-Hermitian,
or unstructured. This is the basic algebraic justification of 
the special role of the number ten.

%%%%%%%%%%%%%%%%%%%%%%%%%%%%%%%%

\medskip 
Let us conclude this section with a summary of our main results thus far. For a given clean, semi-infinite model in any one of the AZ symmetry class and subject to open BCs, we are interested in a WH factorization of the form
\(
\widetilde{H} = \widetilde{H}_+ \widetilde{H}_F \widetilde{H}_+^\dagger
\)
and consistent with the classifying symmetries. That is, \(\widetilde{H}\) and \(\widetilde{H}_F\) should belong to the same AZ symmetry class. Thanks to the functional calculus of BBT operators, we can focus on the symbol (analytic continuation of the Bloch Hamiltonian) in the following, and drop the tildes.
Depending on the symmetry class of $H$, 
the side and middle factors can be chosen to have additional structure. The middle factor is always denoted by $H_F$, where $F$ stands for flat-band. 
The constraints on the middle and side factors are summarized in Table \ref{table:wh}. 

\begin{table}[t]
	\begin{center}
		\begin{tabular}{|c|c|c|c|}
			\hline
			{\bf Symmetry} & {\bf Constraints on $A$} & {\bf Side factors} & {\bf Middle factor} \\
			{\bf class} &&&\\
			\hline\hline
   AIII &  $A \in \mathbb{C}_{d}[z,z^{-1}]$ &
		$A_+ \in \mathbb{C}_{d}[z],\ A_-\in \mathbb{C}_{d}[z^{-1}]$
		& $D(\mathcal{K})$
		\\[10pt]
		\hline	
		BDI &  $A \in \mathbb{R}_{d}[z,z^{-1}]$ &
		$A_+ \in \mathbb{R}_{d}[z],\ A_-\in \mathbb{R}_{d}[z^{-1}]$
		& $D(\mathcal{K})$
		\\[10pt]
		\hline	
		CII &  $A \in \mathbb{H}_{d}[z,z^{-1}]$ &
		$A_+ \in \mathbb{H}_{d}[z],\ A_-\in \mathbb{H}_{d}[z^{-1}]$
		& $D(\mathcal{K})$
		\\[10pt]
		\hline	
			A & $A = A^\dagger \in \mathbb{C}_{d}[z,z^{-1}]$ & $A_+ = A_-^\dagger \in \mathbb{C}_{d}[z]$ & 
            $D_1(\mathcal{K},s)$ \\ 
            [10pt]\hline
            
			AI & $A = A^{\rm T} \in \mathbb{R}_{d}[z,z^{-1}]$ & $A_+ = A_-^{\rm T} \in \mathbb{R}_{d}[z]$ & $D_1(\mathcal{K},s)$\\[10pt]\hline
			
			AII & $A = A^\dagger \in \mathbb{H}_{d}[z,z^{-1}]$ & $A_+ = A_-^\dagger \in \mathbb{H}_{d}[z]$ & $D_1(\mathcal{K},s)$\\[10pt]\hline
				CI & $A = A^{\rm T} \in \mathbb{C}_{d}[z,z^{-1}]$ & 
		$A_+ = A_-^{\rm T}  \in \mathbb{C}_{d}[z]$ & 
		$D_2(\mathcal{K})$\\[10pt]
		\hline	
							
		DIII & $A = -A^{\rm T} \in \mathbb{C}_{d}[z,z^{-1}]$ & 
		$ A_+ = A_-^{\rm T} \in \mathbb{C}_{d}[z]$
		& 
		$D_3(\mathcal{K})$\\[10pt]
		\hline	
  	D & $A = -A^{\rm T} \in \mathbb{R}_{d}[z,z^{-1}]$ & $A_{+} = A_-^{\rm T}\in \mathbb{R}_{d}[z]$ & $D_3(\mathcal{K})$\\[10pt]\hline
  	C & $A = -A^\dagger \in \mathbb{H}_{d}[z,z^{-1}]$ & $A_{+} = A_-^\dagger \in \mathbb{H}_{d}[z]$ & $D_4(\mathcal{K})$\\[10pt]\hline
	\end{tabular}
		\caption[Symmetric WH factorization of 1D Hamiltonians]{\small 
		The form of the three factors in the factorization of $A$, where $A$ is the relevant block of
		the Hamiltonian defined in Sec. \ref{AZblocks}. 
		\label{table:wh}}
  \end{center}
\end{table}

\bigskip

\section{Spectral flattening using the symmetric Wiener-Hopf factorization}
\label{sec:spectralflattening}

We now explain how the SWH factorizations lead to spectral flattening of 1D systems with periodic and semi-open BCs, respectively. Notably, the spectrally flattened Hamiltonians thus obtained are ``dimerized,'' in the sense we explain below. 

\subsection{Spectral flattening under periodic boundary conditions}

In this section we provide a physical interpretation and a practical application of the middle factor in the SWH factorization of the Bloch Hamiltonian. As it turns out, this factor describes a flat-band, dimer Hamiltonian that is in the same symmetry class and topological phase as the original Hamiltonian. Thanks to this property, we can calculate the bulk topological invariant of a Bloch Hamiltonian from the middle factor of its factorization alone. In this sense, the SWH factorization yields a kind of spectral flattening of the Bloch Hamiltonian. 

The conventional approach to spectral flattening is to assign energy 
eigenvalues $-1$ ($+1$) (in arbitrary units) to all the filled (empty) bands of states without changing the states themselves \cite{Kitaev09}. Assuming that the zero of the energy lies in the band gap, this operation can be described as  
\(
H \mapsto \text{sgn}(H).
\)
In general, the flat-band Hamiltonian obtained as a result of conventional spectral flattening need not be local. By contrast, the SWH spectral flattening map \(H\mapsto H_F\) ensures that the flat-band Hamiltonian is a dimer Hamiltonian. This improvement
comes at the expense of deforming the eigenstates of the original Hamiltonian without changing the symmetry class or topological phase. The SWH spectral flattening has another advantage over the standard one: it allows one to construct a {\it local} adiabatic transformation that transforms the given Hamiltonian into the spectrally flattened one. We will show how to achieve this shortly, thereby providing a physical interpretation of the side factors in the SWH factorization.

Let us justify our physical picture for and mathematical claims about \(H_F\). Notice that the middle factor \(H_F\), like the original Bloch Hamiltonian, is a Hermitian matrix Laurent polynomial. Therefore, it itself represents {\it a} Bloch Hamiltonian. We need to show that 
\begin{enumerate}
\item the eigenvalues of $H_F(e^{ik})$ for $k \in [-\pi,\pi)$ are independent $k$; 
\item in real space, $H_F$ can be expressed as a direct sum of independent local Hamiltonians acting on at most two elements of the ordered basis; 
\item $H_F$ is in the same symmetry class as $H$; and
\item $H_F$ is in the same topological phase as $H$.
\end{enumerate}  
The first three properties are straightforward to see from the form of $H_F$ in each class, see 
\ref{exmp:ssh} for a concrete example. The rest of this 
section is devoted to establishing the fourth property. 

\begin{thm}
\label{thm:adiabatic}
If $H$ be an invertible (fully-gapped at zero energy) Bloch Hamiltonian in the symmetry class X, then then there exists a continuous family of Bloch Hamiltonians $H(t)$ for $t\in[0,1]$ such that 
\begin{enumerate}
\item $H(z,0) = H(z)$ and $H(z,1) = UH_F(z)U^\dagger$ for some internal (position independent) unitary transformation $U$; 
\item $H(z,t)$ is invertible for all $t$ and for all $z$ on the unit circle; and
\item $H(z,t)$ belongs to the symmetry class X for all $t$.
\end{enumerate}
\end{thm}

\begin{proof}
We construct a family of Hamiltonians satisfying the above requirements.
Let $H(z) = H_+(z)H_F H_+^\dagger(z)$ be a SWH factorization of $H$.
Then the family is given by $H(z,t) = H_+(z,t) H_F H_+^\dagger(z,t)$, where
\[
H_+(z,t) = \left\{\begin{array}{lcl}
H_+(z(1-2t))&  \text{if} &
 0\le t {\color{blue} \le} 1/2,\\[5pt]
 S_+^{(2-2t)}U_+ & \text{if} & 1/2 \le t \le 1.
 \end{array}\right.,
\]
where $S_+$ and $U_+$ are the positive definite and unitary matrices
appearing in the polar decomposition $H_+(z,t=1/2)=H_+(z=0) = S_+U_+$
of the left factor $H_+(z,t)$ in the SWH factorization of $H(t)$ at \(t=1/2\). 
The first condition in the Lemma is satisfied by construction. 
The second condition follows from the observation that if
$\{z_\ell\}$ are the roots of $\det H_+(z)=0$, then the roots of 
$\det H_+(z,t) =0$ for $0\le t <1/2$ are given by$\{z_\ell/(1-2t)\}$.
Therefore, they remain outside the unit circle for $0\le t <1/2$,
so that $H_+(z,t)$ remains invertible on the unit circle. In fact,
this construction ensure that $H_+(z,t)$ remains invertible inside the 
unit circle as well, and therefore the factorization 
$H(z,t) = H_+(z,t) H_F H_+^\dagger(z,t)$ is a SWH 
factorization of $H(z,t)$. 

The third condition is easy to verify on a case-by-case basis. For illustration, we consider an arbitrary Hamiltonian in class BDI and see why $H(t)$ is in the same symmetry class for any \(t\). Since $H_+(z,t)$ remains block-diagonal throughout by construction,
%it is easily verified that 
$H(z,t)$ is block off-diagonal for all $t$. Note that
$H_+(z,t)$ is a real matrix polynomial for $t\in[0,1/2]$ 
since the factor $(1-2t)$ is real. 
To see that $H_+(z,t)$ is real for $t\in[1/2,1]$, first note that
$H_+(0)$ as dictated by the SWH factorization for class BDI. and therefore
$U_+ = S_+^{-1}H_+(0)$ is also real, and thus the claim follows 
Therefore, $H(z,t)$ belongs to the symmetry class 
BDI for all $t$.  Similar arguments apply to all symmetry classes. For symmetry classes AII, CII, and C, one makes use of the result that for quaternionic matrices, polar decomposition in terms of quaternionic positive definite and quaternionic unitaries exists \cite{Loring12, Feng10}.
\end{proof}

An immediate corollary of Theorem \ref{thm:adiabatic} is that two Hamiltonians in the same symmetry class with the same set of partial indices can be connected adiabatically:
\begin{coro}
If $H_1$ and $H_2$ both belong to
the symmetry class $X$ and have the same set of symmetric
partial indices $\mathcal{K}$, 
then there exists an adiabatic transformation connecting $H_1$ to $H_2$.
\end{coro}
\begin{proof}
The assumptions guarantee that \(H_1\) and \(H_2\) are both adiabatically connected to the same flat-band Hamiltonian $H_F$ by the explicit construction of the theorem above. By reversing the path from $H_2$ to $H_F$ and appending it to the path from $H_1$ to $H_F$, we obtain adiabatic transformation of $H_1$ into $H_2$.  
\end{proof}

\begin{remark}
Note that the conditions of the Corollary are {\em sufficient} but not necessary for adiabatic equivalence: there exist adiabatically connected Bloch Hamiltonians with distinct sets of partial indices. In other words, individual partial indices are not homotopy invariants. However, as we will see shortly, certain combinations of them are homotopy invariants.
\end{remark}

The second condition in Theorem \ref{thm:adiabatic}
ensures the existence of non-zero energy gap during the adiabatic 
transformation. We can, in fact, prove a lower bound on the energy gap during this transformation.
\begin{prop}
\label{prop:boundgap}
	Let $H = H_+ H_F H_+^\dagger$ be a SWH factorization of a translation-invariant $H$ (a block-Laurent operator).
	Then there exists no eigenvalue of $H$ in the energy window $[-\Delta,\Delta]$,
	where 
	\[
	\Delta = \frac{1}{\norm{H_+^{-1}}^2}.
	\] 
\end{prop} 

\begin{proof}
	We use the fact that $\epsilon$ belongs to the spectral gap of $H$ if and only if $H -\epsilon$ is invertible. Let $\epsilon \in [-\Delta,\Delta]$. Then
	\[
	H - \epsilon = H_+ (H_F - \epsilon H_+^{-1}{H_+^{\dagger}}^{-1})H_+^\dagger,
	\]	
	and $\norm{\epsilon H_+^{-1}{H_+^{\dagger}}^{-1}} \le \Delta \norm{H_+^{-1}}^2 \le 1$.
	Since all eigenvalues of $H_F$ are either $+1$ or $-1$, we find the middle factor
	in the expansion of $H-\epsilon$ is invertible. We therefore obtain that $H-\epsilon$
	is invertible for any $\epsilon \in [-\Delta,\Delta]$, thus proving the lemma.
\end{proof}

\noindent 
Using this lemma, 
%it is easy to verify 
it follows that the spectral gap during the adiabatic transformation described in the proof of Theorem \ref{thm:adiabatic} 
is always lower-bounded by $\Delta_\text{min}$, where 
\[
\Delta_\text{min} = \min_{t\in [0,1]} \Delta(t) = \min_{|z| \le 1} \frac{1}{\norm{H_+^{-1}(z)}^2}.
\]

The next proposition, which is one of the main results of this section, states that the spectrally flattened, ``dimerized'' Hamiltonian $H_F$ has the same value for the bulk invariant as $H$ does.

\begin{prop}
\label{thm:bulk}
Let $H = H_+H_FH_+^\dagger$ be a SWH factorization
of $H$ in symmetry class X, where X is any one of the five non-trivial 
symmetry classes in 1D, and let $\mathcal{K} = \{\ind_m,\ m=1,\dots,d\}$ denote the set of symmetric partial indices of $A$. Then,
\[
Q^B(H) = Q^B(H_F)= \left\{ \begin{array}{lcl}
\sum_{m=1}^{d} \ind_m & \text{if} & \text{X} = \text{AIII} \\
2\sum_{m=1}^{d} \ind_m & \text{if} & \text{X} \in \{ \text{BDI, CII} \}\\
(-1)^{\sum_{m \in \mathcal{M}_+} \ind_m}  & \text{if} & \text{X} \in \{\text{D, DIII}\}
\end{array}\right. .
\]
\end{prop}

Before providing a formal proof, note that the first equality follows if one recognizes that 
the quantities \(Q^B(H)\) are homotopy invariants, since we have shown that \(H\) and \(H_F\) are homotopic (adiabatically connected) under the conditions of the Theorem. The second equality then follows by direct calculation. 

\begin{proof} We now show by calculation that \(Q^B(H)=Q^B(H_F)\) for each non-trivial symmetry class.

\smallskip

\noindent\underline{AIII}: 
The bulk invariant is the winding number of
$\det A(z)$. Let us express this determinant as
\[
\det A(z) = \det A_+(z)\det D(z) \det A_-(z).
\]
From the definition of SWH factorization, it follows
that $\det A_+(z)$ has no roots and poles inside the unit circle, whereas for $\det A_-(z)$, the number of roots is equal to the number of poles inside the unit circle. Therefore, the winding number of both $\det A_+(z)$ and $\det A_-(z)$ is zero.
We conclude that the winding number of $\det A(z)$ is 
the same as the winding number of $\det D(z)$, which is the
bulk invariant for $H_F$. Hence, the bulk invariant is unchanged as a result of spectral flattening in these three classes. 
The 
bulk invariant may be calculated explicitly:
\begin{multline*}
Q^B(H_F) = \frac{1}{2\pi i}\int_{0}^{2\pi}\!\! dk 
\,\frac{d\det D(e^{ik})/dk}{\det D(e^{ik})} 
= \frac{1}{2\pi i}\int_{0}^{2\pi}\!\! dk 
\,\frac{i(\sum_{m=1}^{d} \ind_m)
e^{ik\sum_{m=1}^{d} \ind_m}}
{e^{ik\sum_{m=1}^{d} \ind_m}} 
= \sum_{m=1}^{d} \ind_m.
\end{multline*}

\smallskip 

\noindent\underline{Class BDI and CII}: The proof is nearly identical as 
that for class AIII. For class BDI, the factor of $2$ comes from the formula for
the bulk invariant, Eq.\,\eqref{invariantchiral2}. 
For class CII, the factor of $2$ is due to the quaternionic nature, 
which results in the set of standard partial indices 
consisting of two copies of the set of symmetric partial indices.

\smallskip

\noindent\underline{Class D}: 
The SWH factorization allows us to write
$A = A_+D_3A_+^{\rm T}$.
The bulk topological invariant for this system is 
\[
Q^B = \text{sgn}\left[\frac{\text{Pf }(A(z=1))}{\text{Pf }(A(z=-1))}\right].
\]
Using a property of the Pfaffians \cite{Pfaffian},
we get
\[
\text{sgn}[\text{Pf }(A(z=\pm1))] = \text{sgn}[\det A_+(z=\pm1) \text{Pf }(D_3(z=\pm 1))],
\]
which leads to the expression
\[
Q^B = \frac{\text{sgn}[\det A_+(z=1)]}{\text{sgn}[\det A_+(z=-1)]}
\frac{\text{sgn}[\text{Pf }(D_3(z=1))]}{\text{sgn}[\text{Pf }(D_3(z=-1))]},
\]

We will now prove that $\text{sgn}[\det A_+(z=1)] = \text{sgn}[\det A_+(z=-1)]$.
To see this, first note that all eigenvalues of $A(z)$ for $|z|=1$ are imaginary
as $H(z) = iA(z)$ in the Majorana operator basis. Therefore,
$\det A(z) = \det A_+(z) \det A_+^T(z^{-1}) \det D_3(z) = (\det A_+(z))^2 \det D_3(z)$ is real on the unit circle $|z|=1$.
Since $\det D_3(z) = (-1)^{d/2} \in \mathbb{R}$, $(\det A_+(z))^2$ is always real and non-zero. Now $(\det A_+(z))^2$ does not change its sign on the unit circle
because $A_+(z)$ is invertible on the unit circle.
Therefore, $\det A_+(z)$ is either always real or always imaginary.
Furthermore, since $A_+(z)$ is a real matrix polynomial,
$\det A_+(z=\pm1) \in \mathbb{R}$. Therefore, on the unit circle,
$\det A_+(z)$ must be real and non-zero. We conclude that $\text{sgn}[\det A_+(z=1)] =\text{sgn}[\det A_+(z=-1)]$. Now we conclude that
\[
Q^B = \frac{\text{sgn}[\text{Pf }(D_3(z=1))]}{\text{sgn}[\text{Pf }(D_3(z=-1))]}.
\]
This is precisely the expression for bulk topological invariant of the flat-band Hamiltonian described by $H_F = iD_3$. Since $D_3$ has an extremely simple form, 
its Pfaffian can be computed explicitly using the standard formula. We obtain
\[
Q^B(D_3) = \text{sgn}\left[\frac{\text{Pf }(D_3(z=1))}{\text{Pf }(D_3(z=-1))}\right]
= (-1)^{\sum_{m \in \mathcal{M}_+} \ind_m}.
\]

\noindent\underline{Class DIII}: 
The SWH factorization of $A$ is $A = A_+D_3A_+^{\rm T}$.
The bulk invariant is
\[
Q^B = \left[ \frac{\text{Pf}(A(z=1))}{\text{Pf}(A(z=-1))}
\right]
\exp\left(-\frac{1}{2}\int_{k=0}^{\pi}dk \frac{d}{dk}\log \det A(e^{ik})\right). 
\]
Using the factorization, the exponent can be shown to be equal to
\[
\exp\left(-\frac{1}{2}\int_{k=0}^{\pi}dk \frac{d}{dk}\log \det A(e^{ik})\right) 
=\frac{\det(A_+(z=-1))}{\det(A_+(z=1))}.
\]
Since $\text{Pf}(A(z=e^{ik})) = \det A_+(k) \text{Pf}(D_3(z=e^{ik}))$
for $k=0,\pi$, we get
\[
Q^B = \frac{\text{Pf }(D_3(z=1))}{\text{Pf }(D_3(z=-1))}.
\]
This is the bulk invariant of the flat-band system $H_F$.
Similar to the last case, we obtain
\[
Q^B (H_F)= \frac{\text{Pf }(D_3(z=1))}{\text{Pf }(D_3(z=-1))} 
= (-1)^{\sum_{m \in \mathcal{M}_+} \ind_m}.
\]
\end{proof}

Accordingly, 
the bulk invariants of 1D systems in all five non-trivial AZ symmetry classes are functions of the partial indices of the Bloch Hamiltonian. Let us look at some examples.

\begin{exmp} 
\label{exmp:ssh}
The Su–Schrieffer–Heeger (SSH) Hamiltonian is a particularly simple two-band Hamiltonian \cite{SSH}, of the form
\[ \widehat{H} = -\sum_{j\,\mathrm{odd}} t_1 \,c_j^\dagger c_{j+1} - 
\sum_{j\,\mathrm{even}} t_2 \, c_j^\dagger c_{j+1} +\mathrm{H.c.}, \]
where $t_1, t_2$ are hopping amplitudes, which we take here to be real. A completely closed-form expression for the SWH factorization can be obtained in this case. The reduced bulk Hamiltonian for the block $\sigma=0$ is
\[
H(z) = -\begin{bmatrix} 0 & t_1 + t_2z \\
t_1 + t_2z^{-1} & 0\end{bmatrix},
\]
where $P=1$ and $\sigma$ stands for the value of spin (since the Hamiltonian being invariant under SU(2), it is block-diagonal in the basis of $\sigma_z$). The factorization of $H(z)$ according to the relevant symmetry class BDI is
\begin{equation}
\label{SSHWH}
H(z) = \left\{ \begin{array}{lcl}
\begin{bmatrix}1 & 0 \\ 0 & -(t_1z+t_2)\end{bmatrix}
\begin{bmatrix}0 & z \\ z^{-1} & 0\end{bmatrix}
\begin{bmatrix}1 & 0 \\ 0 & -(t_1z^{-1}+t_2)\end{bmatrix}
 & \quad\text{if}\quad & t_1<t_2,\\[20pt]
\begin{bmatrix}-(t_1+t_2z) & 0 \\ 0 & 1\end{bmatrix}
\begin{bmatrix}0 & 1 \\ 1 & 0\end{bmatrix}
\begin{bmatrix}-(t_1+t_2z^{-1}) & 0 \\ 0 & 1\end{bmatrix}
 & \quad\text{if}\quad & t_1>t_2.
 \end{array}\right.
\end{equation}
It is immediate to check that $H_{+}(z) = H_{-}^\dagger(z)$ is satisfied in both the parameter regimes, which is a universal requirement irrespective of the symmetry class. Furthermore, $A_+ \in \mathbb{R}[z]$ and $A_- \in \mathbb{R}[z^{-1}]$, as required in the factorization of symmetry class BDI. Finally, the middle factor is off-diagonal as stipulated by the sublattice symmetry.

In both the parameter regimes, we find that $H_F$ has dispersion relation
\[
\det [H_F(e^{ik})-\epsilon]=0 \implies \epsilon^2 = 1. 
\]
Therefore, $H_F$ is a flat-band Hamiltonian. For $t_1>t_2$, the flat-band Hamiltonian takes the simple real-space form
\[
H_F = \sum_{j,\nu}|j\rangle|\nu=0\rangle\langle j|\langle \nu=1| +\mathrm{H.c.},
\]
which is a dimer Hamiltonian. In the regime $t_1<t_2$, the real space Hamiltonian is 
\[
H_F = \sum_{j,\nu}|j\rangle|\nu=0\rangle\langle j+1|\langle \nu=1| +\mathrm{H.c.},
\]
which is once again a dimer Hamiltonian. Finally, $H_F$ is block off-diagonal with real off-diagonal blocks, and therefore satisfies the same symmetries as $H$.

In the parameter regime $t_1<t_2$, the set of symmetric partial 
indices is the singleton set $\{-1\}$, which makes its bulk invariant $Q^B(H) = Q^B(H_F) = -1$.
\end{exmp}

\begin{exmp} 
Kitaev's Majorana chain provides a paradigmatic model for $p$-wave topological superconductivity \cite{Kitaev,PRB1}, and is described by a Hamiltonian of the form
\[ \widehat{H}= -\sum_{j} \mu \,c_j^\dagger c_j - \sum_j \Big( t \,c^\dagger_j c_{j+1} - \Delta \,c_j^\dagger c^\dagger_{j+1} +\mathrm{H.c.}\Big), \]
where $\mu, t, \Delta \in {\mathbb{R}}$ denote chemical potential, hopping and pairing amplitudes, respectively. 
The SWH factorization of this Hamiltonian is tractable in the parameter regime $t = \Delta$, 
as the symbol takes the simple form 
\begin{equation}
H(z) = \begin{bmatrix}
-\mu-t(z+z^{-1}) & t(z-z^{-1}) \\
-t(z-z^{-1}) & \mu+t(z+z^{-1})
\end{bmatrix} .
\end{equation}
A change of basis by the unitary transformation $U_{\rm D}$ brings this symbol
to the canonical form
\begin{equation}
H(z) = i\begin{bmatrix}
0 & -\mu - 2tz\\
\mu + 2tz^{-1} & 0
\end{bmatrix}.
\end{equation}
The SWH factorization is
\begin{equation}
\label{KitaevWH}
H(z) = \left\{ \begin{array}{lcl}
i\begin{bmatrix}0 & -\mu z^{-1}-2t \\ -i & 0\end{bmatrix}
\begin{bmatrix}0 & iz^{-1} \\ -iz & 0\end{bmatrix}
\begin{bmatrix}0 & -i \\ -\mu z-2t & 0\end{bmatrix}
 & 
 \quad\text{if}\quad & 
 |\mu| < 2|t|,\\[20pt]
i\begin{bmatrix} 0 & 1 \\  i(\mu +2tz^{-1})& 0\end{bmatrix}
\begin{bmatrix}0 & i \\ -i & 0\end{bmatrix}
\begin{bmatrix}0 & i(\mu +2tz^{-1}) \\ 1 & 0\end{bmatrix} 
& \quad\text{if}\quad & |\mu| > 2|t|.
 \end{array}\right.
\end{equation}
The multiplier $i = (e^{i\pi/4})^2$ can be distributed over the factors
$H_+(z)$ and $H_-(z)$ without violating the relation $H_+(z) = H_-^{\rm T}(z)$.
The bulk and the boundary invariants can be calculated from the partial indices of $H(z)$.
They take values $Q^B(H) = Q^{\partial}(\widetilde{H})=-1$
for $|\mu| < 2|t|$ and $Q^B(H) = Q^{\partial}(\widetilde{H})=1$
for $|\mu| > 2|t|$, respectively, in agreement with the known values \cite{Kitaev}.
The Hamiltonian is gapless at the topological phase transition point $|\mu| = 2|t|$.
\end{exmp}

\subsection{Spectral flattening under semi-open boundary conditions}
\label{sec:symwh2}

We have discussed in Sec.\,\ref{sec:wh} the connection between partial indices of $A$ and the kernel of a block-Toeplitz operator $\widetilde{A}$. It suffices to recall the formula
\[
\Ker \widetilde{H} = ({\widetilde{H}_+^\dagger})^{-1}\ \Ker \widetilde{H}_F.
\]
Because $\widetilde{H}_F$ is a dimer, flat-band Hamiltonian,
ZMs (kernel vectors), if present, are perfectly localized on the first few sites and have an extremely simple form. We provide a concrete illustration in Example\,\ref{exmp:sshzm} below.

\begin{prop}
\label{thm:boundary}
Let $H = \widetilde{H}_+\widetilde{H}_F\widetilde{H}_+^\dagger$ be a SWH factorization
of $\widetilde{H}$ in symmetry class X, where X is any one of the five non-trivial symmetry classes in one dimension, and let $\{\ind_m,\ m=1,\dots,d\}$ denote the set of symmetric partial indices. 
Then,
\[
Q^\partial(\widetilde{H}) = Q^\partial(\widetilde{H}_F)=\left\{ \begin{array}{lcl}
\sum_{m=1}^{d} \ind_m & \text{if} & \text{X} = \text{AIII} \\
2\sum_{m=1}^{d} \ind_m & \text{if} & \text{X} \in \{\text{BDI, CII}\} \\
(-1)^{\sum_{m \in \mathcal{M}_+} \ind_m}  & \text{if} & \text{X} \in \{\text{D, DIII}\}
\end{array}\right. .
\]
\end{prop}

\begin{proof}
We will prove the claim by explicit calculation for each symmetry class.

\smallskip

\noindent\underline{AIII}: 
For this class, the boundary invariant is 
$Q^\partial = \mathcal{N}_{\nu=0} - \mathcal{N}_{\nu=1}$,
with $\mathcal{N}_{\nu=0} = \dim \Ker \widetilde{A}$
and $\mathcal{N}_{\nu=1} = \dim \Ker \widetilde{A}^\dagger$.
Because $\widetilde{A}_{+}$ and $\widetilde{A}_{-}$ are invertible, we find that
$ \dim \Ker \widetilde{A} = \dim \Ker\widetilde{D}$ 
and 
$ \dim \Ker \widetilde{A}^\dagger = \dim \Ker\widetilde{D}^\dagger $. Thus, 
\begin{eqnarray*}
Q^\partial(\widetilde{H}_F) = \mathcal{N}_{\nu=0} - \mathcal{N}_{\nu=1}
= \dim \Ker \widetilde{D} - \dim \Ker \widetilde{D}^\dagger. 
\end{eqnarray*}
We identify this to be the analytic
index of $\widetilde{D}$ which, using Eq.\,\eqref{index1}, can be expressed as
\[
Q^\partial(\widetilde{H}_F) = \sum_{m=1}^{d} \ind_m. 
\]
 
\smallskip

\noindent\underline{BDI, CII}:
The proof is nearly identical to that for class A. The factor of $2$ for 
class BDI comes from the Kramers' degeneracy due to time reversal symmetry, which
dictates that every ZM will have a time-reversal partner. This statement holds true for $\widetilde{H}$ as well
as $\widetilde{H}_F$. For class CII, the factor of $2$ is attributed to the
fact that the 
set of standard partial indices consists of two copies of the 
set of symmetric partial indices.

\smallskip

\noindent\underline{D}:  
Since $\widetilde{H}_+$ is invertible, we get
$ \dim \Ker \widetilde{H} = \dim \Ker (i\widetilde{D}_3) = \mathcal{N}$.
For explicit calculation of the boundary invariant, we start with the flat Hamiltonian for class D,
\begin{align}
H_F(z) = -i\sum_{m\in \mathcal{M}_+} \big(z^{\ind_m} |\bar{m}\rangle\langle{m} |\;  
- z^{-\ind_m}|m\rangle\langle{\bar{m}} |\big) 
+  i\sum_{\substack{m\in \mathcal{M}_0\\[1pt] m\le d/2}} \big( |\bar{m}\rangle\langle{m} |\;  
- |m\rangle\langle{\bar{m}} |\big).
\end{align}
The boundary invariant under open BCs is
\[
Q^\partial(\widetilde{H}_F) = (-1)^{\mathcal{N}} = 
(-1)^{\dim \Ker \widetilde{H}_F}  = (-1)^{\sum_{m \in \mathcal{M}_+} \ind_m}. \]

\smallskip

\noindent\underline{DIII}: 
Once again since $\widetilde{H}_+$ is invertible, we get
$ \dim \Ker \widetilde{H} = \dim \Ker \widetilde{H}_F = \mathcal{N}$,
which leads to the equality of the boundary invariant $(-1)^{\mathcal{N}/2}$.
For this class, the flat-band Hamiltonian $H_F$ is off-diagonal, with lower block
\begin{align*}
iD_3(z) = -i\sum_{m\in \mathcal{M}_+} \big(z^{\ind_m} |\bar{m}\rangle\langle{m} |\;  
- z^{-\ind_m}|m\rangle\langle{\bar{m}} |\big) 
 +  i\sum_{\substack{m\in \mathcal{M}_0\\[1pt] m\le d/2}} \big( |\bar{m}\rangle\langle{m} |\;  
- |m\rangle\langle{\bar{m}} |\big).
\end{align*}
The boundary invariant is
\[
Q^\partial(\widetilde{H}_F) = (-1)^{\mathcal{N}/2}= (-1)^{\dim \Ker i\widetilde{D}_3} = 
(-1)^{\sum_{m \in \mathcal{M}_+} \ind_m}.
\]
\end{proof}

\begin{exmp}
\label{exmp:sshzm}
For the SSH Hamiltonian in Example \ref{exmp:ssh}, the factorization in Eq.\,\eqref{SSHWH} immediately allows us to compute the zero-energy edge state for the case $t_1<t_2$. 
The unnormalized kernel of $\widetilde{H}$ is given by 
\[
\Ker \widetilde{H} = ({\widetilde{H}_+^\dagger})^{-1} \Ker \widetilde{H}_F = 
\begin{bmatrix}1 & 0 \\ 0 & -(t_1\S^\dagger+t_2)\end{bmatrix}^{-1}
\begin{bmatrix}0 \\ 1\end{bmatrix}.
\]
Some algebra reveals that the zero-energy edge state $|\psi\rangle$ has 
a wavefunction given by
\[
|\psi\rangle = \sqrt{1-(t_1/t_2)^2} \,
\sum_{j = 0}^{\infty} (-t_1/t_2)^j|j\rangle\begin{bmatrix}0 \\ 1\end{bmatrix}.
\]
We have $\mathcal{N}_{\nu=0} = 0$ and $\mathcal{N}_{\nu=1} = 1$, so that 
$Q^{\partial}(\widetilde{H}) = \mathcal{N}_{\nu=0} - \mathcal{N}_{\nu=1} = -1$ in the parameter 
regime under consideration ($t_1 < t_2$).
\end{exmp}

\section{A bulk-boundary correspondence for boundary-disordered systems}
\label{sec:bb}

In this section, we use SWH factorizations to prove a bulk boundary correspondence for 1D systems subject to arbitrary BCs. We next extend this result to 1D junctions that are formed by two leads in the same symmetry class and for which 
the tunneling term also preserves the symmetries of that class.

\subsection{A bulk-boundary correspondence for arbitrary boundary conditions}

A quick glance at Propositions \ref{thm:bulk} and \ref{thm:boundary} yields the following interesting corollary:
\begin{coro}
\label{coro:bbc}
Let $H$ be the Bloch Hamiltonian of a 1D system in one of the non-trivial AZ symmetry classes and \(\widetilde{H}\) the associated block-Toeplitz Hamiltonian for the same system subject to open BCs. Then, $Q^B(H) = Q^\partial(\widetilde{H})$.
\end{coro}

\noindent 
Clearly, this result is an example of a bulk-boundary correspondence, but it is too weak to be of practical importance: from bulk properties it infers surface properties of a strictly clean Hamiltonian, subject specifically to open BCs.
Within the framework of the SWH factorization, one can strengthen
this Corollary to allow for arbitrary BCs or boundary disorder, depending on ones' point of view. It is this strengthened result that we regard as bulk-boundary correspondence for 1D systems of independent fermions.  

To derive such a strengthened result, we need to relate the boundary invariants of semi-infinite systems subject to arbitrary BCs to those subject to semi-open BCs.
As it turns out, one can extend the SWH spectral flattening
technique to semi-infinite systems subject to arbitrary BCs or boundary disorder, and still calculate the boundary invariants from an appropriate flattened Hamiltonian. Consider a Hamiltonian $\widetilde{H} + {W}$, where both $\widetilde{H}$ and ${W}$ satisfy the constraints of the same symmetry class X. The operator $\widetilde{H}$ is a block-Toeplitz operator and, 
as before, it is associated to the matrix Laurent polynomial \(H(z)\). The operator $W$ is of finite rank, 
that is, it is different from the zero operator on a finite-dimensional subspace only.  Then we can write  
\[
\widetilde{H}+{W} \equiv \widetilde{H}_+ (\widetilde{H}_F+{{W}_F}) \widetilde{H}_+^\dagger, \quad {{W}_F} \equiv (\widetilde{H}_+)^{-1}{W}(\widetilde{H}_+^{\dagger})^{-1}.
\]
We refer to the Hamiltonian $\widetilde{H}_F+{{W}_F}$ as the {\em SWH spectral flattening} of \(\widetilde{H}+{W}\). Note, however, that the eigenvectors of $\widetilde{H}_F+{{W}_F}$ that overlap with the support of $W_F$ can have eigenvalues outside the set $\{+1,0,-1\}$. The next proposition relates the boundary invariant of \(\widetilde{H}+{W}\) to that of $\widetilde{H}_F+{{W}_F}$ in each symmetry class.

\begin{prop}
\label{lemmabound}
Let $\widetilde{H}_F+{{W}_F}$ be the SWH spectral flattening of \(\widetilde{H}+{W}\), as defined above. Then
\begin{enumerate}
\item
$W_F$ is a finite-rank operator in the same symmetry class as \(W\).
\item 
\(Q^\partial(\widetilde{H}+{W}) = Q^\partial(\widetilde{H}_F+ {{W}_F})\).
\end{enumerate}
\end{prop}
\begin{proof}
1. The fact that $W_F$ is of finite support follows from the fact that $\widetilde{H}_+^{-1}$ is a function of $T$ only,
and $(\widetilde{H}_+^{-1})^\dagger$ is a function of $T^\dagger only$. 
Say that $W$ is supported on the sites $j=1,\dots,R_w$, that is, \(\langle j |W = 0\) and \(W|j'\rangle = 0\) for all $j,j' > R_w$. Then,
\[
\langle j |W_F = [\langle j |\widetilde{H}_+^{-1}W] (\widetilde{H}_+^{-1})^\dagger = 0,\quad
W_F|j'\rangle = \widetilde{H}_+^{-1}[W (\widetilde{H}_+^{-1})^\dagger]|j\rangle = 0,
\]
where the terms grouped in square brackets vanish because
\(
\langle j|T = \langle j+1 |\) and \(T^\dagger|j\rangle = |j+1\rangle.
\)
Therefore, ${{W}_F}$ has support at most on the first $R_w$ lattice sites. It is straightforward to see that $W$ and $W_F$ have the same structure, and therefore belong to the same symmetry classes. For example, in the case of chiral classes, $\widetilde{H}_+$ being block-diagonal it commutes with the chiral symmetry, which means that ${W}_F$ anti-commutes 
and hence obeys the chiral symmetry. In classes BDI and D, the entries of $\widetilde{H}_+$ are real, therefore the entries of ${{W}_F}$ are real and imaginary, respectively, as per the requirement of these classes. Similarly, in the case of 
class CII, the entries of $\widetilde{H}_+$, and hence of $W_F$, are quaternionic. 

2. The proof is identical to that of the first part of Theorem \ref{thm:boundary}, except that we use the decomposition
\(
\widetilde{H}+{W} = \widetilde{H}_+ (\widetilde{H}_F+{{W}_F}) \widetilde{H}_+^\dagger.
\)
\end{proof}

We are now ready to prove the main result of this section, namely, a bulk-boundary correspondence for systems with arbitrary BCs.

\begin{thm}[{\bf Bulk-boundary correspondence}]
\label{thm:bbcorr}
Let $H$ be a Bloch Hamiltonian in symmetry class X, $\widetilde{H}$ the associated block-Toeplitz Hamiltonian, and ${W}$ a finite-rank Hermitian operator also in symmetry class X. Then, 
\(
 Q^\partial(\widetilde{H} + {W})=Q^B(H).
\)
\end{thm}
\begin{proof}
The challenge is to show that 
\(Q^\partial(\widetilde{H}_F + {{W}_F}) = Q^\partial(\widetilde{H}_F). \)
Once this equality is proved, the theorem follows from Proposition \ref{lemmabound} and Corollary \ref{coro:bbc}.
Our strategy 
is to exploit the flat-band nature of $\widetilde{H}_F$.
As a consequence, there exists a finite-dimensional subspace of the underlying Hilbert space which is decoupled from the rest 
and hosts all the topological ZMs. The rest of the system is simply a collection of dimers (see Fig.\,\ref{fig:dimers}) with no topological ZMs.  
\begin{figure}[t]
\centering
\includegraphics[width=8cm]{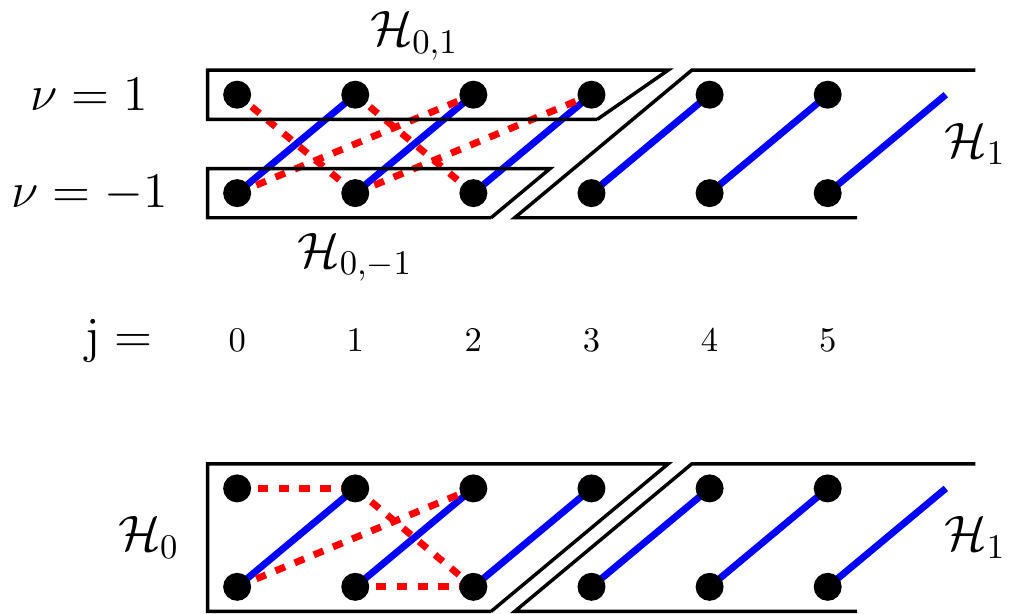}
\caption[Partition of the Hilbert space in the proof of bulk-boundary correspondence]{\small 
Partition of Hilbert space for (top) a system in class AIII
with $d=1$ and $\ind_1 = 1$, and 
(bottom) a system in class D 
with $d=2$ and $\ind_1=1,\ind_2=-1$.
The solid blue lines indicate hopping of the flat-band Hamiltonian
$\widetilde{H}_F$, and the red dashed lines those of ${{W}_F}$
for $R_w = 3$. The partition for classes BDI, CII and DIII 
is similar to the one in the top figure, except that $d_{\rm int}$ 
($2d_{\rm int}$) is necessarily 
greater than $2$ for BDI and CII (DIII).
\label{fig:dimers}}
\end{figure}

\smallskip

\noindent
\underline{AIII}:
We partition the Hilbert space $\mathcal{H}_{\rm SP}$ into three subspaces,
that is, 
\begin{equation}
\label{partition1}
\mathcal{H}_{\rm SP} = \mathcal{H}_{0,1} \oplus \mathcal{H}_{0,-1} \oplus
\mathcal{H}_{1}, \qquad {\rm{where}}
\end{equation}
\[
\mathcal{H}_{0,\nu} \equiv \text{Span}\{\,|j\rangle|\nu,m\rangle, \;
 \j = 0,\dots, R_w+\max{(\nu\,\ind_m,0)},\;
 m = 1,\dots,d \}, 
\]
with $\mathcal{H}_1$ being the orthogonal 
complement of $\mathcal{H}_{0,1} \oplus \mathcal{H}_{0,-1} \equiv \mathcal{H}_{0}$. Let ${P}_{0,\nu}$ and ${P}_1$ 
denote the orthogonal projectors on $\mathcal{H}_{0,\nu}$
and $\mathcal{H}_1$, respectively. 
Notice that the flat-band disordered Hamiltonian $\widetilde{H}_F+{{W}_F}$ is block-diagonal with respect to the partition ${P}_0 \equiv {P}_{0,1} + {P}_{0,-1}$ and $P_1$ (see Fig.\,\ref{fig:dimers}). Thus, we have 
\[
{P}_0(\widetilde{H}_F
+{{W}_F}){P}_1 = 0 = {P}_1(\widetilde{H}_F
+{{W}_F}){P}_0.
\]
Further, the block
describing its action on $\mathcal{H}_1$, namely ${P}_1(\widetilde{H}_F
+{{W}_F}){P}_1$,
 has zero-dimensional kernel and cokernel.
Therefore, the boundary invariant can be computed from the kernel and cokernel of the block $J \equiv {P}_0(\widetilde{H}_F+{{W}_F}){P}_0$ 
acting on $\mathcal{H}_0$, which is finite-dimensional.
Notice that ${P}_0$ commutes with all the 
internal symmetries considered in the AZ classification, 
and therefore $J$ always satisfies the constraints of the symmetry class.

The boundary invariant of the system under open BCs,
${W}=0$, is $Q^\partial(\widetilde{H}_F) = \sum_{m=1}^{}\ind_m$.
Let us denote by $J_{21}$ the lower off-diagonal block of the matrix $J$ acting on $\mathcal{H}_0$.
The size of this block is $\dim \mathcal{H}_{0,-1} \times \dim \mathcal{H}_{0,1}$, where
\[
\dim \mathcal{H}_{0,1} = dR_w +\sum_{m \in \mathcal{M}_+} \ind_m,\quad
\dim \mathcal{H}_{0,-1} = dR_w -\sum_{m \in \mathcal{M}_-} \ind_m .
\]
If the partial indices are non-zero, then $J$ can be a rectangular block. Because of rank-nullity theorem, we have $\dim \Ker J_{21} - \dim \text{coker} J_{21} = \dim \mathcal{H}_{0,1} - \dim \mathcal{H}_{0,-1}$,
which gives us
\[
Q^\partial(\widetilde{H_F}+{W_F}) = \dim \mathcal{H}_{0,1} - \dim \mathcal{H}_{0,-1} = 
\sum_{m=1}^{d} \ind_m.
\]
Therefore, the boundary invariant is unchanged by ${W}$.

\smallskip

\noindent
\underline{BDI}: The proof is identical to that for class AIII, 
except we get 
\[
\dim \mathcal{H}_{0,1} = 2dR_w +2\sum_{m \in \mathcal{M}_+} \ind_m,\quad
\dim \mathcal{H}_{0,-1} = 2dR_w -2\sum_{m \in \mathcal{M}_-} \ind_m 
\]
due to Kramer's degeneracy.
This leads to
\[
Q^\partial(\widetilde{H_F}+{W_F}) = \dim \mathcal{H}_{0,1} - \dim \mathcal{H}_{0,-1} = 
2\sum_{m=1}^{d} \ind_m
\]
as desired.

\smallskip

\noindent
\underline{CII}: The proof is identical to that for class AIII, 
except owing to the fact that each symmetric partial index repeats 
twice in the set of standard partial indices, we have
\[
\dim \mathcal{H}_{0,1} = dR_w +2\sum_{m \in \mathcal{M}_+} \ind_m,\quad
\dim \mathcal{H}_{0,-1} = dR_w -2\sum_{m \in \mathcal{M}_-} \ind_m, 
\]
This leads to
\[
Q^\partial(\widetilde{H_F}+{W_F}) = \dim \mathcal{H}_{0,1} - \dim \mathcal{H}_{0,-1} =
2\sum_{m=1}^{d} \ind_m
\]
as desired.

\smallskip

\noindent
\underline{D}:
We consider the bipartition of $\mathcal{H}_{\rm BdG}$ into 
$\mathcal{H}_{\rm BdG} = \mathcal{H}_{0} \oplus \mathcal{H}_{1}$, where 
\[
\mathcal{H}_{0} \equiv \text{Span}\{|j\rangle|m\rangle, \;
j = 0,\dots, R_w+\max{(\ind_m,0)}, \;\ m=1,\dots,d\},
\]
and $\mathcal{H}_1 = \mathcal{H}_{0}^\perp$.
Let ${P}_0$ denote the orthogonal projector on $\mathcal{H}_0$.
Notice that the dimension $\mathcal{H}_0$, and hence of the matrix 
${P}_0(\widetilde{H}_F +{{W}_F}){P}_0$,
is $\dim \mathcal{H}_{0} = dR_w +\sum_{m \in \mathcal{M}_+} \ind_m$,
which is odd (even) if $\sum_{m \in \mathcal{M}_+}\ind_m$ 
is odd (even). Such a matrix has a kernel of odd (even dimension). 
In other words, the parity of the number of ZMs is the same as the parity of $\sum_{m \in \mathcal{M}_+} \ind_m$. Therefore, we get $(-1)^{\mathcal{N}} = (-1)^{\sum_{m \in \mathcal{M}_+} \ind_m}$, as claimed.

\smallskip

\noindent 
\underline{DIII}:
We consider the partition of $\mathcal{H}_{\rm BdG}$ into 
$\{\mathcal{H}_{0,1}, \mathcal{H}_{0,-1},\mathcal{H}_{1}\}$ 
as in Eq.\,\eqref{partition1}. 
Following similar arguments as in the case of class D, we find that
the dimension of the kernel of the antisymmetric matrix ${P}_{0,-1}(\widetilde{H}_F +{{W}_F}){P}_{0,1}$ is odd (even) if
the $\sum_{m \in \mathcal{M}_+} \ind_m$ is odd (even).
\end{proof}

With advanced mathematical tools, more general statements regarding the boundary invariant can be proved for all classes. For classes AIII, BDI and CII, the boundary invariant is proportional to the analytic index of the off-diagonal block $A$, which is known to be a continuous (hence, constant on connected components) function on the set of Fredholm operators \cite{Gohberg91}. Similarly, for classes D and DIII, the boundary invariant is the secondary index of $A$, which is constant on connected components of antisymmetric Fredholm operators \cite{Schulz-Baldes15}. It follows that
\[
Q^\partial(\widetilde{H} + {W}) = Q^\partial(\widetilde{H}),
\]   
as long as $|| {W} || \ll \Delta E$, where $\Delta E$
is the SP bulk energy gap, ${W}$ satisfies the 
symmetries of the class, and the operator norm is considered.
Furthermore, the analytic (secondary) index is known
to be invariant under the addition of a (antisymmetric) compact 
operator of arbitrary norm. Therefore, one may replace the 
finite-range ${W}$ by a {\it compact} ${W}$ in 
Theorem \ref{thm:bbcorr}. This result also follows from the one for
finite-range operators, since compact operators are limits of
finite-range operators in the operator-norm topology,
and the limit of each boundary invariant is well-defined,
as it depends on the spectral projector of the discrete spectrum.

\subsection{A bulk-boundary correspondence for junctions}
\label{sec:interface}

So far, we have discussed the bulk-boundary correspondence in 
1D systems with one open end. We will now show that our analysis 
extends to interfaces formed by two bulks belonging to the same symmetry 
class. We first provide a definition of bulk and boundary topological invariants for interfaces in all symmetry classes, which has not been made explicit in the literature to the best of our knowledge. In particular, we will show that we can associate the bulk invariant
\begin{align}
\label{interfacebulk}
Q^{B}_I=\sum_{b=1}^nQ^{B}_b\quad (\text{AIII, BDI, CII}), \qquad
Q^{B}_I=\prod_{b=1}^n Q^{B}_b\quad (\text{D, DIII}),
\end{align}
to the junction where $n$ bulk wires meet. Here, $Q^{B}_b$ is the bulk
invariant of the $b$th wire. The boundary invariant for the junction can be 
defined similar to individual wires,
\begin{equation}
Q^\partial_I  = \left\{
\begin{array}{cc}
\mathcal{N}_{\nu=0} -  \mathcal{N}_{\nu=1}& \text{AIII, BDI, CII}\\
(-1)^\mathcal{N} & \text{D}\\   
(-1)^{\mathcal{N}/2} & \text{DIII}
\end{array}\right. .
\end{equation}
With these definitions, we will show that the equality of bulk and boundary
invariants holds for the interface, assuming appropriate symmetry
conditions are satisfied at the interface. Eq.\,\eqref{interfacebulk} implies that the invariants of individual systems combine according to the group
operation (addition/multiplication) of the group of homotopy invariant
(${\mathbb Z}$/${\mathbb Z}_2$).

\begin{figure}
\begin{center}
\includegraphics[width = 12cm] {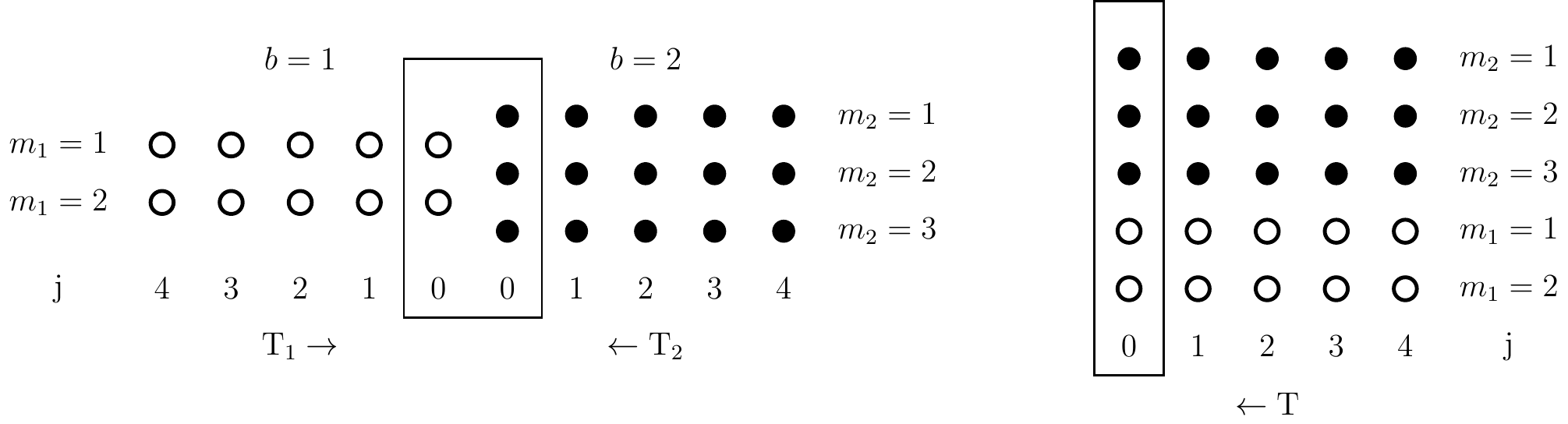}
\caption{The filled and empty circles stand for SP
Hilbert space basis for two different 1D systems in the 
same symmetry class. The sketch on the left side shows an 
interface between them. Visualizing this system in a different way
as shown in the sketch on the right side allows us to describe the system using
a single block-Toeplitz operator.}
\end{center}
\end{figure}

Let us first consider the simplest case of two 1D systems,
both belonging to a given symmetry class, forming a bridge. Both systems
are assumed to extend to infinity in the direction away from the $0$-dimensional interface. Let us label the fermionic degrees in the first 
system by $\mathcal{H}_1 = \text{Span }\{|j\rangle|m_1\rangle,\ j=0,\dots,\infty,\ m_1=1,\dots,d_1 \}$, and the ones in the second system by $\mathcal{H}_2 = \text{Span }\{|\j\rangle|m_2\rangle,\ j=0,\dots,\infty,\ m_2=1,\dots,d_2 \}$, so that the degrees labeled by $\j=0$ in the two systems are adjacent. Here, the labels $m_b$ for $b=1,2$ 
include all internal labels, namely $(\tau_{z,b}, \nu_{z,b}, \sigma_{z,b}, p_b)$. Let $T_b: \mathcal{H}_b \mapsto
\mathcal{H}_b, \ b=1,2$ denote the two shift operators with action 
\[
T_b|\j\rangle|m_b\rangle \equiv  \left\{ \begin{array}{lcl}
|\j-1\rangle|m_b\rangle & \text{if} & \j>0 \\
0 & \text{if} & \j=0 \end{array} \right. .
\]
We define ${T}$ to be the new shift operator on 
$\mathcal{H}=\mathcal{H}_1 \oplus \mathcal{H}_2$, with
action ${T}|\j\rangle|m_i\rangle \equiv T_b|\j\rangle|m_b\rangle$
on the basis. In this representation, the system in bridge configuration
can be described by a single block-Toeplitz operator, with reduced bulk
Hamiltonian
\[
H(z) = \begin{bmatrix}H_1(z) & 0 \\ 0 & H_2(z)\end{bmatrix},
\]
where $H_b(z)$
is the reduced bulk Hamiltonian of system $b$.
The bulk invariant of this system can be computed from
$H(z)$ in the same way asthe  bulk invariants of $H_1(z)$ and $H_2(z)$
are. Similarly, the boundary invariant of $\widetilde{H}+{W}$ follows 
from the appropriate expression for the symmetry class. Here, ${W}$ 
denotes the finite-range disorder at the interface. Note that even though
the two systems forming the interface are assumed to be in the same symmetry class, we have not made any assumption on the two symmetry 
operators $\{S_1\}$ and $\{S_2\}$. In fact, the two systems are allowed, in general, to have a different number of energy bands. However,
for our results on the bulk-boundary correspondence to hold, the operator
${W}$ must satisfy the symmetries $\{S\}$ defined by
\begin{equation}
\label{globalsymmetry}
S\mathcal{H}_b =S_b\mathcal{H}_b,\quad b=1,2.
\end{equation}

We now prove that the $\mathbb{Z}$ ($\mathbb{Z}_2$)
invariant of the interface is obtained by adding (multiplying) the 
invariants of the two bulks forming the interface. 

\smallskip

\noindent{\it Classes characterized by a $\mathbb{Z}$ invariant.---}
The Hamiltonians of classes AIII, BDI and CII possess a chiral symmetry, 
and are characterized by an integer invariant $\mathbb{Z}$. 
For these classes, we can permute the basis by exchanging the
order of sublattice and subsystem spaces, so that the reduced bulk
Hamiltonian takes the form
\[
H(z) = \begin{bmatrix}
0 & 0 & A_1(z) & 0\\
0 & 0 & 0 & A_2(z)\\
A_1^\dagger(z) & 0 & 0 & 0\\
0 & A_2^\dagger(z) & 0 & 0
\end{bmatrix} .
\]
It is now easy to see that the Fredholm index of the block 
$A =  \begin{bmatrix}A_1 & 0\\0 & A_2\end{bmatrix}$ is the sum
of the invidual Fredholm indices of $A_1$ and $A_2$, since
\( \log \det A = \log \det A_1 + \log \det A_2.
\) Therefore, the invariant of $H$ is equal to the sum of the individual invariants of $H_1$ and $H_2$. Note that addition is the group property of 
$\mathbb{Z}$.

\smallskip

\noindent{\it Classes characterized by a $\mathbb{Z}_2$ invariant.---}
For class D, it is easy to check that the Pfaffian invariant for $H$ is the product of Pfaffian invariant for $H_1$ and $H_2$, which follows from the identity 
\(
\text{Pf\,} [H(z=\pm1)] = \text{Pf\,} [H_1(z=\pm1)]\text{\,Pf\,} [H_2(z=\pm1)]. 
\) Using similar arguments, one can derive the same result for class DIII. 
Therefore, for classes D and DIII, the topological invariant of 
$H$ is equal to the product of the individual topological invariants.

\begin{remark}
Similar arguments to those give above allow us to establish a bulk-boundary
correspondence in a 3-way bridge, where one end of three 1D systems in the same symmetry class is connected together. Furthermore, if
a bulk-boundary correspondence is known to hold for a
single bulk in higher dimensions, then the same approach can be
used to extend its validity to interfaces formed by bulks in the same class. 
Consider, for instance, a bridge consisting of two 2D systems, that form a 1D
interface. In this case, we label the fermionic degrees of the two systems 
by $\{|\j_\parallel\rangle|\j\rangle|m_b\rangle,\; \j_\parallel \in \mathbb{Z},\;\j=0,\dots,\infty,\ m_b=1,\dots,2d_{b,{\rm int}}\}$ for $b=1,2$, and proceed along the same lines.
\end{remark}

\begin{figure}[t]
\begin{center}
\includegraphics[width = 8cm] {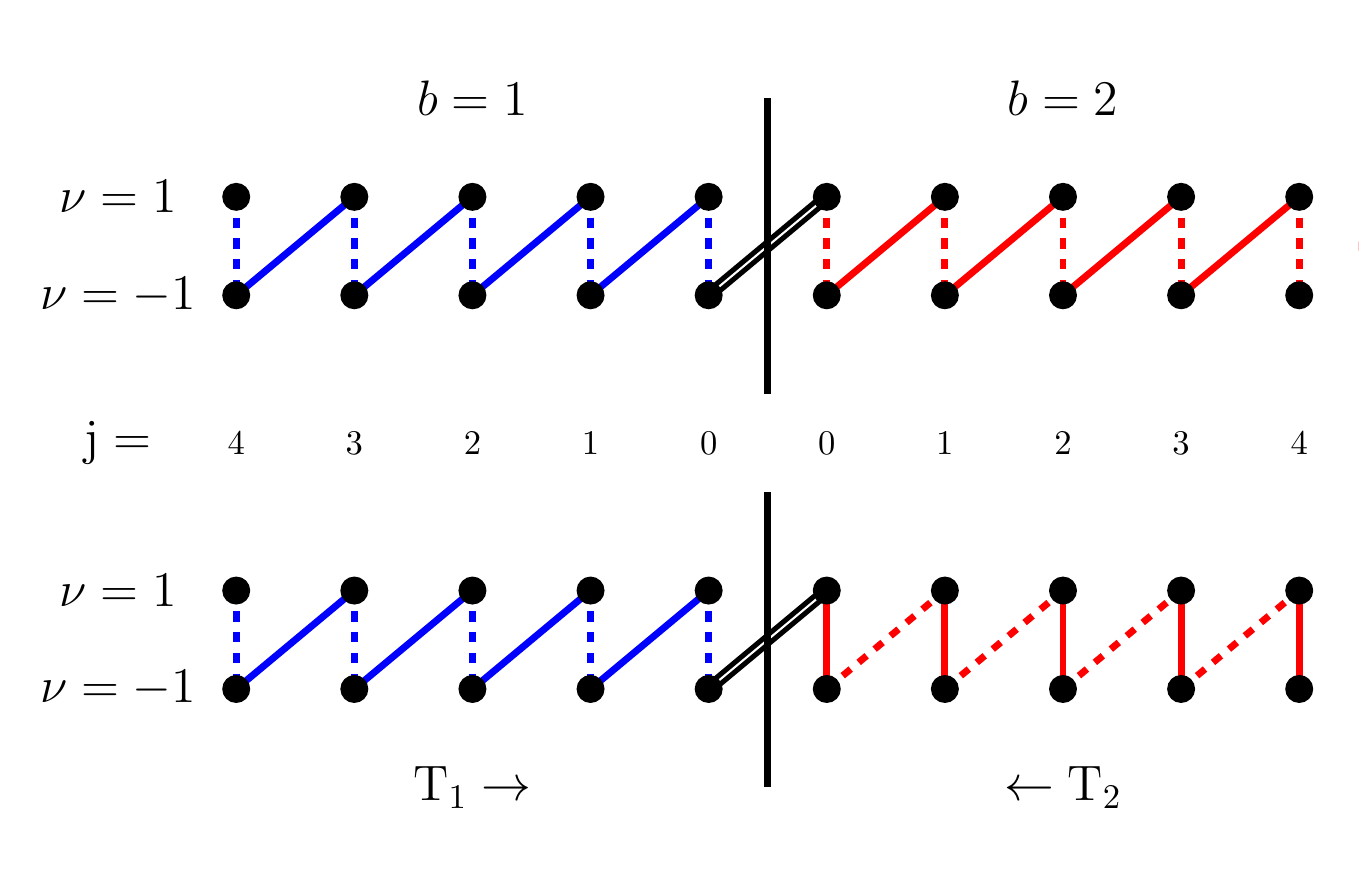}
\vspace*{-4mm}
\caption{Schematic of interfaces formed by two wires individually 
described by the SSH Hamiltonian.
On both sides of the interface (solid black lines),
the dotted lines indicate weaker hopping strength compared to 
the solid lines. The double lines indicate hopping between the 
two wires. Notice the exchange of strong and weak hopping strengths 
for $b=2$ wire in the interface shown in the bottom figure,
compared to the one shown in the top. 
The shift operator $T_1$ implements right shift
on the wire $b=1$, whereas $T_2$ implements left shift on
on the wire $b=2$.
\label{fig:SSHinterface}}
\end{center}
\end{figure}

\smallskip

\begin{exmp} 
Let us look at the bulk-boundary correspondence
of two wires described by the SSH Hamiltonian of Example \ref{exmp:ssh},
connected together at one end. In the interfaces shown in Fig.\,\ref{fig:SSHinterface}),
assuming $\{t_{1,b},t_{2,b},\ b=1,2\}$ to be the strengths of intra-cell
and inter-cell hoppings, the reduced bulk Hamiltonians of the individual wires
are
\begin{eqnarray*}
H_1(z) = 
-\begin{bmatrix} 0 & t_{1,1} + t_{2,1}z^{-1} \\
t_{1,1} + t_{2,1}z & 0\end{bmatrix},\quad 
H_2(z) = 
-\begin{bmatrix} 0 & t_{1,2} + t_{2,2}z^{-1} \\
t_{1,2} + t_{2,2}z & 0\end{bmatrix}.
\end{eqnarray*}
Notice the interchange between $z$ and $z^{-1}$,
which is a result of the convention that $T_1$ implements right shift
on $b=1$, whereas $T_2$ implements left shift
on $b=2$. The hopping between the two wires is such that 
it respects the global chiral symmetry $\nu_{z,1} \oplus \nu_{z,2}$.
For the interface shown in the top panel,
we have $t_{1,b} < t_{2,b}$ for $b=1,2$, which leads to
$Q^B_1 = -1$ and $Q^B_2 = 1$. Therefore, the interface
has $Q^B_I = Q^B_1 + Q^B_2 = 0$ as expected, and 
its boundary invariant $Q^\partial_I = 0$.
For the interface shown in the lower panel, we have
$t_{1,1} < t_{2,1}$ and $t_{1,2} > t_{2,2}$.
This leads to $Q^B_I = Q^B_1 + Q^B_2 = -1+0 = -1$,
which must equal the boundary invariant as well. In this case,
the interface must host at least one edge state of chirality $-1$.
\end{exmp}

\section{The stability and sensitivity of topological zero modes to perturbations}

\subsection{Stability}
\label{sec:stability}

Symmetry conditions protect bulk and boundary invariants against 
compact perturbations of arbitrary strength as well as against 
bulk disorder that is weak enough not to close the energy gap. 
However, for all five non-trivial classes in 1D, the total dimension of the zero-energy SP eigenspace, and hence the degeneracy of the many-body ground state, is {\it not} guaranteed to be invariant under such perturbations \cite{ours2}. 
In this section, we establish necessary and sufficient conditions for the stability of the ground manifold. These conditions may be especially useful for assessing the performance of symmetry-protected quantum memories. A somewhat similar notion was discussed in Sec. 7.4 of Ref.\,\cite{Prodan}
in the context of the integer quantum Hall effect.

The stability of the ground manifold depends entirely on the stability of the zero-energy eigenspace of the SP Hamiltonian, both in the case of topological
insulators as well as topological superconductors. In our analysis we will assume that the system is long enough to ignore finite size effects. Thus we will focus on the zero-energy eigenspace of the semi-infinite system. 
We say that the zero-energy eigenspace of a Hamiltonian ${H}_0$ is stable against perturbations in the symmetry class, if the zero-energy eigenspace is {\it nearly} unchanged
under any such perturbation ${W} = \lambda{\pert}$
for small enough strength $\lambda \in \Rds$. Since we are only interested in the strength of the perturbation ${W}$ relative to that of ${H}_0$, we will assume $\norm{{\pert}} = \norm{{H}_0} = 1$, so that $\lambda = \norm{{W}}/\norm{{H}_0}$ is the dimensionless parameter of interest. Let $P_{\Ker(\cdot)}$ denote the orthogonal projector on $\Ker(\cdot)$.

\begin{defn}
\label{def:stability}
$\Ker({H}_0)$ is said to be ``stable against a perturbation $\pert$\,'' if 
\[
\lim_{\lambda \rightarrow 0} P_{\Ker({H}_0 + \lambda{\pert})} = 
P_{\Ker({H}_0)},
\quad \lambda \in \R.
\]
$\Ker({H}_0)$ is said to be ``stable in the symmetry class of ${H}_0$'' if it is stable against every perturbation $\pert$ of unit norm in the symmetry class of ${H}_0$.
\end{defn}

We will next show that a necessary and sufficient condition for the stability of the zero-energy eigenspace is that its dimension is invariant for small but arbitrary perturbations in the class. Before formally proving this 
statement, some additional tools are needed. 
Let ${H}_0^{(-1)}$ be a ``reflexive generalized inverse'' 
of ${H}_0$, defined by the properties
\[{H}_0^{(-1)}{H}_0{H}_0^{(-1)} = {H}_0^{(-1)},\quad
{H}_0{H}_0^{(-1)}{H}_0={H}_0.
\]
We define the ``generalized condition number'' \cite{Chen03} of ${H}_0$ with respect to ${H}_0^{(-1)}$ by
\begin{equation}
\label{gencond}
{\cond}({H}_0,{H}_0^{(-1)}) \equiv \norm{{H}_0}\norm{{H}_0^{(-1)}}.
\end{equation}
If $H_0$ is invertible, then ${H}_0^{(-1)} = H_0^{-1}$ is unique,
and in this case we will drop the second argument and simply denote
the generalized condition number by ${\cond}({H}_0)$.
Generally, a large value of the generalized condition number implies that the associated system of equations is ill-conditioned. We will see in the next lemma that it also means that the kernel of the operator is very sensitive to perturbations.  

\begin{lem}
\label{lemmastability}
For $\lambda <  1/{\cond}({H}_0,{H}_0^{(-1)})$,
we have $\dim \Ker ({H}_0 + {W}) \le \dim \Ker ({H}_0) 
%\equiv \mathcal{N}
$.
If, in addition, the equality holds for some ${W}$, then
\begin{equation}
\label{newkernel}
\Ker({H}_0 + {W}) = (\widetilde{I} + {H}_0^{(-1)}{W})^{-1}\Ker({H}_0).
\end{equation}
\end{lem}

\begin{proof}
We recall that the pseudo-inverse ${H}_0^{(-1)}$ of ${H}_0$ satisfies ${H}_0^{(-1)}{H}_0{H}_0^{(-1)}={H}_0^{(-1)}$. With a simple manipulation, one can write ${H}_0^{(-1)}({H}_0 + {W}) = {H}_0^{(-1)}{H}_0(\widetilde{I}+ {H}_0^{(-1)}{W})$.
Due to the restriction on the norm of ${W}$, the term inside the bracket on the right hand-side is invertible, and hence 
\[
\Ker {H}_0^{(-1)}({H}_0 + {W}) 
= (\widetilde{I} + {H}_0^{(-1)}{W})^{-1}
\Ker({H}_0^{(-1)}{H}_0) = (\widetilde{I} + {H}_0^{(-1)}{W})^{-1}\Ker {H}_0.
\]
Since $\dim \Ker {H}_0^{(-1)}({H}_0 + {W}) \ge 
\dim \Ker({H}_0 + {W})$, we have proved the first statement.
In the case of equality of dimensions, we have
$\Ker {H}_0^{(-1)}({H}_0 + {W}) = \Ker({H}_0 + {W})$,
which proves the second statement of the lemma.
\end{proof}

We are now in a position to characterize the stability of the zero-energy eigenspace: 

\begin{prop}
A necessary and sufficient condition for the stability of $\Ker ({H}_0)$ is 
that the number of ZMs is invariant against small but
arbitrary perturbations in the symmetry class, that is,
\[
\lim_{\lambda \mapsto 0} \dim\Ker({H}_0 + \lambda{\pert}) = \dim\Ker({H}_0),
\quad \lambda \in \Rds,
\]
for every ${\pert}$ of unit norm in the symmetry class.
\end{prop}

\begin{proof}
The necessity of this condition is straightforward, since if it is not satisfied, the limit of projectors in the definition of stability is ill-defined. The sufficiency follows Lemma \ref{lemmastability}.
\end{proof}

Our next step is to determine the conditions under which the dimension of
the zero-energy eigenspace is stable for each of the symmetry classes. This condition can be written in terms of dimension of the zero-energy subspace and the topological invariants.

\begin{thm}[Stability of zero modes]
\label{thm:stability}
The zero energy eigenspace of ${H}_0$ is stable 
with respect to symmetry-preserving perturbations if and only if
\begin{equation}
\label{condition}
%\mathcal{N}  
\dim \Ker ({H}_0) = \left\{
\begin{array}{cc}
|Q^\partial ({H}_0)| & \text{for ${H}_0$ in chiral classes}\\
 (1-Q^\partial ({H}_0))/2 & \text{for ${H}_0$ in class D}\\   
 1-Q^\partial ({H}_0) & \text{for ${H}_0$ in class DIII}
\end{array}\right. .
\end{equation}
\end{thm}

\begin{proof}
Let \(\mathcal{N}\equiv  
\dim \Ker ({H}_0)\) for short.
We first show that when condition \eqref{condition} is satisfied, the zero-energy eigenspace is stable against small perturbations in the symmetry class. 
Notice that for every symmetry class and for a fixed value of the topological invariant $Q^\partial ({H}_0)$, the dimension $\mathcal{N}$ of the zero-energy
eigenspace is the {\it minimum} allowed for that value of the topological invariant. For example, in the case of chiral classes, we have $\mathcal{N}_{\nu=0} - \mathcal{N}_{\nu=1}
=Q^\partial({H}_0)$. Then $\mathcal{N}=\mathcal{N}_{\nu=0} + \mathcal{N}_{\nu=1} \ge |Q^\partial({H}_0)|$. Because ${W}$ satisfies the symmetries of the class,
the boundary invariant is continuous. Therefore, we have 
$\dim \Ker ({H}_0 + {W}) \ge \dim \Ker {H}_0$.
However, we proved in the previous lemma that for small enough ${W}$, $\dim \Ker ({H}_0 + {W}) \le \dim \Ker {H}_0$.
These two statements lead to the conclusion that
$\dim \Ker ({H}_0 + {W}) = 
\dim \Ker {H}_0$, which is a sufficient 
condition for the stability of $\Ker {H}_0$.

We now provide a proof by contradiction for the necessity of the condition \eqref{condition}. It suffices to construct some perturbation in the class which changes the dimension of the zero-energy eigenspace. This is achieved if we can find a perturbation ${\pert}$
that satisfies the symmetries of the class and has 
non-trivial action in the zero-energy eigenspace,
that is, $P_{\Ker({H}_0)}{W}P_{\Ker({H}_0)} \ne 0$.

\smallskip

\noindent \underline{AIII}: 
Violation of  \eqref{condition}
means that there exist at least two zero-energy eigenstates, say,  
$\{|\eta_+\rangle,|\eta_-\rangle\}$ with chiralities $+1,-1$, respectively.
In this case, the perturbation ${\pert} = |\eta_+\rangle\langle \eta_-| + \text{H.c.}$
satisfies the chiral symmetry, and 
perturbs the two eigenstates away from zero energy, hence reducing the dimension of the zero energy subspace. 

\smallskip

\noindent \underline{BDI}:
The perturbation operator is the same as constructed for class AIII,
that is ${\pert} = |\eta_+\rangle\langle \eta_-| + \text{H.c.}$. Notice that
the extra reality condition is satisfied by ${\pert}$ since the entries of
${H}_0$ and hence of $\{|\eta_+\rangle,|\eta_-\rangle\}$ can be chosen to be real.

\smallskip

\noindent \underline{CII}:
Because of Kramer's degeneracy, we can find at least four zero-energy
eigenstates $\{|\eta_{\nu\sigma}\rangle, \nu=\pm1, \sigma=\pm1\}$ satisfying
$\mathcal{T} |\eta_{\nu\sigma}\rangle = (-1)^\sigma|\eta_{\nu\sigma'}\rangle,
\quad \sigma' \ne \sigma$.
The perturbation operator in this case is
\[
{\pert} = \sum_\sigma |\eta_{\nu\sigma}\rangle\langle\eta_{\nu'\sigma}| + \text{H.c.}.
\]

\noindent \underline{D}:
Violation of the condition \eqref{condition} means that there are at least two 
zero-energy eigenstates $\{|\eta_1\rangle,|\eta_2\rangle\}$
with real entries. The perturbation operator in this case is
${\pert} = i|\eta_1\rangle\langle \eta_2| + \text{H.c.}$.

\smallskip

\noindent \underline{DIII}:
Once again due to Kramer's degeneracy, we have at least four
zero-energy eigenstates $\{|\eta_{\nu\sigma}\rangle, \nu=\pm1, \sigma=\pm1\}$ satisfying
$\mathcal{T} |\eta_{\nu\sigma}\rangle = (-1)^\sigma|\eta_{\nu\sigma'}\rangle,$
$\sigma' \ne \sigma$.
We can use the perturbation 
\[
{\pert} = \sum_{\sigma} |\eta_{\nu\sigma}\rangle\langle \eta_{\nu'\sigma'}|-
|\eta_{\nu\sigma'}\rangle\langle \eta_{\nu'\sigma}| + \text{H.c.}.\]  
\end{proof}

\begin{remark} 
The condition in Eq.\,\eqref{condition} is also a necessary and sufficient condition for protecting the {\em degeneracy} of the many-body ground state, which depends entirely on the number of zero-energy boundary states. 
\end{remark}

For systems with open BCs, we can upper bound the strength of perturbations for which the stable ZMs are not destroyed. We first compute a bound on the generalized condition number ${\cond}$. Computing the condition number for clean systems with open BCs is made possible by the SWH factorization, thanks to the following proposition:

\begin{prop}
Let $\widetilde{H} = \widetilde{H}_+ \widetilde{H}_F \widetilde{H}_+^\dagger$ 
be the SWH
factorization of the Hamiltonian $\widetilde{H}$. Then
its pseudo-inverse is given by
\(
\widetilde{H}^{(-1)} = (\widetilde{H}_+^\dagger)^{-1} \widetilde{H}_F \widetilde{H}_+^{-1}.
\)
\end{prop}
\begin{proof}
Note that all eigenvalues of $\widetilde{H}_F$ belong to the set $\{1,-1,0\}$,
therefore $\widetilde{H}_F^{(-1)}=\widetilde{H}_F$.
Now the statement follows from the fact that $\widetilde{H}_+$ is invertible.
\end{proof}

\noindent
One can now bound the generalized condition number of $\widetilde{H}$ by
\begin{equation}
\label{conditionnumberbound}
{\cond}(\widetilde{H},\widetilde{H}^{(-1)}) \le {\cond}(\widetilde{H}_+)^2. 
\end{equation}
(See Eq.\,\eqref{gencond} for the definition of ${\cond}(\widetilde{H},\widetilde{H}^{(-1)})$ and ${\cond}(\widetilde{H}_+)$.)
A similar analysis for the spectral flattening (SWH factorization) of a Hamiltonian subject to open BCs leads to the inequality
${\cond}(\widetilde{H},\widetilde{H}^{(-1)}) \le {\cond}(\widetilde{H}_+)^2$.

Importantly, the quantity on the right-hand side is easily computable, since \cite{Halmos}
\begin{eqnarray}
\label{eq:bulksensitivity}
\norm{\widetilde{H}_+} = \max_{k\in[0,2\pi)} (\norm{H_+(e^{ik})}),\quad 
\norm{\widetilde{H}_+^{-1}} = \max_{k\in[0,2\pi)} (\norm{H^{-1}_+(e^{ik})}).
\end{eqnarray}
\noindent
The bound on the condition number of Eq.\,(\ref{conditionnumberbound})
immediately leads to the following corollary of Theorem 
\ref{thm:stability}: 

\begin{coro}
Let $H_0$ be a Hamiltonian in one of the non-trivial 
symmetry classes in 1D, subject to open BCs and obeying Eq.\,\eqref{condition}. Let $H_0 = H_{0,+}H_FH_{0,+}^\dagger$ be a SWH factorization. Then $\mathcal{N}(H_0) = \mathcal{N}(H_0+W)$
for any symmetry-preserving perturbation $W$ satisfying 
$\|W\| \le \|H_{0,+}\|^{-2}$.
\end{coro}

We now discuss the stability of ZMs of some representative systems in light of our definition and the criterion established in Theorem \ref{thm:stability}. We have already discussed the chiral ZMs of
the SSH Hamiltonian in Example \ref{exmp:sshzm}. For the parameter regime $t_1<t_2$, we have $Q^\partial(\widetilde{H}) = -1$ and $\mathcal{N} = 1$,
which means the condition for stability in Eq.\,\eqref{condition} is satisfied. Therefore, the edge state of the SSH Hamiltonian subject to open BCs is stable. The same can be said about the Majorana ZMs at the end of Kitaev's Majorana chain, discussed in Ref.\,\cite{Kitaev}.

We emphasize that the condition in Eq.\,\eqref{condition} is {\em not} a consequence of  the non-trivial topology of the Hamiltonian. Rather, it should be regarded as an additional constraint, on top of non-trivial topology, that a system's Hamiltonian must meet in order to host {\it stable} ZMs. This becomes evident from the fact that Eq.\,\eqref{condition}
is explicitly violated by some Hamiltonians, despite the latter having non-trivial topology. For instance, if ${H}_0$ represents a stack of three Kitaev chains in a non-trivial regime subject to open BCs, then the system as a whole hosts three Majorana modes on each end. This system, which belongs to class D, does not satisfy Eq.\,\eqref{condition}. The fact that the zero-energy eigenspace is not stable can be seen from the effect of particle-hole symmetric perturbations. If $\hat{\eta}_1$ and $\hat{\eta}_2$ are two of these Majoranas (the hat indicates operators in Fock space),
then a perturbation of arbitrarily small strength of the form $\widehat{W}=i\lambda\hat{\eta}_1\hat{\eta}_2$ splits the two Majoranas away from zero energy, while still satisfying particle-hole symmetry. 

Our definition of stability also necessarily implies that Andreev bound states on a junction that accidentally appear at zero energy are necessarily unstable. Since Andreev bound states obey fermionic statistics, their anti-excitations must also appear at zero energy. This means that ${H}_0$ in this case will have $\mathcal{N}\ge 2$, which is unstable in class D.

\subsection{Sensitivity}
\label{sec:sensitivity}

Having rigorously characterized stability, we now propose a mathematical definition of the sensitivity of ZMs to symmetry-preserving perturbations. Using the SWH factorization, we also derive an upper bound on the strength of the symmetry-preserving perturbations that are guaranteed not to destroy stable ZMs.
If the zero-energy eigenspace is not stable, then it is
extremely sensitive to some perturbations, in particular, the ones that change its dimension. Therefore, we will only 
consider the zero-energy eigenspaces of 1D SPT phases that
are stable according to Definition \ref{def:stability} given below. 

A measure of ``distortion'' of the zero-energy eigenspace may be naturally built by using the linear-algebraic notion of {\em maximal angle} between linear subspaces \cite{Meyer}. Specifically, by focusing on $\Ker({H}_0 + \lambda{\pert})$
and $\Ker{H}_0$, we define
\[
\sin \theta_{\max} ({\pert},\lambda) \equiv 
\norm{P_{\Ker({H}_0 + \lambda{\pert})}-
P_{\Ker({H}_0)}},\quad
\theta_{\max}({\pert},\lambda) \in [0,\pi/2],
\]
\noindent
This motivated our proposed sensitivity definition:

\begin{defn} 
The ``sensitivity $\mathcal{X}(\pert)$ of $\Ker {H}_0$ to a perturbation $\pert$'', where $\pert$ is an operator of unit norm, is given by
\[
\mathcal{X}(\pert) \equiv  
\frac{d}{d\lambda} \theta_{\max} ({\pert},\lambda) \Big\vert_{\lambda = 0^+},
\]
where ``$+$'' indicates the right derivative at $\lambda =0$. The ``sensitivity $\mathcal{X}$ of $\Ker {H}_0$ to perturbations in the symmetry class of ${H}_0$'' is then defined as
\[
\mathcal{X} \equiv \sup_{||\pert||=1}\mathcal{X}(\pert),
\] 
with $\pert$ being in the symmetry class of ${H}_0$.
\end{defn}

\begin{thm}[Sensitivity to perturbations]
For $\lambda < 1/2{\cond}({H}_0,{H}_0^{(-1)})$, the sine of the maximal angle is upper-bounded by
\begin{equation}
\label{upperbound}
\sin \theta_{\max} ({\pert},\lambda) \le 
\frac{\lambda{\cond}({H}_0,{H}_0^{(-1)}) }
{\sqrt{1 - 2\lambda{\cond}({H}_0,{H}_0^{(-1)})}}.
\end{equation}
The sensitivity to symmetry-preserving perturbations is upper-bounded by
\[
\mathcal{X} \le {\cond}({H}_0,{H}_0^{(-1)}).
\]
\end{thm}
\begin{proof}
It has been shown in chapter 5 of Ref.\,\cite{Meyer} that the maximal angle can be alternatively expressed as
\[
\sin ^2\theta_{\max} ({\pert},\lambda) = 
1-\min_{|\psi\rangle \in \Ker{H}_0}
\langle \psi |P_{\Ker({H}_0 + \lambda{\pert})}|\psi \rangle.
\]
Using Eq.\,\eqref{newkernel}, we can always express any $|\psi\rangle \in \Ker ({H}_0)$ in the form
\begin{equation}
\label{defpsi}
|\psi\rangle = \frac{(\widetilde{I}+
{H}_0^{(-1)}{W})|\phi\rangle}{\sqrt{\langle \phi|
(\widetilde{I}+{H}_0^{(-1)}{W})^\dagger
(\widetilde{I}+{H}_0^{(-1)}{W})|\phi\rangle}} ,
\end{equation}
for some $|\phi\rangle \in \Ker({H}_0 + \lambda{\pert})$.
Further, since $P_{\Ker({H}_0 + \lambda{\pert})} - |\phi\rangle\langle\phi| \ge 0$,
we have
\[
\sin^2 \theta_{\max} ({\pert},\lambda)
 = 1 - \min_{|\psi\rangle \in \Ker{H}_0}
\langle \psi |P_{\Ker({H}_0 + \lambda{\pert})}|\psi \rangle
 \le 1 - \min_{|\psi\rangle \in \Ker{H}_0 }
 |\langle \phi |\psi \rangle|^2.
\]
Using Eq.\,\eqref{defpsi}, we can write
\[
|\langle \phi |\psi \rangle|^2  = \frac{|\langle\phi |(\widetilde{I}+{H}_0^{(-1)}{W})|\phi\rangle|^2}
{\langle \phi|
(\widetilde{I}+{H}_0^{(-1)}{W})^\dagger
(\widetilde{I}+{H}_0^{(-1)}{W})|\phi\rangle}.
\]
After simplification, we can express the right hand-side of the above inequality as
\[
1- |\langle \phi |\psi \rangle|^2 = \lambda^2\left(
\frac{\langle \phi|
({H}_0^{(-1)}{\pert})^\dagger
{H}_0^{(-1)}{\pert}|\phi\rangle - |\langle \phi|
({H}_0^{(-1)}{\pert})^\dagger|\phi\rangle |^2}
{1 +2\lambda{\rm Re}\langle \phi|({H}_0^{(-1)}{\pert})|\phi\rangle 
+ \lambda^2\langle \phi|
({H}_0^{(-1)}{\pert})^\dagger
{H}_0^{(-1)}{\pert}|\phi\rangle}
\right).
\]
Now, the denominator in the above expression can
be bounded from below by
\begin{equation}
\label{denombound}
1 +2\lambda{\rm Re}\langle \phi|({H}_0^{(-1)}{\pert})|\phi\rangle 
+ \lambda^2\langle \phi|
({H}_0^{(-1)}{\pert})^\dagger
{H}_0^{(-1)}{\pert}|\phi\rangle
\ge 
1 +2\lambda{\rm Re}\langle \phi|({H}_0^{(-1)}{\pert})|\phi\rangle
\ge 
1-2\lambda\norm{{H}_0^{(-1)}{\pert}}.
\end{equation}
Similarly, the numerator can be bounded from above by
\begin{equation}
\label{numbound}
\langle \phi|
({H}_0^{(-1)}{\pert})^\dagger
{H}_0^{(-1)}{\pert}|\phi\rangle - |\langle \phi|
({H}_0^{(-1)}{\pert})^\dagger|\phi\rangle |^2
\le
({H}_0^{(-1)}{\pert})^\dagger
{H}_0^{(-1)}{\pert}|\phi\rangle
\le 
\norm{{H}_0^{(-1)}{\pert}}^2.
\end{equation}
Furthermore, 
\begin{equation}
\label{opnormtocond}
\norm{{H}_0^{(-1)}{\pert}} \le \norm{{H}_0^{(-1)}}
\norm{{\pert}} =
{\cond}({H}_0,{H}_0^{(-1)}) \le 1/2\lambda,
\end{equation}
where the last inequality follows from the assumption in the theorem. Finally, using
Equations \ref{denombound}, \ref{numbound}
and \ref{opnormtocond}, we get
\[
1-|\langle \phi |\psi \rangle|^2 \le  \lambda^2
\frac{{\cond}({H}_0,{H}_0^{(-1)})^2}{1 - 2\lambda {\cond}({H}_0,{H}_0^{(-1)})},
\]
which provides the desired bound.
Consequently, the sensitivity $\mathcal{X}$ is bounded by the right derivative of this bound with respect to $\lambda$ at $\lambda = 0$, that yields
\(
\mathcal{X} \le {\cond}({H}_0,{H}_0^{(-1)}),
\)
as claimed.
\end{proof}

We can now use the bound on the generalized condition number 
in Eq.\,\eqref{conditionnumberbound} to derive new bounds on the maximal angle and the sensitivity, namely,
\begin{eqnarray}
\label{upperboundWH}
\sin \theta_{\max} ({\pert},\lambda) \le \frac{\lambda {\cond}(\widetilde{H}_{+})^2}
{\sqrt{1 - 2\lambda {\cond}(\widetilde{H}_{+})^2}},\qquad 
\mathcal{X} \le {\cond}(\widetilde{H}_{+})^2.
\end{eqnarray}
Using Eq.\,\eqref{eq:bulksensitivity}, we also obtain
\begin{equation}
\mathcal{X} \le \left(\max_{k\in[0,2\pi)} (\norm{H_+(e^{ik})})
\max_{k\in[0,2\pi)} (\norm{H_+^{-1}(e^{ik})})\right)^2.
\end{equation}
If we now restrict to Hamiltonians with $\norm{H} \le 1$, we can simplify the above bound to
\begin{equation}
\mathcal{X} \le \left(\max_{k\in[0,2\pi)} (\norm{H_+^{-1}(e^{ik})})\right)^2.
\end{equation}
Since the side factors in the SWH factorization of $\widetilde{H}$ are not necessarily unique, two different valid side factors can lead to two different bounds on the maximal angle and the sensitivity. If $H_k$ is unitary, that is, if 
all energy bands of $H_k$ are flat at energies $\pm 1$, then
we get $\mathcal{X} \le 1$. This means that the ZMs hosted by 
flat-band Hamiltonians are the least sensitive to external perturbations.

\begin{figure}[h]
\centering
\includegraphics[width = 10cm] {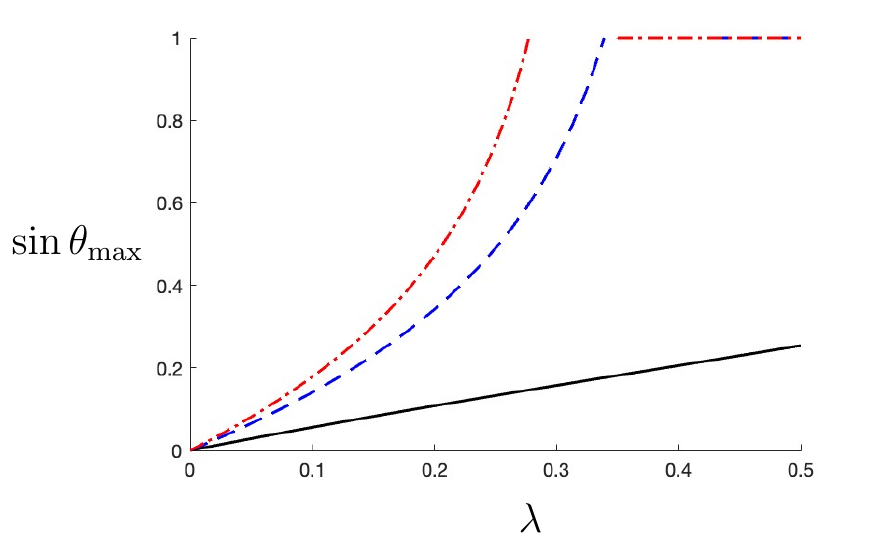}
\caption[Stability of edge states of Su-Schrieffer-Heeger Hamiltonian]{\small The sine of the maximal angle (solid black line), the upper bound in Eq.\,\eqref{upperbound} using full diagonalization (dashed blue line) and the one in Eq.\,\eqref{upperboundWH} using the SWH factorization (dotted red line) vs. perturbation strength $\lambda$ for an arbitrarily chosen ${\pert}$. The SSH Hamiltonian parameters %take values 
are $t_1 = 0.1, t_2 = 1$.
\label{fig:stability}}
\end{figure}

\begin{exmp}
Let us compute a bound on the stability of the ZM of
the SSH Hamiltonian considered in Example \ref{exmp:ssh},
subject to open BCs for $t_1<t_2$, using the SWH 
factorization. Since 
\[
H_+ = \begin{bmatrix}1 & 0 \\ 0 & -(t_1 z + t_2) \end{bmatrix},
\]
we can easily deduce that 
\begin{eqnarray*}
\norm{\widetilde{H}_+} = \norm{H_+(e^{ik})}\vert_{k = 0} =  t_2 + t_1, 
\quad
\norm{\widetilde{H}_+^{-1}} = \norm{H_+^{-1}(e^{ik})}\vert_{k = \pi} =  1/(t_2 - t_1),
\end{eqnarray*}
so that (see also Fig. \ref{fig:stability}))
\[
\mathcal{X} \le {\cond}(\widetilde{H},\widetilde{H}^{(-1)}) \le 
\left(\frac{t_1 + t_2}{t_1 - t_2}\right)^2.
\]
The bound approaches its best possible value $1$ as $t_1 \rightarrow 0$,
which corresponds to the decoupled, flat-band limit.
Near the phase boundary $t_1 \lessapprox t_2$, the value of this bound diverges, indicating instability of the ZM.
\end{exmp}

\section{Conclusions and outlook}
\label{sec:conclusion}

We have investigated the five non-trivial quasi-1D SPT phases of free fermions identified by the tenfold way. 
Our focus was on the impact of BCs on the bulk-boundary correspondence and the expectation that topological ZMs
should be stable in some mathematically precise, quantifiable sense. In an effort to gain concrete, constructive insight into these questions, we avoided the powerful but abstract K theory.
Instead, we used a basic algebraic tool for the analysis of matrix functions of one complex variable: the WH factorization. After adapting the WH factorization to meet the symmetry constraints of the tenfold way in the form of SWH factorizations, we succeeded to 
\begin{itemize}
    \item Prove rigorously a bulk-boundary correspondence with
    basic algebraic tools; and 
    \item Develop a quantitative stability theory for ZMs.
\end{itemize}

Our bulk-boundary correspondence is more basic than the standard bulk-boundary correspondence, in that the BCs are arbitrary (possibly going arbitrarily deep into the bulk, for example), however, the bulk is ultimately disorder-free. Hence, we think of it in physical terms as a bulk-boundary correspondence for the surface physics of very carefully grown and large crystals. We expect that our proof 
will be accessible and illuminating to a large group of physicists who are not familiar with K theory ~\cite{Prodan, alldridge} or operator algebra techniques. 

As for the stability of topologically mandated ZMs, our first task was to identify an appropriate, quantitative notion 
of stability. In our work, the stability is analyzed by characterizing the symmetry-preserving perturbations that induce
a continuous deformation of the space of ZMs. Within this mathematical framework, we established in a rigorous manner that SPT phases, and {\em only} SPT phases, can host stable ZMs. In other words, there is no such a thing as ``non-topological stability'' of ZMs. However, the ZMs of an SPT phase need not be stable without further restrictions. Topology can only guarantee the stability of a specific number of ZMs. To make this point clear, we construct examples of topological bulks that host unstable ZMs. In addition, for a topological Hamiltonian that hosts stable ZMs, we establish an upper bound on the sensitivity of the ZMs to perturbations in terms of a quantity that depends on the bulk Hamiltonian only. 

Our work has two main implications. First, the matrix SWH factorization shows constructively that any SPT phase of free fermions in one dimension can be fully described in terms of a sum of commuting dimer Hamiltonians. Such Hamiltonians emerge naturally in some models; two examples are the sweet spot of the Kitaev chain and the extreme limits of the SSH Hamiltonian. As we showed in this paper, any given SPT Hamiltonian can be mapped to a sum of commuting dimer Hamiltonians by an adiabatic transformation. The bulk-boundary correspondence then follows in a very transparent manner from this adiabatic transformation. Interestingly, a similar principle was recently applied in Ref.\,\cite{zhang20} for constructing models of higher-order topological superconductors, starting from the Kitaev chain as a building block. Second, since our work provides a framework to study quantitatively the stability and the sensitivity of ZMs, it lends itself naturally to tackling the problem of identifying the Hamiltonians with the most robust ZMs. We 
find that SPT Hamiltonians with flat bands host the most stable ZMs. Previous works have asserted that the stability of Majorana ZMs is directly related to their localization length~\cite{boutin18}. Our analysis agrees with the conclusion of Ref.~\cite{boutin18} for local Hamiltonians with flat bands since, in this case, the ZMs are necessarily perfectly localized.

Let us conclude with an outlook on other possible implications of our work. Our bounds on the sensitivity of the ZMs could be immediately useful for assessing their viability in the context of technological applications. For example, as we already mentioned, there are proposals for topological quantum memories that use Majorana ZMs for storing and processing quantum information~\cite{Kitaev,MMilestones}. Likewise, it is possible that our results may lead to improved or new quantum algorithms for simulating topological phases~\cite{huang15}, that rely on the explicit adiabatic transformation between given Hamiltonians, as we have described. 

From a fundamental theory perspective, our work points to several questions for future investigation. On a basic level, it is worth determining what advantages does
the WH factorization offer for the computation of edge modes over alternative, more conventional techniques~\cite{PRB1,istas18}. Another compelling question pertains to whether the fairly simple matrix-function algebra techniques we have used may admit generalizations suitable to study SPT phases in higher dimension. Existing results on the WH factorization of loop groups may provide some hints in this direction~\cite{Pressley86,Giorgadze11,Khimshiashvili07}. Moreover, since the matrix WH factorization technique does not rely on the statistics of fermions, another interesting and timely question is the extent to which this tool be profitably exported to systems of free bosons. Notably, the use of the WH factorization has proved important in establishing no-go theorems for SPT physics in bosonic systems described by a gapped, stable quadratic Hamiltonian~\cite{Xu20}. However, we expect that additional applications will emerge for broader classes of Hamiltonians where stability constraints are relaxed~\cite{Decon}. Even more generally, for both free fermions and bosons, matrix factorization techniques may prove useful to uncover topological physics in dissipative systems, in particular described by quadratic Markovian master equations~\cite{Bardyn,Bosoranas,Barthel}.

\section*{Acknowledgements}

It is a pleasure to thank Vincent Flynn and Mariam Ughrelidze for insightful discussions and a critical reading of the manuscript. A.\,A. gratefully acknowledges insightful discussion with Ilya Spitkovsky and support from a Killam 2020 postdoctoral fellowship. Work at Dartmouth College was partially supported by the US NSF through grants No. PHY-1620541 and No. OIA-1921199.

\appendix

\section{The Wiener-Hopf factorization of matrix Laurent polynomials}
\label{computeWHF}

We provide a constructive proof of Prop.\,\ref{proof:wh} adapted from Ref.\,\cite{gohberg2005convolution}. It amounts to an algorithm for computing the WH factorization of a matrix Laurent polynomial that is invertible on the unit circle.

\begin{prop}[Wiener-Hopf factorization]
Let $A(z,z^{-1})$ denote a matrix Laurent polynomial. If $A$ is invertible on the unit circle, that is, if \(\det A(z,z^{-1})\neq 0\) for all \(z\) with \(|z|=1\), there exists a factorization  of the form
\begin{equation}
A(z,z^{-1}) = A_+(z) D(z,z^{-1}) A_-(z^{-1}) ,
\end{equation}
with the following properties:
\begin{enumerate}
\item
$A_+(z) = \sum_{r= 0}^q  a_{+,r} {z}^\r$ is invertible for all $|z|\le 1$. 
\item
$A_-(z^{-1}) = \sum_{r=p}^0  a_{-,r} {z}^{-|\r|}$ is  is invertible for all $|z^{-1}|\leq 1$. 
\item
$D$ is a diagonal matrix Laurent polynomial of the form
$$D(z,z^{-1})= \sum_{m=1}^{d}z^{\ind_m}|m\rangle\langle m|,$$ where the integers
$\{\ind_m\, |\,  m=1,\dots,d\}$ are reverse-ordered as $\ind_1 \ge \ind_2 \ge \dots \ge \ind_d$.
\end{enumerate}
\end{prop}

\begin{proof}
If $p<0$, then one can associate to \(A\)  an invertible matrix polynomial  as
\(
A_{0}(z) = z^{-p}A(z) = \sum_{\r=0}^{q-p}a_{\r}z^{\r}.
\)
A WH factorization for $A_0 = A_+ D_0 A_-$ yields one for $A = A_+ D A_-$ with 
$D(z,z^{-1}) \equiv z^p D_0(z,z^{-1})\). Hence, \(A\) stands for a matrix polynomial in the following. Let 
\(z_\i,\ \i=1,\dots,\n\) denote the zeroes of \(\det A(z,z^{-1})\) 
inside the unit circle, listed in some fixed but arbitrary order and with multiplicity. We will show that there is an associated sequence of  factorizations of $A$ of the form $A= A_+^{(\i)}D^{(\i)}A_-^{(\i)},\ \i=0,\dots, \n,$  and ending at \(\i=0\)
with a WH factorization of \(A\). 
In order to construct the \(\i-1\) factorization out of the \(\i\)th factorization we will need to break $A_+^{(\i)}$ into column vectors polynomials. Our notation is   
\[
A_+^{(\i)}(z)=
\sum_{n=1}^d |A^{(\i)}_{+,n}(z)\rangle\langle n|, \quad 
|A^{(\i)}_{+,n}(z)\rangle\equiv \sum_{m=1}^d A_{+,mn}^{(\i)}(z)|m\rangle.
\]

\begin{enumerate}
\item \label{step1} 
\textit{Launch the sequence of  factorizations.---}
If the are are no roots of \(det A\) inside the unit circle, that is, if \(\n=0\), then  \(A_+=A\) and \(D=A_-=\mathbb{1}\) is a WH factorization of \(A\) and
we are done. Otherwise,  set $A_+^{(\n)} = A$ and $D^{(\n)}=A_-^{(\n)} = \mathbb{1}$.   

\item \label{step2} 
\textit{Remove the zero \(z_i\) from \(A_+^{(\i)}\).---}
Because $\det A_+^{(\i)}(z_\i) = 0$, there exists a smallest \(s\) such that the \(s\)th column 
$|A_{+,s}^{(\i)}(z_\i)\rangle$ 
 is linearly dependent on the previous \(s-1\) columns. If it does not vanish identically, then 
 $|A_{+,s}^{(\i)}(z_\i)\rangle=\sum_{m=1}^{s-1}\alpha_{m}|A_{+,m}^{(\i)}(z_\i)\rangle$. If it vanishes identically, use the formulas below with the \(\alpha_m=0,\ m=1,\dots, s-1\).
In either case, there is a zero at \(z_i\) of the  vector polynomial $A_{+,s}^{(\i)}(z)-\sum_{m=1}^{s-1}\alpha_{m}A_{+,m}^{(\i)}(z)$. 
 We can use this information to define matrix meromorphic matrix function
\begin{align*}
\qquad U_{\i}^{-1}(z) &\equiv Q_s + \frac{1}{(z-z_\i)}P_s
- \sum_{m=1}^{s-1}\frac{\alpha_{m}}{(z-z_\i)}|m\rangle\langle s|,
\end{align*}
in terms of the orthogonal projectors \(P_s=|s\rangle\langle s|\) and its complement \(Q_s=\mathbb{1}-P_s\).
It is invertible on the unit circle because \(|z_i|<1\) by assumption. Its inverse is
\begin{align*}
U_{\i}(z) &\equiv Q_s + (z-z_\i)P_s
+ \sum_{m=1}^{s-1}\alpha_{m}|m\rangle\langle s|.
\end{align*}
Post-multiplying \(A_+^{(\i)}\) by \(U_\i^{-1}\) carries out the column substitution
\[
A_+^{(\i)}(z)U_{\i}^{-1}(z)=A_+^{(\i)}-|A_{+,s}^{(i)}\rangle\langle s|+
\frac{|A_{+,s}^{(\i)}(z)\rangle\langle s| - \sum_{m=1}^{s-1}\alpha_{m}|A_{+,m}^{(\i)}(z)\rangle\langle s|}{z-z_\i}
\]
and removes exactly one zero inside the unit circle from $A_+^{(\i)}$. 

\item \label{step5}
\textit{Add a zero at \(z=0\) to \(D^{(\i)}\) and reshuffle to obtain \(D^{(\i-1)}\).---}
Let's add a zero at zero to \(D^{(\i)}\) to compensate for the one tha we took away from \(A_+^{(\i)}\). In matrix form, \(D^{(\i)}\mapsto 
D^{(\i)}(Q_s+zP_s)\).
Let $\{\ind_m^{(\i-1)},\ m=1,\dots,d\}$ 
as the non-increasing reordering of the non-negative integers
associated to the powers $\{z^{\ind_1^{(\i)}},\dots,z^{\ind_s^{(\i)}}+1,
\dots,z^{\ind_d^{(\i)}}\}$ on the diagonal of \(D^{(\i)}(Q_s+zP_s)\).
The reordering only takes one exchange of two entries. Let $\Pi_\i=\Pi_\i^{-1}$ denote the matrix that performs the corresponding reordering of the diagonal entries of \(D^{(\i)}(Q_s+zP_s)\). Then,  
\[
D^{(\i-1)} \equiv \Pi_i D^{(\i)}(Q_s+zP_s)\Pi_i. 
\]

\item \textit{Update \(A_{\pm}^{(i)}\mapsto A_{\pm}^{(i-1)} \).}---
Having updated the middle factor, the final step is to update \(A_{\pm}^{(i)}\mapsto A_{\pm}^{(i-1)} \). As for \(A_+^{(i)}\), the only task is to shuffle its columns to account for the relocation of the new zero in the middle factor. Hence,
\[
A_+^{(\i-1)}\equiv A_+^{(\i)}U_\i^{-1}\Pi_\i.
\]
To update \(A_-^{(i)}\), define a matrix  polynomial \(V_{\i}(z^{-1})\) such that $U^{(\i)}D^{(\i)} = \Pi_i D^{(\i-1)}\Pi_iV_{\i}$, that is,  
\begin{align*}
V_{\i}(z) &=(Q_s+z^{-1}P_s)D_{\i}^{-1}U_\i D_\i
= I  - z_\i z^{-1}|s\rangle\langle s|
+ \sum_{m=1}^{s-1}\alpha_{m}z^{\ind_s^{(\i)}-\ind_m^{(\i)}}|m\rangle\langle s|
\end{align*}
and let 
\[A_-^{(\i-1)} \equiv \Pi_iV_\i A_-^{(\i)}.\]
\end{enumerate}
By construction, \(A=A_+^{(\i)}D^{(i)}A_-^{(\i)}.\)
Moreover, after \(n\) steps one obtains a matrix polynomial
\(A_+^{(0)}(z)\) with no roots in the closed unit disk and so invertible there. Similarly, since  $\det V_{\i}(z^{-1}) = 1-z_{\i} z^{-1}$, \(V_\i(z^{-1})\) is invertible provided \(|z^{-1}|\leq 1\) and so is \(A_-^{(0)}(z^{-1})\). Finally, 
\(D^{(0)}\) is of the form required by the WH factorization as well.
\end{proof}

It is apparent from the proof why the middle factor is  unique. The non-uniqueness of the side factors can be traced to the fact that we were forced to order the roots \(z_\i\) in some fixed but ultimately arbitrary way. Different ordering will, in general, result in 
different side factors.

\section{The Wiener-Hopf factorization of banded block-Toeplitz operators}
\label{app:whproof}

Here we will construct the WH factorization of a Fredholm BBT operator directly, without reference to the functional calculus of Sec.\,\ref{func}, and then use 
it to derive the WH factorization of a matrix Laurent polynomial invertible on the unit circle. 
This proof of the WH factorization will help us in \ref{app:symwhproofs} to 
prove SWH factorizations under additional symmetry constraints.

Let \(\widetilde{A}=\sum_{r=p}^qT^r\otimes a_r\) denote a fixed but arbitrary BBT operator 
that is also a Fredholm operator. The associated family of BBT operators 
$\mathcal{A} = \{\widetilde{A}_\ind,\ \ind \in \mathbb{Z}\}$ 
is defined as follows: 
\begin{enumerate}
    \item 
 \(\widetilde{A}_{-p}\equiv\widetilde{A}\), and 
 \item 
\(
\widetilde{A}_{\ind+1} = \widetilde{A}_{\ind}T^\dagger\), \ \(
\widetilde{A}_{\ind-1}=T\widetilde{A}_\ind  .
\)
\end{enumerate}
The family \(\mathcal{A}\) is in fact a family of Fredholm operators because $T$, $T^\dagger$, and \(\widetilde{A}\)  are Fredholm and the product of two Fredholm operators is again Fredholm.

\begin{prop}
\label{lem:l1}
The kernel of the \(\widetilde{A}_\ind\) consists of absolutely summable sequences, that is, 
$\Ker\,\widetilde{A}_\ind \subseteq \ell^1(\mathbb{N})\otimes \mathbb{C}^d\subset \ell^{2}(\mathbb{N})\otimes \mathbb{C}^d$ as well.
\end{prop}
\begin{proof}
The proof of this lemma relies on Theorem A.5 in the Appendix of Ref.\,\cite{JPA}.
The key idea is that any vector in the Kernel of a banded block-Toeplitz operator
can be expressed as a sum of finitely many ``bulk solution'' vectors. 
Some of these bulk solution vectors have finite support, and therefore belong to
$\ell^1(\mathbb{N})\otimes \mathbb{C}^d$. Each of the remaining 
bulk solution vectors is expressible as 
$|z,\nu\rangle|u\rangle \in \ell^1(\mathbb{N})\otimes \mathbb{C}^d$, with 
\[
|z,\nu\rangle = \frac{\partial^{\nu-1}}{\partial z^{\nu-1}} \sum_{j \in \mathbb{N}}z^j|j\rangle ,
\]
for some $z \in \mathbb{C}, |z|<1$ and $\nu \in \mathbb{N}$.
It is a straightforward exercise in the geometric series to show that $|z,\nu\rangle \in \ell^1(\mathbb{N})$, 
so that all bulk solution vectors belong to $\ell^1(\mathbb{N}) \otimes \mathbb{C}^d$. \end{proof}

The kernels of the operators in $\mathcal{A}$ 
are nested, in a sense that is made precise as follows:
\begin{prop}
\label{lem:inclusions}\ 
\begin{enumerate}
\item $\Ker\,\widetilde{A}_\ind \subseteq \Ker\,\widetilde{A}_{\ind-1}$
\item $T^\dagger \Ker\,\widetilde{A}_\ind \subseteq \Ker\,\widetilde{A}_{\ind-1}$
\item The inclusions in 1. or 2. are in fact equalities if and only if $\Ker\,\widetilde{A}_{\ind}$ is zero-dimensional, that is, $\Ker\,\widetilde{A}_{\ind}=\{0\}$ is trivial.
\end{enumerate}

\end{prop}

\begin{proof}
1. holds because $T \widetilde{A}_\ind = \widetilde{A}_{\ind-1}$
while 2. follows from $ \widetilde{A}_\ind = \widetilde{A}_{\ind-1}T^\dagger$. As for
3., if an strict equality holds for either inclusion, then 
\[
\dim \Ker\,\widetilde{A}_\ind = \dim  \Ker\,\widetilde{A}_{\ind-1} =
 \dim T^\dagger \Ker\,\widetilde{A}_{\ind}
\]
(recall that $T^\dagger$ is an isometry), and  so
\[
T^\dagger \Ker\,\widetilde{A}_\ind=\Ker\,\widetilde{A}_\ind ,
\] 
because both spaces coincide with  $\Ker\,\widetilde{A}_{\ind-1}$. Finally,
$\Ker\,\widetilde{A}_\ind$ is finite-dimensional and $T^\dagger$ has no eigenvectors.
Therefore, $\Ker\,\widetilde{A}_\ind$ must be trivial.  
\end{proof}

\begin{figure}
\begin{center}
\includegraphics[width=10cm]{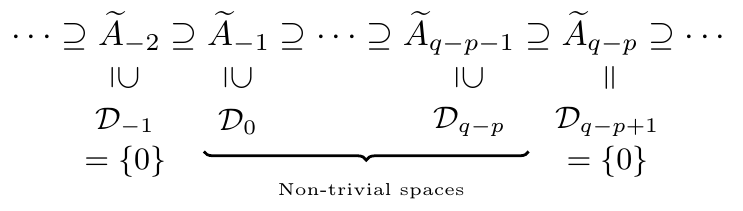}
\end{center}
\vspace*{-5mm}
\caption{ 
The subspace \({\cal D}_\kappa\) of \(\text{Ker}\,\widetilde{A}_{k-1}\) is implicitly defined in Eq.\,\eqref{nesting}.  \(\text{Ker}\,\widetilde{A}_{k-1}\) is partially determined by  
\(\text{Ker}\,\widetilde{A}_{k}\)
 according to the formula \(\text{Ker}\,\widetilde{A}_{\ind-1} \supseteq \text{Ker}\,\widetilde{A}_{\ind}\vee T^\dagger\text{Ker}\,\widetilde{A}_{\ind}\). The difference between these two spaces is quantified by the space \({\cal D}_\kappa\). It is
 automatically trivial for \(\kappa<0\) or \(\kappa>q-p+1\). }
\end{figure}

We will use the notation
\[
V\vee W\equiv \text{Span}\,V\cup W
\]
for \(V, W\) finite-dimensional subspaces of the underlying Hilbert space. For future reference let us point out that, by definition, the span of any set of vectors, finite or infinite, consists only of finite linear combinations of said vectors.  If \(V\cap W=\{0\}\), then \(V\vee W=V\bigoplus W\). The canonical convention for the span of the empty set is that it is the trivial vector space. 

\begin{coro}
\label{coro:inc}
For each $\ind \in \mathbb{Z}$ there exists a finite-dimensional space 
\({\cal D}_\ind\) such that 
\begin{equation}
\label{nesting}
\text{Ker}\,\widetilde{A}_{\ind-1} = \text{Ker}\,\widetilde{A}_{\ind}\vee T^\dagger\text{Ker}\,\widetilde{A}_{\ind} \oplus 
{\cal D}_\ind.
\end{equation}
\end{coro}
\noindent
It will be useful in the following to declare some fixed but arbitrary basis of \({\cal D}_\ind\), so let
\(s_\ind=\dim{\cal D}_\ind, 
\)
and
\[
\text{Span}\,\{|\psi_{\ind,1}\rangle,\dots,|\psi_{\ind,s_\ind}\rangle\}= {\cal D}_\ind.\]

\begin{prop}
\label{lem:kernel}
\label{lem:basis}
\label{lem:relabel}   \
\begin{enumerate}
\item $\Ker\, \widetilde{A}_\ind =\{0\}$ for \(\ind \geq q-p\)
\item If \(\kappa< q-p\), then \(\text{Ker}\,\widetilde{A}_{\ind} = \bigoplus_{\ind'=\ind+1}^{q-p}
\bigoplus_{r=0}^{\ind'-\ind-1} T^{\dagger\, r}{\cal D}_{\ind'}\). 

\item There is a \(\ind\) such that \({\cal D}_{\ind}\neq \{0\}\).

\item 
$\operatorname{Closure}\,\operatorname{Span}\,\{T^{\dagger r}|\psi_{\ind' s}\rangle\,|\, \ind'\leq q-p,\ s = 1,\dots,s_{\ind'},\ r\in\mathbb{N}\}
= \ell^2(\mathbb{N})\otimes \mathbb{C}^d$.

\item 
The integer $s_\ind=\text{dim}\,{\cal D}_\k$
vanishes if $\ind >q-p$ or $\ind<0$. 

\item 
\(
\sum_{\ind =0}^{q-p} s_\ind = d.
\)
\end{enumerate}
\end{prop}

\begin{proof}
1. It is enough to show that \(\Ker\, \widetilde{A}_{q-p}=\{0\}\) since the kernels are nested downwards.
Because $\widetilde{A}_{q-p}$ is a function of $T^\dagger$ only, 
if $|\psi\rangle \in \text{Ker}\,\widetilde{A}_{q-p}$, then 
\[
\widetilde{A}_{q-p} T^\dagger |\psi\rangle = T^\dagger \widetilde{A}_{q-p}|\psi\rangle = 0.
\]
This would imply that the kernel of $\widetilde{A}_{q-p}$ contains infinitely many
linearly indpendent vectors $|\psi\rangle, T^\dagger |\psi\rangle, (T^\dagger)^2 |\psi\rangle,\dots$.
However, $\widetilde{A}_{q-p}$ is a Fredholm operator and so its kernel must be
finite-dimensional. Hence, \(|\psi\rangle=0\).

2. This statement is equivalent to proving that 
$$\text{Ker}\,\widetilde{A}_{\ind} = \bigvee_{\ind'=\ind+1}^{q-p}%\bigvee
\bigvee_{r=0}^{\ind'-\ind-1} T^{\dagger\, r}{\cal D}_{\ind'}$$ and that 
the vector spaces $\{T^{\dagger\, r}{\cal D}_{\ind'}, r=0,\dots,\ind'-\ind-1,
\ind'=\ind+1, \dots, q-p\}$ are linearly independent. We now prove these two statements separately.

First, at this point we know that  
\begin{align*}
  &\ \  \vdots\\
  \text{Ker}\,\widetilde{A}_{q-p}&=\{0\}, \\
  \text{Ker}\,\widetilde{A}_{q-p-1}&={\cal D}_{q-p},\\
  \text{Ker}\,\widetilde{A}_{q-p-2}&={\cal D}_{q-p}\vee T^\dagger {\cal D}_{q-p}\bigoplus {\cal D}_{q-p-1},\\
  \text{Ker}\,\widetilde{A}_{q-p-3}&={\cal D}_{q-p}\vee T^\dagger {\cal D}_{q-p}\vee T^{\dagger\,2}{\cal D}_{q-p}\vee {\cal D}_{q-p-1}\vee T^{\dagger\, 2}{\cal D}_{q-p-1}\bigoplus{\cal D}_{q-p-2}\\
  &\ \  \vdots
\end{align*}
In the process of  rearranging and simplifying the formula for \(\text{Ker}\,\widetilde{A}_\ind\) as \(\ind\) goes down,  one is forced to weaken most of the \(\bigoplus\) into \(\vee\). Now one can guess the formula  \(\text{Ker}\,\widetilde{A}_{\ind} = \bigvee_{\ind'=\ind+1}^{q-p}%\bigvee
\bigvee_{r=0}^{\ind'-\ind-1} T^{\dagger\, r}{\cal D}_{\ind'}\) for \(\ind<q-p\) and complete its proof by an induction argument.
 Let \(\gamma\) denote the largest integer such that \(\text{Ker}\,\widetilde{A}_{\gamma}\neq\{0\}\) (we know from the previous Lemma that \(\gamma<q-p\)). Then,
\[
\text{Ker}\,\widetilde{A}_{\gamma}=\text{Ker}\,\widetilde{A}_{\gamma+1}\vee T^\dagger \text{Ker}\,\widetilde{A}_{\gamma+1}
\oplus{\cal D}_{\gamma+1}=\bigvee_{\ind'>\gamma}\bigvee_{r=0}^{\ind'-\gamma-1}T^{\dagger\,r}{\cal D}_{\ind'} ,
\]
since \(\text{Ker}\,\widetilde{A}_{\gamma+1}=0\) and \({\cal D}_{\ind'}=\{0\}\) if \(\ind'>\gamma+1\). Now suppose that the induction hypothesis is true for some \( \delta \leq \gamma\). Then, 
\begin{align*}
\text{Ker}\,\widetilde{A}_{\delta-1} =
\text{Ker}\,\widetilde{A}_{\delta}\vee T^\dagger \text{Ker}\,\widetilde{A}_{\delta}\oplus{\cal D}_{\delta}=
\bigvee_{\ind'>\delta}\bigvee_{r=0}^{\ind'-\delta}T^{\dagger\,r}{\cal D}_{\ind'}\oplus{\cal D}_{\delta}= 
\bigvee_{\ind'>\delta-1}\bigvee_{r=0}^{\ind'-(\delta-1)-1}T^{\dagger\,r}{\cal D}_{\ind'},
\end{align*}
and the induction argument is complete.

To complete the proof we need to show that 
the vectors spaces $\{T^{\dagger\, r}{\cal D}_{\ind'}, r=0,\dots,\ind'-\ind-1, \ind'=\ind+1, \dots, q-p\}$ are linearly independent.
We prove this by induction. The statement is true for $\ind \ge q-p$, as both the left- and the right-hand sides are trivial vector spaces. We assume that the statement is true for some integer $\ind<q-p$, that is,
the vector spaces $\{T^{\dagger\, r}{\cal D}_{\ind'}, r=0,\dots,\ind'-\ind-1, \ind'=\ind+1, \dots, q-p\}$ are linearly independent for such a $\ind$. To complete the proof by induction, we need to show that $\{T^{\dagger\, r}{\cal D}_{\ind'}, r=0,\dots,\ind'-\ind-1, \ind'=\ind+1, \dots, q-p\}$ are linearly independent for $\ind-1$.

By construction in Corollary \ref{coro:inc}, the space ${\cal D}_{\ind-1}$ is linearly independent
from the span of $\{T^{\dagger\, r}{\cal D}_{\ind'}, r=0,\dots,\ind'-\ind, \ind'=\ind+1, \dots, q-p\}$, and so it suffices to prove linear independence of the latter set of vector spaces. Suppose there exists a 
vanishing linear combination of vectors from these spaces, so that
\[
\sum_{\ind'=\ind+1}^{q-p}\sum_{r=0}^{\ind'-\ind}\sum_{s=1}^{s_{\ind'}}
\alpha_{rs\ind'}(T^\dagger)^r|\psi_{\ind' s}\rangle = 0.
\]
By separating this summation into two parts and transferring the second part to the right-hand side,
we get
\begin{eqnarray*}
\sum_{\ind'=\ind+1}^{q-p}\sum_{r=0}^{\ind'-\ind-1}\sum_{s=1}^{s_{\ind'}}
\alpha_{rs\ind'}(T^\dagger)^r|\psi_{\ind' s}\rangle &=& 
-\sum_{\ind'=\ind+1}^{q-p}\sum_{s=1}^{s_{\ind'}}
\alpha_{rs\ind'}(T^\dagger)^{\ind'-\ind}|\psi_{\ind' s}\rangle\\
&=& -T^\dagger \sum_{\ind'=\ind+1}^{q-p}\sum_{s=1}^{s_{\ind'}}
\alpha_{rs\ind'}(T^\dagger)^{\ind'-\ind-1}|\psi_{\ind' s}\rangle.
\end{eqnarray*}
This is an equation of the form $|\Psi\rangle = -T^\dagger |\Phi\rangle$, and 
$|\Psi\rangle \in \text{Ker}\,\widetilde{A}_{\kappa}$. Therefore, we find that 
$T^\dagger |\Phi\rangle \in \text{Ker}\,\widetilde{A}_{\kappa}$, and consequently
$|\Phi\rangle \in \text{Ker}\,\widetilde{A}_{\kappa+1}$. However, by the induction hypothesis 
for $\kappa$, the terms 
in the expression for $|\Phi\rangle$ are linearly independent from the 
$\text{Ker}\,\widetilde{A}_{\kappa+1}$, as the latter is spanned by 
$\{T^{\dagger\, r}{\cal D}_{\ind'}, r=0,\dots,\ind'-\ind-1,
\ind'=\ind+2, \dots, q-p\}$. Therefore, we find $|\Phi\rangle = 0$, and
consequently $|\Psi\rangle = 0$. We have thus proven the correctness of the statement for $\kappa-1$,
completing the proof.

3. Suppose otherwise. Then, the kernel of any member of \({\cal A}\) is trivial  by 2. above. However, \(\widetilde{A}_{-1}\)
necessarily annihilates the states \(|0\rangle\otimes |m\rangle\). To avoid this contradiction, it must be that \({\cal D}_{\ind}\neq \{0\}\) for some \(\ind\).

4. Let 
\({\cal B}_{C}=\{|j\rangle|m\rangle\,|\, r\in\mathbb{N},\ m=1,\dots, d\}\)
denote the canonical
orthonormal basis of \(\ell^{2}(\mathbb{N})\otimes \mathbb{C}^d\). Then, since \(\widetilde{A}_{-j-1}\) annihilates the \(d\) states \(|j\rangle|m\rangle\), one has that \({\cal B}_C\subset \text{Span}\,{\cal B}_{\widetilde{A}}\). 
The family \(\{T^m\}_{m\in \mathbb{Z}}\) shows that the sets ${\cal B}_C$ and ${\cal B}_{\widetilde{A}}$ might coincide but they do not in general. If \({\cal B}_{C}\) is strictly contained in \({\cal B}_{\widetilde{A}}\),  the vectors in latter basis are not mutually orthogonal. In any case,
\[
\text{Closure}\,\text{Span}\,{\cal B}_{\widetilde{A}}=\text{Closure}\,\text{Span}\,{\cal B}_{C}=\ell^2(\mathbb{N})\otimes \mathbb{C}^d.
\]

5. Since $\Ker\,\widetilde{A}_{\ind}$ is trivial for \(\ind \geq q-p\), it follows that  $s_\ind = 0$ for all $\ind > q-p$. 
For $\ind < 0$, $\widetilde{A}_{\ind}$ is a function of $T$ only, so that 
$\widetilde{A}_{\ind-1} = T\widetilde{A}_{\ind} = \widetilde{A}_{\ind}T$.
Therefore, if \(\kappa<0\), 
\begin{align*}
\Ker \widetilde{A}_{\ind-1} &= T^\dagger \Ker \widetilde{A}_{\ind} \oplus 
\text{Span}\,\{|j=0\rangle|m\rangle,\ m=1,\dots,d\}.
\end{align*}
In addition, for $\kappa <0$, the smallest power of $T$ in the decomposition of 
$\widetilde{A}_{\ind}$ is non-negative, so that
\begin{align*}
\Ker \widetilde{A}_{\ind} &\supseteq \text{Span}\,\{|j=0\rangle|m\rangle,\ m=1,\dots,d\,\}.
\end{align*}
Therefore, 
\[
\text{Ker}\,\widetilde{A}_{\ind-1} = \text{Ker}\,\widetilde{A}_{\ind}\vee
T^\dagger\text{Ker}\,\widetilde{A}_{\ind}
\]
and so 
$s_\ind =0$ as well if $\ind<0$. 

6. It follows from 2. and 4. above that  
\begin{equation}
\label{eqproof3}
\dim \Ker \widetilde{A}_{\ind} = \sum_{\ind'=k+1}^{q-p} (\ind'-\ind)s_{\ind'} ,
\end{equation}
with the canonical convention that a sum over an empty set vanishes. In addition, as we saw in 5. above that, for \(\ind<0\) 
\begin{equation}
\label{eqproof2}
\dim \Ker \widetilde{A}_{\ind-1} = \dim \Ker \widetilde{A}_{\ind} + d \qquad (\ind < 0).
\end{equation}
Hence, for any $\ind < 0$, we have from Eq.\,\eqref{eqproof3} and 5. above that
\[
d=\dim \Ker \widetilde{A}_{\ind-1} - \dim \Ker \widetilde{A}_{\ind} = 
\sum_{\ind'=\ind}^{q-p}s_{\ind'} = \sum_{\ind'=0}^{q-p}s_{\ind'}.
\]
\end{proof}

Let us relabel the vectors 
$\{|\psi_{\ind s}\rangle\,|\ \ind \in \mathbb{Z}, \ s=1,\dots,s_\ind\}$ 
as $\{|\psi_m\rangle\,|\ m=1,\dots,d\}$. 
This is possible since Lemma\,\ref{lem:relabel} guarantees that there are exactly 
$d$ such vectors. We denote the bijective relabeling map by $(\ind,s) \mapsto m(\ind,s)$
and we choose a relabeling such that 
\[
\ind > \ind' \implies m(\ind,s) < m(\ind',s).
\]
The inverse of this map, upon ignoring the number $s$, gives rise to the
(not necessarily bijective)
{\it partial index function}
\(
m \mapsto (\ind,s) \mapsto \ind.
\)

Up to this point we have managed to proceed essentially by algebraic means. In the following, we will need to appeal to some basic facts of functional analysis. 

\begin{prop}
The mapping
\(T^{\dagger\,j} |\psi_{m}\rangle \mapsto |j\rangle |m\rangle\)
induces a bounded, invertible, and lower-triangular BBT  operator \(\widetilde{A}_-\). The inverse mapping \(|j\rangle |m\rangle\mapsto T^{\dagger\,j} |\psi_{m}\rangle\)
also induces a bounded operator; the latter is the inverse of \(\widetilde{A}_-\).
\label{propAminus}
\end{prop}

\begin{proof}
Consider the operator $\widetilde{A}_-$, defined over the (non-closed) linear space 
\(\text{Span}\,\{T^{\dagger j}|\psi_m\rangle\}_{r=0}^{\infty}\),
obtained from linearly extending the correspondence 
\(T^{\dagger\,j} |\psi_{m}\rangle \mapsto |j\rangle |m\rangle\) to other vectors.
Observe that 
\[
T\widetilde{A}_-T^\dagger (T^\dagger)^j|\psi_{m}\rangle = T\widetilde{A}_- (T^\dagger)^{j+1}|\psi_{m}\rangle
 = T|j+1\rangle|m\rangle = |j\rangle|m\rangle = \widetilde{A}_-(T^\dagger)^j|\psi_{m}\rangle.
\]
This proves the identity $T\widetilde{A}_-T^\dagger = \widetilde{A}_-$, which is 
necessary and sufficient condition for $\widetilde{A}_-$ to be block-Toeplitz. To prove that $\widetilde{A}_-$ is lower-triangular, we show that $[\widetilde{A}_-,T^\dagger] = 0$.
This follows from the simple calculation
\[
\widetilde{A}_-T^\dagger(T^\dagger)^j|\psi_{m}\rangle = |j+1\rangle |m\rangle
= T^\dagger |j\rangle|m\rangle = 
T^\dagger\widetilde{A}_-(T^\dagger)^j|\psi_{m}\rangle.
\]
For bandedness, it suffices to prove that $[T^{q-p}\widetilde{A}_- , T]=0$.
We have
\[
T (T^{q-p}\widetilde{A}_- ) (T^\dagger)^j|\psi_{m}\rangle = 
\left\{\begin{array}{lcl}|j-q+p-1\rangle|m\rangle &
\text{if} & j > q-p\\
0 & \text{if} & j \le q-p \end{array}\right..
\]
For the other term in the commutator, we get
\[
 (T^{q-p}\widetilde{A}_- )T (T^\dagger)^j|\psi_{m}\rangle = 
\left\{\begin{array}{lcl}|j-q+p-1\rangle|m\rangle &
\text{if} & j > q-p\\
0 & \text{if} & 1 \le j \le q-p \end{array}\right..
\]
This calculation does not go through for $j=0$. However, since 
$[\widetilde{A}_0, T] = 0$ and $\ind_m \ge 0 \  \forall m$, we have 
\[
\widetilde{A}_0 T|\psi_m\rangle = T\widetilde{A}_0 |\psi_m\rangle = 0.
\]
Then by Lemma\,\ref{lem:kernel}, we can express $T|\psi_m\rangle$ as
\[
T|\psi_m\rangle = \sum_{r,\ind,s} \alpha_{r\ind s}(T^\dagger)^r|\psi_{m(\ind,s)}\rangle , 
\]
for some complex numbers $\{\alpha_{r\ind s}\}$. In the above summation, $r<\ind \le q-p$. Using this expression, we get
\begin{align*}
 (T^{q-p}\widetilde{A}_- )T|\psi_{m}\rangle =
 \sum_{r,\ind,s} \alpha_{r\ind s} T^{q-p}\widetilde{A}_-  (T^\dagger)^r|\psi_{m(\ind,s)}\rangle 
 =
 \sum_{r,\ind,s} \alpha_{r\ind s} T^{q-p} |j=r\rangle|\psi_{m(\ind,s)}\rangle = 0.
 \end{align*}
Because $\widetilde{A}_{-} $ is banded block-Toeplitz, therefore it is bounded. Since its domain is dense in  $\ell_2(\mathbb{N})\otimes \mathbb{C}^d$
it can be uniquely extended to this space by closed subgraph theorem.

We now prove that $\widetilde{A}_{-}$ is invertible. 
Let $\widetilde{B}$
denote the linear transformation induced on the span of the \(|j\rangle|m\rangle\) by the mapping

\(|j\rangle|m\rangle \mapsto (T^\dagger)^j|\psi_m\rangle.\) 
Our strategy is to show that $\widetilde{B}$ is bounded and therefore
extendible to $\ell_2(\mathbb{N})\otimes\mathbb{C}^d$. Then, since \(\widetilde{B}\) undoes the action of \(\widetilde{A}_-\) on a dense domain, one can conclude that \(\widetilde{B}=\widetilde{A}_-^{-1}\). 

To show that $\widetilde{B}$ is bounded, we first 
consider the case $d=1$ and define $\varrho:\ell_2(\mathbb{N}) \to H_2(\mathbb{T})$, which 
maps vectors in $\ell_2(\mathbb{N})$ to functions in the Hardy space $H_2(\mathbb{T})$ according to $|\alpha\rangle \mapsto \alpha(z)$, with
\begin{equation}
|\alpha\rangle = \sum_j \alpha_j |j\rangle, \;\;  \alpha(z) = \sum_j \alpha_j z^j.
\end{equation}
It is easy to check that 
\begin{equation}
\varrho(\sum_j \alpha_j (T^\dagger)^j|\psi_1\rangle) =  
\varrho(\sum_j \alpha_j|j\rangle)\varrho(|\psi_1\rangle) = \alpha(z)\psi(z).
\end{equation}
Lemma \ref{lem:l1} states that $|\psi_1\rangle \in \ell^1(\mathbb{N})$ 
being the kernel vector of a banded block-Toeplitz operator. Therefore, 
$\psi_1(z)$ has an absolutely convergent Fourier series, and so
$\psi_1(z)$ is well-defined at each point on the unit circle. 
The norm of the vector $\sum_j \alpha_j (T^\dagger)^j|\psi_1\rangle$ 
can then be bounded by
\begin{equation}
\norm{\sum_j \alpha_j (T^\dagger)^j|\psi_1\rangle}^2 = 
\int_{\mathbb{T}} dz |\alpha(z)\psi_1(z)|^2
\le \text{max}_{z \in \mathbb{T}}(|\psi_1(z)|^2) \int dz |\alpha(z)|^2.
\end{equation}
Therefore, $\norm{B|\alpha\rangle}/\norm{|\alpha\rangle} \le 
\text{max}_{z \in \mathbb{T}}|\psi_1(z)|$,
hence $\widetilde{B}$ is bounded. The proof can be extended to $d>1$ with only minor modifications.
\end{proof}

\begin{coro}
\label{coro:linind}
Let $|\alpha\rangle = \sum_{j=0}^{\infty}\sum_{m=1}^{d} \alpha_{jm}|j\rangle|m\rangle 
\in \ell^2(\mathbb{N})\otimes \mathbb{C}^d$ be any vector. Define
$|\psi_{\alpha}\rangle = \sum_{j=0}^{\infty}\sum_{m=1}^{d} \alpha_{jm}(T^\dagger)^j|\psi_m\rangle$. Then $|\psi_{\alpha}\rangle =0$ if and only if $\alpha_{jm}=0$ for all $j,m$.
\end{coro}
\begin{proof}
Note that $|\psi_{\alpha}\rangle = \widetilde{A}_{-}|\alpha \rangle$. Because $\widetilde{A}_{-}$ is invertible, $|\psi_{\alpha}\rangle=0$ if and only if 
$|\alpha\rangle = 0$.
\end{proof}

\begin{thm}
\label{thm:wh} 
The Fredholm BBT operator
$\widetilde{A}_0\in\mathcal{A}$ can be factorized as 
\(
\widetilde{A}_0 = \widetilde{A}_+\widetilde{D}_0\widetilde{A}_-,
\)
with $\widetilde{A}_-$ ($\widetilde{A}_+$) an invertible lower (upper) block-triangular BBT operator and 
\(
\widetilde{D}_0 = \sum_{m=1}^{d} T^{\ind_m} |m\rangle\langle m|.
\)
\end{thm}

\begin{proof}

Let $\widetilde{A}_+$ be the operator defined on ${\cal B}_C$ by its action
as follows: for every $j\in \mathbb{N}$ and $m=1,\dots,d$,
\begin{eqnarray}
\label{Aminus}
\widetilde{A}_+: |j\rangle |m\rangle \mapsto \widetilde{A}_0(T^\dagger)^{j+\ind_m} |\psi_{m}\rangle. \nonumber
\end{eqnarray}
The proof of block-Toeplitz property is very similar to the earlier case,
\begin{eqnarray*}
T\widetilde{A}_+T^\dagger |j\rangle |m\rangle  = T\widetilde{A}_+ |j+1\rangle |m\rangle
= T \widetilde{A}_0(T^\dagger)^{j+\ind_m+1} |\psi_{m}\rangle = 
\widetilde{A}_0(T^\dagger)^{j+\ind_m} |\psi_{m}\rangle = \widetilde{A}_+|j\rangle |m\rangle.
\end{eqnarray*}
In the second to last step, we used $T\widetilde{A}_0T^\dagger = \widetilde{A}_0$,
which is true because $\widetilde{A}_0$ is block-Toeplitz.
To prove that $\widetilde{A}_+$ is upper-triangular, we will show that $[\widetilde{A}_+,T]=0$.
For $j\ge 1$, we have
\begin{eqnarray*}
 \widetilde{A}_+ T |j\rangle|m\rangle = \widetilde{A}_+ |j-1\rangle |m\rangle  
 = \widetilde{A}_0(T^\dagger)^{j+\ind_m-1}|\psi_m\rangle 
= T\widetilde{A}_0(T^\dagger)^{j+\ind_m}|\psi_m\rangle
= T\widetilde{A}_+  |j\rangle|m\rangle.
\end{eqnarray*}
For $j=0$, the left hand side yields $\widetilde{A}_+ T |j=0\rangle|m\rangle =0$, whereas right hand side
leads to
\[
T\widetilde{A}_+  |j = 0\rangle|m\rangle = T\widetilde{A}_0(T^\dagger)^{\ind_m}|\psi_m\rangle
= \widetilde{A}_0(T^\dagger)^{\ind_m-1}|\psi_m\rangle = 0.
\]
For bandedness, it suffices to prove that $[\widetilde{A}_+ (T^\dagger)^{q-p},T^\dagger] = 0$. We have
\[
\widetilde{A}_+ (T^\dagger)^{q-p}T^\dagger|j\rangle|m\rangle = 
\widetilde{A}_+ |j +q-p+1\rangle|m\rangle = \widetilde{A}_0(T^\dagger)^{j+q-p+1+\ind_m} |\psi_{m}\rangle.
\]
For the other term of the commutator,
\[
T^\dagger\widetilde{A}_+ (T^\dagger)^{q-p}|j\rangle|m\rangle =
T^\dagger\widetilde{A}_+ |j+q-p\rangle|m\rangle =
T^\dagger\widetilde{A}_0 (T^\dagger)^{j+q-p+\ind_m} |\psi_m\rangle.
\]
Since $\ind_m \ge 0$ and $\widetilde{A}_0 (T^\dagger)^{q-p}$ is a function of $T^\dagger$ only,
we can push $T^\dagger$ to the right to obtain
\[
T^\dagger\widetilde{A}_+ (T^\dagger)^{q-p}|j\rangle|m\rangle =
\widetilde{A}_0(T^\dagger)^{j+q-p+1+\ind_m} |\psi_{m}\rangle,
\]
as desired. Having proved that $\widetilde{A}_+$ is banded block-Toeplitz, it follows that it can be uniquely extended to $\ell_2(\mathbb{N})\otimes \mathbb{C}_d$.

We now prove the invertibility of $\widetilde{A}_+$.
It suffices to show that $\widetilde{A}_+$ is Fredholm 
and $\Ker\,\widetilde{A}_+$, $\Ker\,\widetilde{A}_+^\dagger$ are trivial. 
Note that 
\begin{equation}
\widetilde{A}_+ = \widetilde{A}_0\widetilde{A}_-^{-1}
\left(\sum_{m}(T^\dagger)^{\ind_m}|m\rangle\langle m|\right).
\end{equation}
Since each of the three factors on the right hand-side are Fredholm, $\widetilde{A}_+$ is also Fredholm. 
For any $|\alpha\rangle = \alpha_{jm}|j\rangle|m\rangle \in 
\ell_2(\mathbb{N})\otimes\mathbb{C}^d$, we have 
\[
\widetilde{A}_+ \sum_{j,m}\alpha_{jm}|j\rangle|m\rangle = 
\widetilde{A}_0 (\sum_{j,m} \alpha_{jm}(T^\dagger)^{j+\ind_m}|\psi_m\rangle).
\]
Since the terms in the bracket on the right hand side
are linearly independent as in Corollary \ref{coro:linind}
and do not belong to the kernel of $\widetilde{A}_0$, therefore 
$\widetilde{A}_+$ has trivial kernel. 
We now show that $\widetilde{A}_0^\dagger$ has trivial kernel.
We use the additive property of Fredholm index, we have
\begin{equation}
\text{index}(\widetilde{A}_+) = 
\text{index}(\widetilde{A}_0) + \text{index}(\widetilde{A}_-^{-1}) 
+ \text{index}\left(\sum_{m}(T^\dagger)^{\ind_m}|m\rangle\langle m|\right).
\end{equation}
Now, $\text{index}(\widetilde{A}_0) = \dim\ker\,\widetilde{A_0} - 
\dim\ker\,\widetilde{A_0}^\dagger = \sum_m \ind_m$, using 
Lemma \ref{lem:kernel}. We also have $\text{index}(\widetilde{A}_-^{-1}) = 0$,
due to the invertibility of $\widetilde{A}_-$, whereas 
$\text{index}\left(\sum_{m}(T^\dagger)^{\ind_m}|m\rangle\langle m|\right)
= -\sum_m \ind_m$ follows from the elementary properties of the shift operator.
Therefore, $\text{index}(\widetilde{A}_+) = 0$. Since $\dim\ker\,\widetilde{A}_+ =0$,
we have $\dim\ker\,\widetilde{A}_+^\dagger = \dim\ker\,\widetilde{A}_+
-\text{index}(\widetilde{A}_+) =0$.

Having proved that $\widetilde{A}_+$ and $\widetilde{A}_-$ satisfy the required properties,
the factorization of $\widetilde{A}_0$ is verified easily by noting the action of $\widetilde{D}$,
\[
\widetilde{D}: |j\rangle |m\rangle \mapsto \left\{ \begin{array}{lcl}
|j-\ind_m\rangle|m\rangle & \text{if} & j\ge \ind_m\\
0 & \text{if} & j<\ind_m \end{array}\right..
\]
\end{proof}

Let  $A_+(z^{-1}), D_0(z,z^{-1})$ and $A_-(z)$ stand for the symbols of \(\widetilde{A}_0\),
$\widetilde{A}_+, \widetilde{D}$ and $\widetilde{A}_-$, respectively.

\begin{coro}
\label{coro:wh} 
\label{coro:wh2} 
\
\begin{enumerate}
\item
The symbol $A_0(z,z^{-1})$
of the Fredholm BBT operator \(\widetilde{A}\) is invertible 
 on the unit circle.  
\item
$A_0(z,z^{-1}) = A_+(z) D_0(z,z^{-1}) A_-(z^{-1})$ is a WH  factorization.
\end{enumerate}
\end{coro}

\begin{proof}
The factorization of $A_0$ follows from the multiplicative property of the symbol, therefore we only need to
show that the three factors satisfy the requirements in Definition\,\ref{def:wh}.
Because $\widetilde{A}_+$ is upper-triangular banded block-Toeplitz operator, we have $A_+ \in \mathbb{C}_d[z]$.
$\widetilde{A}_+$ is a polynomial in $T$, whose spectrum coincides with the unit disk. Since $\widetilde{A}_+$
is invertible, $A_+(z)$ is invertible on the unit disk. Similarly, since $\widetilde{A}_-$ is lower-triangular 
banded block-Toeplitz operator, we have $A_- \in \mathbb{C}_d[z^{-1}]$. $\widetilde{A}_-$ is a function
of $T^\dagger$, whose spectrum also coincides with the unit disk. Therefore, $A_-$ is invertible, while
$z^{-1}$ takes values on the unit disk. By definition,
\(
D(z) = \sum_{m=1}^{d} z^{\ind_m} |m\rangle\langle m|.\)
\end{proof}

\section{Proof of the symmetric Wiener-Hopf factorizations}
\label{app:symwhproofs}

In this appendix we collect the proofs of the Theorems \ref{thm:symwh1} - \ref{thm:symwh5}. In the following, $\mathcal{K}$ denotes the set
of symmetric partial indices of $A$ as defined in Sec.~\ref{sec:symwh}.
\smallskip

\noindent
\textbf{Theorem 3.1} (Symmetric factorization over \(\mathbb{R}\) and \(\mathbb{H}\))\textbf{.}
\textit{
Let $A \in\F_d[z,z^{-1}]$ denote a matrix Laurent polynomial with coefficients from
\(\mathbb{F}_d\), $\F=\R$ or \(\H\). If \(A(z)\), regarded as a complex
matrix Laurent polynomial, is invertible on the unit circle, then there exists a factorization 
$A = A_+DA_-$ such that 
\[A_+ \in \F_d[z], \quad
 A_- \in \F_d[z^{-1}], \quad D(\mathcal{K}) = \sum_{m=1}^{d}z^{\ind_m}|m\rangle\langle m|\bm{1}_\F,
\]
in terms of $\{\ind_1,\dots,\ind_d\} = \mathcal{K}$ and $A_+(z)$ ($A_-(z)$) is invertible inside (outside) the unit circle. 
}

\begin{proof}
It suffices to prove that an  
$A_- \in \mathbb{F}_d[z^{-1}]$
constructed following the construction in the proof of Prop.~\ref{propAminus}. 
If this is true, then $A_+ = AA_-^{-1}D^{-1} \in \mathbb{F}_d[z]$ 
(Note that $D$ is invertible on the unit circle).
For the case with real entries, the vectors $\{|\psi_{\ind,s}\rangle,\ s=1,\dots,s_{\ind}\}$
can be chosen to have real entries, since they belong to the kernel of $\widetilde{A}_{\ind-1}$,
which is a real block-Toeplitz operator. This choice of $\{|\psi_{\ind,s}\rangle\}$ ensures that
$\widetilde{A}_-$ in Eq.\,\eqref{Aminus} is real, therefore $A_- \in \mathbb{R}_d[z^{-1}]$.
For the case quaternionic entries, the kernel of quaternionic
block-Toeplitz operators always appear in Kramers' pairs. If
$|\psi\rangle$ is a kernel vector of some quaternionic block-Toeplitz operator,
then so is $i\sigma_y\cc|\psi\rangle$. Following Corollary \ref{coro:inc},
we find that $s_\ind$ is always even for any integer $\ind$.
We may choose a basis 
$\{|\psi_{\ind,s}\rangle,\ s=1,\dots,s_{\ind}\}$ of $\mathcal{V}_\ind$ 
for each $\ind\in\mathbb{Z}$ such that $|\psi_{\ind,2s}\rangle = i\sigma_y\cc|\psi_{\ind,2s-1}\rangle$. 
It is easy to see that $\widetilde{A}_-$ constructed using Eq.\,\eqref{Aminus} has quaternionic entries, so that 
$A_- \in \mathbb{H}_{d}[z^{-1}]$.
\end{proof}

The proof of Theorems \ref{thm:symwh2}-\ref{thm:symwh5} rely on Theorem \ref{lem:uniquepi} and the following Lemma.
The \textit{signature} \(s\) of a Hermitian, invertible, real or complex matrix is the difference between
the number of its positive and its negative eigenvalues. For quaternionic matrices, we define the signature 
to be half of this quantity.

\begin{lem}[Symmetric factorizations of constant matrix Laurent polynomials]
\label{lem:constantfact}
Let \(C\in\F_d\) denote an invertible matrix with \(\F=\R, \C\) or \(\H\).
\begin{enumerate}
\item 
If $C \in \F_d$ is Hermitian ($C=C^\dagger$) with signature
signature $s$, then there exists  $G \in \F_d$ invertible such that
\(
C = G\left(\sum_{m=1}^{d}s_m|m\rangle\langle m|\bm{1}_\F \right)G^\dagger,
\)
$s_m=\pm 1$, $s_m>s_{m'}$ if and only if $m>m'$.

\item 
Let \(d\) be even and \(\F=\mathbb{R}\) or \(\C\). If $C \in \F_d$ is skew-symmetric ($C= - C^{\rm T}$),
then there exists $G \in \F_d$ invertible, such that 
\(
C = G\left(\sum_{m=1}^{d/2}\big(|\bar{m}\rangle\langle m| - 
|m\rangle\langle \bar{m}|\big)\right)G^{\rm T}.
\)

\item 
If $C \in \C_d$ is symmetric ($C=C^{\rm T}$), then there exists 
$G \in \C_d$ invertible, such that 
\(
C = G\left(\sum_{m=1}^{d}|m\rangle\langle m|\right)G^{\rm T}.
\)

\item 
If $C \in \H_d$ is skew-Hermitian ($C=-C^\dagger$), then there exists $G \in \H_d$ invertible, 
such that 
\(
C = G\left(\sum_{m=1}^{d}|m\rangle\langle m|\bm{k}\right)G^\dagger.
\)

\end{enumerate}
\end{lem}

\begin{proof}
1. This factorization follows from the orthogonal or unitary 
diagonalization of $C = UDU^\dagger$ which exists for $\F \in \{\R,\C,\H\}$ \cite{Horn,Loring12}.

2. By the Youla decomposition \cite{Youla}, every real (complex) skew-symmetric matrix $C$ can be decomposed as $C = U\Sigma U^{\rm T}$ where $U$ is a real orthogonal (complex unitary) matrix and $\Sigma$ is a block-diagonal matrix with $2\times 2$ blocks of the form $\begin{bmatrix}0 & \lambda \\ -\lambda & 0\end{bmatrix}$ with $\lambda \ge 0$. The claim follows.

3. By the Autonne-Takagi factorization \cite{Horn},  
every complex symmetric matrix $C$ can be factorized as $C = U\Sigma U^{\rm T}$, where $U$ is unitary
and $\Sigma$ is a diagonal matrix of singular values of $C$. 
Since $C$ is invertible in our case, all singular values are strictly positive and the claim follows. 

4. This result follows from diagonalization of 
quaternionic matrices \cite{Loring12}: every quaternionic skew-Hermitian matrix $C$ can be factorized as $C = UDU^\dagger$, where $U$ is quaternionic 
unitary and $D$ is block-diagonal with $2\times 2$ blocks of the form 
$\begin{bmatrix}i\lambda & 0 \\ 0 & -i\lambda \end{bmatrix}$. To further
factorize $D$, we use the identities
\[
\begin{bmatrix}i\lambda & 0 \\ 0 &-i\lambda \end{bmatrix} = 
\left\{ \begin{array}{ccl}
\begin{bmatrix}0 & i\sqrt{|\lambda|} \\ i\sqrt{|\lambda|} & 0  \end{bmatrix}
\begin{bmatrix}i & 0 \\ 0 &-i \end{bmatrix}
\begin{bmatrix}0 & -i\sqrt{|\lambda|} \\ -i\sqrt{|\lambda|} & 0  \end{bmatrix}
& \text{if} & \lambda < 0 \\[12pt]
\begin{bmatrix}\sqrt{|\lambda|} & 0 \\ 0 & \sqrt{|\lambda|}  \end{bmatrix}
\begin{bmatrix}i & 0 \\ 0 &-i \end{bmatrix}
\begin{bmatrix}\sqrt{|\lambda|} & 0 \\ 0 & \sqrt{|\lambda|}  \end{bmatrix}
& \text{if} & \lambda > 0 \end{array}\right..
\]

\end{proof}
\smallskip

\noindent\textbf{Theorem 3.2} (A  factorization of Hermitian matrix Laurent polynomials)\textbf{.}
{\it 
Let $A=A^\dagger \in\F_d[z,z^{-1}]$ be invertible on the unit circle for \(\F=\R\), \(\C\), or \(\H\).
There exists a factorization 
$A = A_+ D_1 
A_+^\dagger$ such that $A_+ \in \F_d[z]$
is invertible for \(|z|\leq 1\)
and
\begin{align*}
D_1
=\sum_{m \in \mathcal{M}_+}(z^{\ind_m} |\bar{m}\rangle\langle{m} | 
+z^{-\ind_m}|m\rangle\langle \bar{m}|)\bm{1}_\F+ \sum_{m \in \mathcal{M}_0}s_m|m\rangle\langle m|\bm{1}_\F
\end{align*}
in terms of $\{\ind_1,\dots,\ind_d\} = \mathcal{K}$ and a set of signs $\{s_m=\pm 1\}_{m\in{\cal M}_0}$ such that  $s_m>s_{m'}$ for $m>m'$. The signs add up to the signature of \(A\), that is, \(\sum_{m\in\mathcal{M}_0}s_m = s$. 
}

\begin{proof}
1. Let $A = A_+ D A_-$ be a factorization of $A$ as guaranteed by
Theorem \ref{thm:symwh1} if \(\F=\R\) or \(\H\) or simply the classical WH
factorization if \(\F=\C\). The condition \(A  = A^\dagger\) implies that
\(A_+ D A_- = A_-^\dagger D^\dagger A_+^\dagger\).
Moreover, the $\dagger$ operation flips the sign of all the partial indices.
Since partial indices are unique, we conclude that they appear in pairs $(\ind,-\ind)$ 
with equal multiplicities ($s_{{-\ind}} = s_\ind$). By defining the permutation operator 
$X \equiv \sum_{m=1}^{d} |m\rangle \langle \bar{m}|\bm{1}_\F$, \(\bar{m}\equiv d+1-m\),
we can write $D^\dagger = XDX$. It is easy to check that $X$ satisfies $X = X^\dagger = X^{-1}$. 
Then, since $A = A_-^\dagger XDX A_+^\dagger$, we conclude from Lemma 
\ref{lem:uniquepi}
\[
A_-^\dagger X = A_+ B,
\]
for a matrix polynomial $B$ with entries $B_{mm'} \in \F[z]$ of degree $\ind_m- \ind_{m'}$.
Now, $A=A_+BDXA_+^\dagger$, and since $A_+$ is invertible, we get $(CDX)^\dagger = BDX$.
By comparing entries, we get 
\begin{equation}
\label{Centries}
B_{mm'}^* =z^{(\ind_{m'}-\ind_m)}B_{{-\ind_{m'}}{-\ind_m}} \quad 
\forall \ind_m,\ind_{m'} \in \mathcal{K}.
\end{equation}
Also, since 
$\sum_{m,m' \in \mathcal{M}_0}\langle m |BDX|m'\rangle|m\rangle\langle m'|$ is Hermitian, we can use Lemma \ref{lem:constantfact} to express it as 
\[
\sum_{m,m' \in \mathcal{M}_0}\langle m |BDX|m'\rangle|m\rangle\langle m'| = 
G\left(\sum_{m \in \mathcal{M}_0}s_m|m\rangle\langle m|\bm{1}_\F\right)G^\dagger,
\] 
with diagonal entries $\{s_m\}$ satisfying the requirements stated in the theorem, and $G \in \F_{s_0}$ invertible. Let us define a matrix polynomial 
\[
B_+ \equiv G+ \sum_{m \in \mathcal{M}_+}|m\rangle\langle \bar{m}|
\sum_{m\in \mathcal{M}_-} \Big( \sum_{|\ind_{m'}|<{-\ind_m}}|{m'}\rangle
\langle\bar{m}| B_{{m'}m} +
|m \rangle\langle \bar{m}| B_{mm} +
 \frac{1}{2}|\bar{m}\rangle\langle \bar{m} |B_{\bar{m},m}  \Big),
\]
so that
\begin{multline*}
B_+^\dagger = G^\dagger + \sum_{m \in \mathcal{M}_-}|{m}\rangle\langle \bar{m}| +
\sum_{m\in \mathcal{M}_+} \Big( \sum_{|\ind_{m''}|<\ind_m}
|m\rangle\langle m'' | B_{m'',\bar{m}}^* +
|m \rangle\langle \bar{m}| B_{\bar{m},\bar{m}}^*
+\frac{1}{2}|{m}\rangle\langle m |B_{m,\bar{m}}^*  \Big).
\end{multline*}
\begin{multline*}
 \qquad = G^\dagger +  \sum_{m \in \mathcal{M}_-}|m\rangle\langle \bar{m}|+ 
 \sum_{m\in \mathcal{M}_+} \Big( \sum_{|\ind_{m''}|<\ind_m}
 |m\rangle\langle m'' | z^{-\ind_m-\ind_{m''}}B_{m,-m''} +
|m \rangle\langle \bar{m}| B_{mm}
\\ +\frac{1}{2}|m\rangle\langle m |z^{-2\ind_m}B_{m,\bar{m}}  \Big).
\end{multline*}
From this point on, a straightforward if long multiplication yields 
\begin{equation}
\label{eqsym}
BDX = B_+D_1(\mathcal{K},s)B_+^\dagger,
\end{equation}
and therefore 
\(
A = A_+B_+D_1(\mathcal{K},s)(A_+B_+)^\dagger.
\)
It remains to prove that $B$ is invertible inside the unit circle. To see that this is indeed the case, we multiply Eq.\,\eqref{eqsym}
by $z^{-k}$, where $k$ is an integer less than the smallest partial index of $A$. Both the left and right hand-sides of Eq.\,\eqref{eqsym} are well-defined and invertible inside the unit circle (including $z=0$). Therefore, $B$ must be invertible
inside the unit circle.
\end{proof}

\noindent\textbf{Theorem 3.3} (A factorization of antisymmetric matrix Laurent polynomials)\textbf{.}
{\it 
Let 
 $A=-A^{\rm T}\in\mathbb{F}_d[z,z^{-1}]$ be invertible on the unit circle for \(\F=\R\) or \(\F=\C\). There exists a
factorization $A = A_+ D_3 
A_+^{\rm T}$ such that $A_+ \in \F_d[z]$
is invertible for \(|z|\leq1\) and
\begin{align*}
&D_3
=-\sum_{m\in \mathcal{M}_+} \big(z^{\ind_m} |\bar{m}\rangle\langle{m} |\;  
- z^{-\ind_m}|m\rangle\langle{\bar{m}} |\big) 
 +  \sum_{\substack{m\in \mathcal{M}_0\\[1pt] m\le d/2}} \big( |\bar{m}\rangle\langle{m} |\;  
- |m\rangle\langle{\bar{m}} |\big),
\end{align*}
in terms of $\{\ind_1,\dots,\ind_d\} = \mathcal{K}$.
}

\begin{proof}
The proof is similar to that of Theorem \ref{thm:symwh2}.
We replace all $\dagger$ by ${\rm T}$ (transpose) and remove
complex conjugation $*$ in every place they appear in the previous proof.
The main difference is the factorization of $BDX$ for $\mathcal{M}_0$ block, 
which in this case is complex
antisymmetric. We use the factorization 2. from Lemma \ref{lem:constantfact},
which leads to $B_{m m'}^\star =
-z^{(\ind_{m'}-\ind_m)}B_{\bar{m'},\bar{m}}$. The rest of the proof is identical to that of 
Theorem \ref{thm:symwh2}.
\end{proof}

%%%%%%%%%%%%%%%%%%%%%%%%%%%%%%%%%%%%%%%%%%%%%%%%%%%%

\noindent\textbf{Theorem 3.4} (A factorization of complex symmetric matrix Laurent polynomials)\textbf{.}
{\it Let $A=A^{\rm T} \in\C_d[z,z^{-1}]$ be invertible on the unit circle. There exists a
factorization $A = A_+ D_2 
A_+^{\rm T}$ such that $A_+ \in \mathbb{C}_d[z]$
is invertible for \(|z|\leq 1\)
and 
\begin{align*}
D_2 
= 
\sum_{m \in \mathcal{M}_+}\big(z^{\ind_m}|\bar{m}\rangle\langle{m} |  
+ z^{-\ind_m}|m\rangle\langle{\bar{m}} |\big) 
+ \sum_{m \in \mathcal{M}_0}|m\rangle\langle m| ,
\end{align*}
in terms of $\{\ind_1,\dots,\ind_d\} = \mathcal{K}$. 
}

\begin{proof}
The proof is again similar to that of Theorem \ref{thm:symwh2}.
We replace all $\dagger$ by ${\rm T}$ and remove
complex conjugation $*$ in every place they appear in the previous proof.
The factorization 3. from Lemma \ref{lem:constantfact} of $BDX$
in this case leads to is $B_{m m'}^\star =
z^{(\ind_{m'}+\ind_m)}B_{\bar{m'},\bar{m}}$. 
\end{proof}

%%%%%%%%%%%%%%%%%%%%%%%%%%%%%%%%%%%%%%%%%%%%%%%%%%%%

\noindent\textbf{Theorem 3.5} (A factorization of quaternionic anti-Hermitian matrix Laurent Polynomials)\textbf{.}
{\it 
Let $A=-A^\dagger \in\H_d[z,z^{-1}]$ be invertible on the unit circle.
There exists a
factorization $A = A_+ D_4
A_+^\dagger$ such that $A_+ \in \H_d[z]$
is invertible for \(|z|\leq 1\) and
\begin{align*}
D_4 
= -\sum_{m\in \mathcal{M}_+} \big(z^{\ind_m} |\bar{m}\rangle\langle{m} |\;  
- z^{-\ind_m}|m\rangle\langle{\bar{m}} |\big)\bm{1}_\H 
+  \sum_{m\in \mathcal{M}_0} \big( |m\rangle\langle m |\big)
\bm{k},
\end{align*}
in terms of $\{\ind_1,\dots,\ind_d\} = \mathcal{K}$.
}

\begin{proof}
The proof is again very similar to those of the three theorems above. The factorization 4. from Lemma \ref{lem:constantfact} of $BDX$ in this case leads to $B_{m m'}^\star =
z^{(\ind_{m'}+\ind_m)}B_{\bar{m'},\bar{m}}$. 
\end{proof}

%%%%%%%%%%%%%%%%%%% END TEXT 

\bibliography{References}

\end{document}